\documentclass[11pt,a4paper]{article}
\pdfoutput=1                                                          
\usepackage{jheppub}   
\usepackage{ulem}
\usepackage{amsthm}
\usepackage{amsmath,amssymb,epsfig,amsfonts}
\usepackage{graphicx}
\usepackage[usenames,dvipsnames]{xcolor}
\usepackage{setspace}
\usepackage{lscape}
\usepackage{bm}
\usepackage{picture}
\usepackage{tikz}
\usepackage{hyperref}
\usepackage[all,cmtip]{xy}
\usepackage{multirow}
\usepackage{array}
\usepackage{epstopdf}

\usepackage{float}

\usepackage[english]{babel} 

\newcommand{\beq}{\begin{equation}}
\newcommand{\eeq}{\end{equation}} 
\newcommand{\bea}{\begin{eqnarray}}
\newcommand{\eea}{\end{eqnarray}}

\makeatletter

\newtheorem{theorem}{Theorem}
\newtheorem{lemma}{Lemma}
\newtheorem{corollary}[theorem]{Corollary}
\newtheorem{definition}{Definition}
\newtheorem{remark}{Remark}

\newcommand{\bZ}{\mathbb{Z}}

\newcommand{\bP}{\mathbb{P}}
\newcommand{\bR}{\mathbb{R}}
\newcommand{\bF}{\mathbb{F}}

\newcommand{\cN}{\mathcal{N}}

\def\unit{{1\kern-.65ex {\rm l}}}
\def\1{{1\kern-.65ex {\rm l}}}

\newcount\hour \newcount\minute
\hour=\time \divide \hour by 60
\minute=\time
\def\now{%
\ifnum \hour<13
  \ifnum \hour=0 \advance \hour by 12 \number\hour:\else \number\hour:\fi%
     \ifnum \minute<10 0\fi%
     \number\minute%
\ A.M.%
\else \advance \hour by -12 \number\hour:%
  \ifnum \minute<10 0\fi%
  \number\minute%
  \ P.M.%
\fi%
}

\makeatother

\newtheorem{proposition}{Proposition}

\title{On finiteness of Type IIB compactifications:
Magnetized branes on elliptic Calabi-Yau threefolds}

\author[1,2]{Mirjam Cveti\v c,}
\author[3]{\hspace{.2cm}James Halverson,}
\author[1]{\hspace{.2cm}Denis Klevers,}
\author[1]{\hspace{.2cm}Peng Song} 
\affiliation{$^1$ Department of Physics and Astronomy, \\University of Pennsylvania,
  Philadelphia, PA 19104-6396, USA \vspace{.25cm} }
\affiliation{$^2$ Center for Applied Mathematics and Theoretical Physics,\\
University of Maribor, Maribor, Slovenia \vspace{.25cm} }
 \affiliation{$^3$ Kavli Institute for Theoretical Physics, \\ University of California,
  Santa Barbara, CA 93106-4030, USA \vspace{.25cm} }

\emailAdd{cvetic@cvetic.hep.upenn.edu}
\emailAdd{jim@kitp.ucsb.edu}
\emailAdd{klevers@sas.upenn.edu}
\emailAdd{songpeng@sas.upenn.edu}  

\preprint{UPR-1259-T, 
NSF-KITP-13-259}

\abstract{The string landscape satisfies interesting finiteness
  properties imposed by supersymmetry and string-theoretical
  consistency conditions. We study $\cN=1$ supersymmetric
  compactifications of Type IIB string theory on smooth elliptically
  fibered Calabi-Yau threefolds at large volume with magnetized
  D9-branes and D5-branes. We prove that supersymmetry and tadpole
  cancellation conditions imply that there is a finite number of such
  configurations.  In particular, we derive an explicitly computable
  bound on the number of magnetic flux quanta, as well as the number
  of D5-branes, which is independent of the continuous moduli of the
  setup.  The proof applies if a number of easy to check geometric
  conditions of the twofold base are met.  We show that these
  geometric conditions are satisfied for the almost Fano twofold
  bases given by each toric variety associated to a reflexive
  two-dimensional polytope as well as by the generic del Pezzo
  surfaces $dP_n$ with $n=0,\ldots,8$.    Physically, this finiteness proof shows that there exist a
  finite collection of four-dimensional gauge groups and chiral matter
  spectra in the 4D supergravity theories realized by these
  compactifications. As a by-product we explicitly
  construct all generators of the K\"ahler cones of 
  $dP_n$ and work out their relation to representation theory.}
\begin{document}

\maketitle

\section{Introduction}

M-theory or superstring compactification to four dimensions remains
the most promising framework for the unification of the fundamental
forces in Nature. The set of associated low energy effective
theories  which can arise in consistent compactifications
is known as the string landscape. There have been many efforts to
quantify this space, with the hope of uncovering observable properties
shared by large classes of vacua which lead to novel insights in
particle physics or cosmology. However, this has proven to be a very
difficult problem deserving a multi-faceted approach. 

The traditional one is to study the effective scalar potential on
moduli space and to examine its associated vacua; in general a variety
of perturbative and non-perturbative effects are utilized to this end.
For example, in the much studied moduli stabilization scenarios of
Type IIB compactifications
\cite{Kachru:2003aw,Balasubramanian:2005zx}, these effects include
superpotential contributions from background Ramond-Ramond flux and
D-instanton effects. Increasingly more detailed realizations of these
constructions have been studied in recent years; for progress
on vacua with
explicit complex structure moduli stabilization, see \cite{Louis:2012nb,MartinezPedrera:2012rs}, and on constructing
explicit de Sitter flux vacua with a chiral spectrum, see the recent
\cite{Cicoli:2013cha}.  While this progress is significant and
provides excellent proofs of principle, a clear caveat to the explicit
construction of vacua is the enormity of the landscape.

Another approach is to study properties of the landscape more
broadly. In Type IIB flux compactifications this has included, for
example, the importance of four-form fluxes in obtaining the observed
value of the cosmological constant \cite{Bousso:2000xa}; issues of
computational complexity, including finding vacua in agreement with
cosmological data \cite{Denef:2006ad} and the systematic computation
of non-perturbative effective potentials \cite{Cvetic:2010ky}; and the
distribution and number of various types of supersymmetric and
non-supersymmetric vacua \cite{Denef:2004ze,Denef:2004cf}. 
Progress has also been made in understanding vacua in strongly coupled 
corners of the landscape. 
For example there has been much progress in F-theory, beginning with
\cite{Donagi:2008ca,Beasley:2008dc}.

A final approach, which will be the one utilized in this paper, is to
understand how consistency conditions and properties of the landscape
differ from those of generic quantum field theories. The former case
is motivated in part by the existence of a swampland
\cite{Vafa:2005ui} of consistent effective theories which do not admit
a string embedding.  There are a number of examples of limitations on
gauge theories in the landscape not present in generic gauge theories.
In weakly coupled theories with D-branes, Ramond-Ramond tadpole
cancellation places stronger constraints
\cite{Uranga:2000xp,Aldazabal:2000dg,Ibanez:2001nd,Cvetic:2001nr} on low energy
gauge theories than anomaly cancellation, which include additional
anomaly nucleation constraints \cite{Halverson:2013ska} on SU$(2)$
gauge theories; see also \cite{Cvetic:2012xn} for a recent analysis of anomalies at strong coupling
in F-theory; ranks of gauge groups are often bounded \cite{Lerche:1986cx,Gmeiner:2005vz}; and the matter representations are 
limited by the fact that open strings have precisely two ends. While more
matter representations are possible at strong coupling, the
possibilities are still limited. For example, in F-theory the possible
non-Abelian \cite{Katz:1996xe, DeWolfe:1998zf,
Grassi:2000we,Grassi:2011hq,Morrison:2011mb,Grassi:2013kha,Hayashi:2014kca,Grassi:2014sda,Esole:2014bka} and Abelian
\cite{Morrison:2012ei,Borchmann:2013jwa,Cvetic:2013nia,Grimm:2013oga,Braun:2013nqa,Cvetic:2013uta,Borchmann:2013hta,Cvetic:2013jta,Cvetic:2013qsa}
matter representations are limited by the structure of holomorphic
curves in the geometry.

In \cite{Douglas:2006xy}, Douglas and Taylor studied the landscape of
intersecting brane models\footnote{See
  \cite{Blumenhagen:2005mu,Blumenhagen:2006ci,Cvetic:2011vz} for
  reviews of these compactifications and their implications for
  particle physics.} in Type IIA compactifications on a particular
$\mathbb{Z}_2\times\mathbb{Z}_2$-orientifold\footnote{See \cite{Blumenhagen:2004xx,Gmeiner:2005vz}
for a  finiteness proof of the number of supersymmetric D-branes 
for fixed complex structures of this orientifold and  
\cite{Cvetic:2001tj,Cvetic:2001nr}
for a first construction of chiral $\mathcal{N}=1$ supersymmetric three-family models.}. They found that the conditions on D6-branes
necessary for $\cN=1$ supersymmetry in four dimensions, together with
the D6-brane tadpole cancellation condition required for consistency
of the theory, allow only a finite number of such D6-brane
configurations\footnote{See \cite{Cvetic:2004ui} for a counting of three 
family vacua, that yields eleven such vacua.}. In each configuration, the 
four-dimensional gauge
group and matter spectrum can be determined explicitly. Thus, the
finite number of D6-brane configurations gives a finite number of
gauge sectors in a 4D supergravity theory that arise from these 
compactifications, and their
statistics were studied explicitly. It is expected that the finiteness
result which they obtained is a much more general consequence of
supersymmetry and tadpole cancellation conditions, rather than a
phenomenon specific to their construction. In fact, they proposed a
potential generalization of their result to theories with magnetized
D9- and D5-branes on smooth elliptically fibered Calabi-Yau
threefolds, which can also be motivated by mirror symmetry, for 
example.

In this paper, we demonstrate that finiteness results are indeed   much
more general phenomena, providing further evidence that the string
landscape itself is finite.
 Specifically, in large
volume Type IIB compactifications on many smooth elliptically fibered
Calabi-Yau threefolds $\pi:\,X\rightarrow B$, we prove that there are
finitely many configurations of magnetized D9- and
D5-branes satisfying Ramond-Ramond tadpole cancellation and the
conditions necessary for $\cN=1$ supersymmetry in four dimensions.  We
formulate a general, mathematical proof showing the existence of computable, explicit bounds on the number of magnetic flux quanta on
the D9-branes and on the number of D5-branes, which only depends on
the topology of the manifold $B$ and is in particular independent of
its K\"ahler moduli, as long as they are in the large volume
regime of $X$. These bounds involve simple geometric quantities of the
twofold base $B$ of $X$ and the proof applies to any base $B$ that
satisfies certain geometric conditions, that are easy to check and
summarized in this paper.  Furthermore, we show that these conditions
are met by the almost Fano twofold bases $B$ given by the toric
varieties associated to all 16 reflexive two-dimensional polytopes
and the generic del Pezzo surfaces $dP_n$ for $n=0,\ldots, 8$. This
list in particular includes also the Hirzebruch surfaces $\bF_0 =
\bP^1 \times \bP^1$, $\bF_1 = dP_1$, and $\bF_2$. In this work, we
focus on the finiteness question only, leaving the analysis of gauge
group and matter spectra for this finite set of configurations to
future work.

This paper is organized is follows. In section \ref{sec:background} we
provide the relevant background on Type IIB compactifications with
magnetized D9- and D5-branes and elliptically fibered Calabi-Yau
threefolds at large volume. We first discuss the tadpole and
supersymmetry conditions of general such setups, then present a basic
account on elliptically fibered Calabi-Yau threefolds and end with a
detailed discussion of the geometries of the twofold bases
$B=\mathbb{F}_k$, $dP_n$ and the 16 toric twofolds.  In section
\ref{sec:finiteness} we prove the finiteness of such D-brane
configurations. We begin by rewriting the tadpole and supersymmetry
constraints in a useful form for the proof and make some definitions,
then show the power of these definitions by proving finiteness on
$\mathbb{P}^2$. Finally, we prove the existence of explicit bounds on
the number of fluxes and D5-branes, that apply certain geometric
conditions on $B$ are satisfied.  In section \ref{sec:conclusion} we
conclude and discuss possibilities for future work.  In appendix
\ref{app:MoriKaehlerFano} we discuss the detailed structure of the
K\"ahler cone of generic del Pezzo surfaces $dP_n$ and give the
proof of positive semi-definiteness of certain intersection matrices
on these K\"ahler cones, which is essential for the proof. In appendix 
\ref{app:explicitdata} we summarize the geometrical data of the considered 
almost Fano twofolds which is necessary to explicitly compute the bounds
derived in this work.

While finishing this manuscript we learned about the related work 
\cite{AndersonTaylor} in which heterotic compactifications and their F-theory duals are 
constructed systematically.

\clearpage

\section{Magnetized Branes on Elliptically Fibered Calabi-Yau Manifolds}
\label{sec:background}

We consider an $\mathcal{N}=1$ compactification of Type IIB string theory
on a Calabi-Yau threefold $X$ with spacetime-filling D5-branes,
magnetized D9-branes, i.e.~D9-branes with magnetic fluxes\footnote{For the 
generic case of gauge bundles with
non-Abelian structure groups, see \cite{Blumenhagen:2005pm}.}, 
and an O9-plane.  We group the D9-branes into stacks of $N^\alpha$ branes and 
their orientifold image branes. The corresponding line bundle magnetic fluxes in 
$H^{(1,1)}(X,\mathbb{Z})$ are denoted by $F^\alpha$, respectively, $-F^\alpha$ 
for the image brane. In addition, we add stacks of  D5-branes wrapping a curve 
$\Sigma^{\text{D5}}$.

In the following discussion of these models\footnote{These models were  first proposed for model-building in \cite{Bachas:1995ik}.} we mainly follow the notations and conventions of \cite{Douglas:2006xy}, to which 
we also refer for more details. For a concise review see  \cite{Blumenhagen:2006ci}.

\subsection{Tadpole Cancellation and SUSY Conditions}
\label{sec:TapoleSUSY}

D-branes carry Ramond-Ramond charge and source flux lines that must
be cancelled in the compact extra dimensions, in accord with Gauss'
law. These give rise to the so-called tadpole cancellation
conditions. The D5-brane tadpole cancellation conditions are
\begin{equation}\label{Tadpoles}
n_I^{D5}-T_I=\sum_{\alpha}N^{\alpha}\mathcal{K}(F^{\alpha},F^{\alpha},D_I)\,,\qquad  \forall\,D_I\in H^{(1,1)}(X)
\end{equation}
(we note a sign  difference between the D5-tadpoles\footnote{We thank Washington Taylor and Michael Douglas for helpful correspondence related to this 
issue.} in \cite{Douglas:2006xy} and \cite{Blumenhagen:2006ci}; here, we use the sign in \cite{Blumenhagen:2006ci})
where $D_I$ is a basis of divisors on $X$, $\mathcal{K}(\cdot,\cdot,\cdot)$
is the classical triple intersection of three two-forms or their dual
divisors, where we  denote, by abuse of notation, a divisor $D_I$ and its
Poincar\'e dual by the same symbol. Furthermore, we define the curvature 
terms
\beq \label{eq:TI}
	T_I = \int_{D_I}c_2(X)\,,\qquad n_I^{\text{D5}}=\Sigma^{\text{D5}}\cdot D_I\,,
\eeq
where $c_2(X)$ is the second Chern-class on $X$ and $\Sigma^{\text{D5}}$ is the curve wrapped by all D5-branes.
The integral wrapping numbers $n^{\text{D5}}_I$ are positive if $\Sigma^{\text{D5}}$ is a holomorphic curve
and the $D_I$ are effective divisors. Following \cite{Blumenhagen:2006ci}, the D9-brane tadpole cancellation condition reads
\begin{equation}
  \label{eq:D9tadpole}
  16 = \sum_\alpha N^\alpha\,.
\end{equation}

Compactification of Type IIB string theory on a Calabi-Yau manifold
gives rise to a four-dimensional $\cN=2$ supergravity theory. An O9-orientifold
breaks half of these supersymmetries and give rise to an $\cN=1$ supergravity theory at low
energies. Only D9- and D5-branes can be added in a supersymmetric way to this orientifold. 
However, this requires aligning the central charges $Z(F^\alpha)$ of the branes appropriately with the O9-plane. 
For consistency with the supergravity approximation, we have to assume that the K\"ahler parameters of the Calabi-Yau threefold
$X$ are at large volume. In this case, the conditions on the central charges\footnote{In general, the central charge (and also the K\" ahler 
potential on the K\"ahler moduli space) receives perturbative and non-perturbative $\alpha'$
corrections. Recently it has been understood
\cite{Honda:2013uca,Sugishita:2013jca,Hori:2013ika,Halverson:2013qca}
that these corrections are captured by the so-called Gamma class $\hat
\Gamma_X$ on $X$ rather than $\sqrt{Td_X}$. Since we study compactifications at large volume, these corrections can be neglected.} 
necessary for $\cN=1$ supersymmetry, with $J$ denoting the K\"ahler form on $X$, reduce to
\beq
\label{eq:SUSYcondition}
3\mathcal{K}(J,J,F^\alpha)=\mathcal{K}(F^\alpha,F^\alpha,F^\alpha)\,, \qquad
\mathcal{K}(J,J,J)>3\mathcal{K}(J, F^\alpha,F^\alpha)\,,
\eeq
to which we will refer in the following as the SUSY equality and the SUSY inequality respectively. 

\subsection{Smooth Elliptic Calabi-Yau Threefolds}
\label{sec:geometryEllThreefolds}

We study compactifications where $X$ is a smooth elliptically fibered
Calabi-Yau threefold over a complex two-dimensional base $B$,
$\pi:X\rightarrow B$, with a single section $\sigma:B\rightarrow X$,
the zero-section. The class of the section $\sigma$ is the base
$B$. By the adjunction formula and the Calabi-Yau condition, the
section $\sigma$ obeys the relation
\begin{equation} \label{eq:sigma^2}
	\sigma^2=-c_1 \sigma\,,
\end{equation} 
where $c_1$ denotes the first Chern class of the base $B$.
For a smooth threefold the second cohomology is given by $H^{(1,1)}(X)=
\sigma H^{0}(B)\oplus \pi^*H^{(1,1)}(B)$. A basis of $H^{(1,1)}(X)$ generating the
K\"ahler cone of $X$ is given by
\begin{equation} \label{eq:D_I}
	D_I=(D_0,D_i)\,,	\qquad D_0=\sigma+\pi^*c_1\,,\qquad I=0,1,\ldots, p \equiv h^{(1,1)}(B)
\end{equation} 
 with Poincar\' e duality implied when discussing divisors. The divisors $D_i$, $i=1,\ldots, p$, 
are inherited from generators of the K\"ahler cone of the base, by abuse of notation denoted by the same symbol as their counterparts in 
$B$. The divisor $D_0$ is dual to the elliptic fiber $\mathcal{E}$ in the sense that it does not intersect any 
curve in $B$, i.e.~$D_0\cdot \sigma\cdot D_i=0$  by \eqref{eq:sigma^2}, and obeys $D_0\cdot \mathcal{E}=1$. We note that $\mathcal{E}$ is an effective curve.

We emphasize that the requirement of a smooth elliptically fibered $X$, which means 
that the fibration can at most have $I_1$-fibers, 
restricts the choice of two-dimensional bases $B$. The bases we consider here 
are smooth almost Fano twofolds, which are the nine del Pezzo surfaces $dP_n$, 
$n=0,\ldots,8$, that are the blow-ups of $\mathbb{P}^2$ at up to eight generic points, the Hirzebruch surfaces $\mathbb{F}_k$, 
$k=0,1,2$ and the toric surfaces described by the 16 reflexive two-dimensional polytopes.  For these bases, the elliptic fibration $X$ is smooth.

We abbreviate the triple intersections of three divisors on $X$ as $\mathcal{K}_{IJK}= \mathcal{K}(D_I,D_J,D_K)$.  In the particular basis \eqref{eq:D_I}, 
we obtain the following structure of the triple intersections,
\begin{equation} \label{eq:C_IJKrels}
	\mathcal{K}_{ijk}=0\,,\quad \mathcal{K}_{00i}=\sum_j^p b_j\mathcal{K}_{0ij}\,,\quad \mathcal{K}_{000}=\sum_{i,j}^pb_ib_j 
	\mathcal{K}_{0ij}=\sum_i^pb_i\mathcal{K}_{00i}\,,
\end{equation}
where the first equation is a property of the fibration and the second and third relations can be derived using \eqref{eq:sigma^2}.
We also introduce the $p\times p$-matrix
\beq \label{eq:CmatOnB}
(C)_{ij}:=\mathcal{K}(D_0,D_i,D_j)=\mathcal{K}_{0ij}\,,
\eeq
which defines a bilinear pairing on divisors on the base $B$. For the cases we consider here its signature is $(1,p-1)$ for $\mathbb{F}_k$ 
and $dP_n$, $n=1,\ldots,8$, and $C=1$ for $\mathbb{P}^2=dP_0$. Note that it will be convenient at some places in this work to view 
$H^{(1,1)}(B)$ as a $p$-dimensional vector
space equipped with an  inner product \eqref{eq:CmatOnB}. We denote the inner product of two vectors $v$, $w$ in $H^{(1,1)}(B)$ simply by $C(v,w)$.  In addition, we view the first Chern class $c_1$ of $B$, the fluxes $F^\alpha$ and
the K\"ahler form $J$ as column vectors
\begin{equation}\label{eq:vectors}
j= 
\left( \begin{array}{c}
j_1 \\
. \\
. \\
. \\
j_p  \end{array} \right)
\qquad
m^{\alpha}= 
\left( \begin{array}{c}
m^{\alpha}_1 \\
. \\
. \\
. \\
m^{\alpha}_p  \end{array} \right)
\qquad
b= 
\left( \begin{array}{c}
b_1 \\
. \\
. \\
. \\
b_p  \end{array} \right)\,.
\end{equation}
Here the components of these vectors are defined via the expansion w.r.t.~the $D_I$ in \eqref{eq:D_I},
\begin{equation} \label{eq:basisexp}
	\pi^*c_1=\sum_{i=1}^p b_iD_i\,,\quad F^\alpha=m_0^\alpha D_0+\sum_{i=1}^p m^\alpha_i D_i\,,\quad J=j_0D_0+\sum_{i=1}^p j_i D_i\,,
\end{equation}
where $b_i\in \mathbb{Q}^+$, $m^\alpha_I \in \mathbb{Q}$ and $j_I \in \bR^+$.\footnote{We allow 
here for rational coefficients $m^\alpha_I$, $b_i$ in the expansion of  $F^\alpha$, $\pi^*c_1$ 
that are in the integral homology $H^{(1,1)}(X,\mathbb{Z})$ in order to account for the possibility of K\"ahler generators $D_I$ that only span a 
sublattice of $H^{(1,1)}(X,\mathbb{Z})$ of index greater than one. This can happen for non-simplicial K\"ahler cones.}

We emphasize that the flux quantization condition $F^\alpha\in H^{(1,1)}(X,\mathbb{Z})$ can be equivalently written as
\beq
	\int_{C}F^\alpha\in \bZ\,,\quad\qquad \forall\, C\in H_2(X,\mathbb{Z})\,,
\eeq  
where $C$ is any curve in $X$. Noting that the elliptic fiber 
$\mathcal{E}$ and the K\"ahler generators $D_i$ of $B$ are integral 
curves in $X$, this implies, using \eqref{eq:basisexp}, 
\beq \label{eq:integralitymalpha}
	\int_{\mathcal{E}}F^\alpha=m_0^\alpha\in \mathbb{Z}\,,\qquad \int_{D_i}F^\alpha=\sum_j^p C_{ij}m^\alpha_j \in \bZ\,.
\eeq

We conclude by noting that for smooth elliptically fibered Calabi-Yau threefolds, the second Chern class $c_2(X)$ can be computed explicitly, 
see e.g.~\cite{Friedman:1997yq} for a derivation. By adjunction one obtains $c_2(X)=12\sigma\cdot c_1+\pi^*(c_2+11c_1^2)$ with 
$c_2$ the second Chern class on $B$, employing the relation \eqref{eq:sigma^2}. Using this and \eqref{eq:C_IJKrels} we evaluate 
the curvature terms in \eqref{eq:TI} as
\beq \label{eq:T_iEvaluated}
	T_0=\int_B(c_2+11c_1^2)\,,\qquad T_i=12\int_{D_i}c_1=12\mathcal{K}_{00i}\,,
\eeq
which is straightforward to evaluate for concrete bases $B$.

\subsection{Basic Geometry of Almost Fano Twofolds}
\label{sec:B2geometries}

In this section we briefly discuss the geometrical properties of 
the almost Fano twofolds $B=\mathbb{F}_k$, $dP_n$ and the toric surfaces. 
The discussion in this section is supplemented by the explicit computations of the K\"ahler cones of $dP_n$
in appendix \ref{app:MoriKaehlerFano} and the summary of the key geometric
data of $\mathbb{F}_k$, $dP_n$ in Appendix \ref{app:explicitdata}, which is critical
for the proof in Section \ref{sec:finiteness}.

\subsubsection{Hirzebruch Surfaces}

The Hirzebruch surfaces $\mathbb{F}_k$ are $\mathbb{P}^1$-bundles 
over $\mathbb{P}^1$ of the form $\mathbb{F}_k=\mathbb{P}(\mathcal{O}\oplus \mathcal{O}(k))$. There is an infinite family 
of such bundles for every positive $k\in\mathbb{Z}_{\geq 0}$.

The isolated section of this bundle, $S$, and the fiber $F$ are effective 
curves generating the Mori cone and spanning the entire second homology 
\beq \label{eq:H2Fk}
	H_2(\mathbb{F}_k,\mathbb{Z})=\langle S,F\rangle\,.
\eeq
Their intersections read
\beq \label{eq:intsFk}
	S^2=-k\,,\qquad S\cdot F=1\,,\qquad F^2=0\,.
\eeq
From this we deduce that the generators $D_i$, $i=1,2$,
of the K\"ahler cone, which are 
defined to be dual to the generators in \eqref{eq:H2Fk}, read
\beq\label{eq:KaehlerConeFk}
	D_1=F\,,\qquad D_2=S+kF\,.
\eeq
The Chern classes on $\mathbb{F}_k$ read
\beq \label{eq:ChernFk}
	c_1(\mathbb{F}_k)=2S+(2+k)F=(2-k)D_1+2D_2\,,\qquad c_2(\mathbb{F}_k)=4\,,
\eeq
which implies that the vector $b$ in \eqref{eq:vectors} is $b=(2-k,2)^T$.

Using \eqref{eq:intsFk}, we compute the triple intersections in \eqref{eq:C_IJKrels}, in particular \eqref{eq:CmatOnB}, as
\beq \label{eq:tripleIntsFk}
	C=\begin{pmatrix}
	0 & 1 \\
	1 & k
\end{pmatrix}\,,\qquad 	 \mathcal{K}_{001}=2\,,\qquad \mathcal{K}_{002}=2+k\,,\qquad \mathcal{K}_{000}=8\,,
\eeq
from which the curvature terms in \eqref{eq:T_iEvaluated} immediately 
follow as
\beq \label{eq:T_iEvaluatedFk}
	T_0=92\,,\qquad T_1=24\,,\qquad T_2=24+12k
\eeq

We emphasize that $\mathbb{F}_k$ by means of \eqref{eq:ChernFk} is Fano 
for $k< 2$ and almost Fano for $k=2$, since the coefficient 
$b_1=2-k\geq 0$. The general elliptic Calabi-Yau fibration $X$ over $F_k$ 
with $k=0,1,2$ is smooth and develops $I_3$-singularities for $k=3$ up to 
$II^*$-singularities for $k=12$, before terminal singularities occur for 
$k>12$ \cite{Morrison:1996pp}. Thus, we focus on the Hirzebruch surfaces with 
$k=0,1,2$.

\subsubsection{Del Pezzo Surfaces}

The Fano del Pezzo surfaces $dP_n$ are the blow-up of $\mathbb{P}^2$ at 
up to eight generic points.\footnote{See \cite{Huang:2013yta,Huang:2014nwa} for 
recent computations of refined BPS invariants on del Pezzo surfaces as well as their interpretation in M-/F-theory.} 

Their second homology group is spanned by the pullback of the hyperplane 
on $\mathbb{P}^2$, denoted by $H$, and the classes of the exceptional 
divisors, denoted as $E_i$, $i=1,\ldots, n$, 
\beq \label{eq:H2dPn}
	H_2(dP_n,\mathbb{Z})=\langle H,E_{i=1,\ldots,n}\rangle\,.
\eeq
The intersections of these classes read
\beq \label{eq:intsdPn}
	H^2=1\,,\qquad H\cdot E_i=0\,,\qquad E_i\cdot E_j=-\delta_{ij}\,.
\eeq
The Chern classes on $dP_n$ read
\beq \label{eq:CherndPn}
	c_1(dP_n)=3H-\sum_{i=1}^nE_i\,,\qquad c_2(dP_n)=3+n\,.
\eeq

The Mori cone of $dP_n$ for $n>1$ is spanned by the curves $\Sigma$ obeying \cite{demazure1980seminaire,Donagi:2004ia}
\beq
\label{eq:MoriConedPn}
	\Sigma^2=-1\,,\qquad \Sigma\cdot [K_{dP_n}^{-1}]=1\,,
\eeq
where $[K_{dP_n}^{-1}]$ is the anti-canonical divisor in $dP_n$, which 
is dual to $c_1(dP_n)$. By adjunction, we see that the curves obeying 
\eqref{eq:MoriConedPn} obey the necessary condition for being 
$\mathbb{P}^1$'s. By solving the conditions \eqref{eq:MoriConedPn} with
the ansatz $a_0H+\sum_{i=1}^n a_i E_i$ for $a_0,\,a_i\in \mathbb{Z}$, we 
obtain a cone that is simplicial, i.e.~generated by $h^{(1,1,)}(B)=1+n$ 
generators, for $n=0,1,2$ and non-simplicial for $n>2$. The number of 
generators, beginning with $dP_2$, furnish irreducible 
representations of $A_1$, $A_1\times A_2$, $A_4$, $D_5$, $E_{n}$, for 
$n=6,7,8$, which concretely are $\mathbf{3}$, 
$\mathbf{2}\otimes\mathbf{3}$, 
$\mathbf{10}$, $\mathbf{16}$, $\mathbf{27}$, $\mathbf{56}$, 
$\mathbf{248}$.\footnote{The genuine roots in $H_2(dP_n)$ are the 
$-2$-curves orthogonal to $[K^{-1}_{dP_n}]$, i.e.~$\alpha_i=E_i-E_{i+1}$, $i=1,\ldots, n-1$, 
$\alpha_n=H-E_1-E_2-E_3$ for $n>2$. These act on 
$H_2(dP_n)$ by means of the Weyl group, cf.~\cite{demazure1980seminaire}.}
For the simplicial cases the Mori cone reads
\beq \label{eq:simplicialMCdPn}
	\mathbb{P}^2\,:\,\,\langle H\rangle\,,\qquad dP_1\,:\,\,\langle 
	E_1,H-E_1\rangle\,,\qquad dP_2\,:\,\,\langle E_1,E_2, 
	H-E_1-E_2\rangle\,
\eeq 
and we refer to appendix \ref{app:MoriKaehlerFano} for more details on the 
non-simplicial cases. 

Consequently, also the K\"ahler cones of the $dP_n$, which are the dual
of the Mori cones defined by \eqref{eq:MoriConedPn}, are non-simplicial 
for $n>2$. The K\"ahler cone is spanned by rational curves $\Sigma$ 
obeying 
\beq \label{eq:KaehlerConedPn}
	\Sigma^2=0\,,\quad\Sigma\cdot [K^{-1}_{dP_n}]=2\qquad \text{or}\qquad \Sigma^2=1\,,\quad \Sigma\cdot [K^{-1}_{dP_n}]=3\,,
\eeq
which again implies by adjunction that $\Sigma=\mathbb{P}^1$.
The solutions over the integers of these conditions yield the generators of the K\"ahler cone of $dP_n$
which again follow the representation theory of the above mentioned Lie 
algebras. The number of generators, starting with $dP_0$, is $1$, $2$, 
$3$, $5$, $10$, $26$, $99$, $702$ and $19440$, see appendix 
\ref{app:MoriKaehlerFano}. 
In the simplicial cases, the K\"ahler cone generators read 
\beq\label{eq:simplicialKCdPn}
	\mathbb{P}^2\,:\,D_1= H\,,\quad dP_1\,:\,D_1=H-E_1,\,D_2=H\,,\quad dP_2\,:\, D_1=H-E_1, \,
	D_2=H-E_2,\,D_3=H\,
\eeq

Generically, for $n\geq 2$ the vector $c_1(dP_n)$ is the 
center both of the K\"ahler and Mori cone. This implies that for all del 
Pezzo surfaces, the coefficients $b_i$ are positive. For
the simplicial K\"ahler cones, this can be computed explicitly. For
the non-simplicial cases we will argue in appendix 
\ref{app:MoriKaehlerFano}, that a covering of the K\"ahler cone by 
simplicial subcones,  i.e.~subcones with 
$h^{(1,1)}$ generators, with all $b_i\geq 0$ always exists. 
We note that for all $dP_n$, the defining property of the K\"ahler cone \eqref{eq:KaehlerConedPn}, together with \eqref{eq:C_IJKrels}, implies
the intersections
\beq \label{eq:C00iC000dPn}
	\mathcal{K}_{00i}=2,3\,, \qquad \mathcal{K}_{000}=9-n\,.
\eeq
In addition, by explicit computations we check in general that all 
$C_{ij}\geq 0$ for all pairs of K\"ahler cone generators. The intersections 
\eqref{eq:C00iC000dPn} together with \eqref{eq:intsdPn}, \eqref{eq:CherndPn} 
further imply that the curvature terms in \eqref{eq:T_iEvaluated}  read
\beq \label{eq:T_iEvaluateddPn}
	T_0=102-10n\,,\qquad T_i=24,\, 36
\eeq
For the three simplicial cases of $\mathbb{P}^2$, $dP_1$ and $dP_2$,
we compute the matrices \eqref{eq:CmatOnB} in the basis \eqref{eq:simplicialKCdPn} as
\beq \label{eq:tripleIntsdPnsimplicial}
	C_{\mathbb{P}^2}=1\,,\qquad  C_{dP_1}=\begin{pmatrix}
	0 & 1 \\
	1 & 1
\end{pmatrix}\,,\qquad  C_{dP_2}=\begin{pmatrix}
	0 & 1 & 1\\
	1 & 0 & 1\\
	1 & 1 & 1
\end{pmatrix}\,.
\eeq

We emphasize that the del Pezzos $dP_n$ by means of \eqref{eq:C00iC000dPn} are 
Fano for $n< 9$ and almost Fano for $n=9$,  since 
$c_1^2= 0$. The surface $dP_9$ is the rational elliptic surface. Its 
Mori cone is the Mordell-Weil group of rational sections by 
\eqref{eq:MoriConedPn}. Thus, it as well as its 
dual K\"ahler cone is infinite dimensional.  We will only consider the Fano del 
Pezzo surfaces $dP_n$, $n<9$.

\subsubsection{Toric Surfaces from Reflexive Polytopes}

Toric surfaces obtained from fine star triangulations of reflexive
polytopes are smooth almost Fano twofolds.\footnote{See the recent  
\cite{Klemm:2012sx} for a systematic study of the quantum geometry of the 
elliptically fibered Calabi-Yau manifolds over these bases.} There are $16$ such polytopes
in two dimensions, which are displayed in Figure \ref{fig:2dpoly}.

A number of these twofolds are simply toric descriptions of
previously described surfaces. Specifically, these are $\bP^2$,
$dP_1$, $dP_2$, $dP_3$, $\bF_0$ and $\bF_2$ which are described by polytopes
$1$, $3$, $5$, $7$, $2$ and $4$, respectively. From the form of some of
the other polytopes it is clear that they can be obtained from $\bP^2,
dP_1, dP_2, $ or $dP_3$ via toric blow-up. For example, reflecting polytope
$7$ through the vertical axis going through its center and performing
a toric blow-up associated to the point $(-1,1)$, one obtains polytope
$12$. Thus, the smooth Fano surface associated to polytope $12$ is
a toric realization of $dP_4$ at a non-generic point in its complex
structure moduli space.  

The toric varieties associated to all these 16 reflexive polytopes
can be constructed explicitly using the software package Sage \cite{sage}.  The 
intersections \eqref{eq:C_IJKrels}, \eqref{eq:CmatOnB} are readily constructed in a given fine
star triangulation and the K\"ahler cone can be obtained. We summarize the
geometric data necessary for the computation of the bounds derived below in the 
proof in Appendix \ref{app:explicitdata}.

\begin{figure}[h]
  \centering
  \includegraphics[scale=.4]{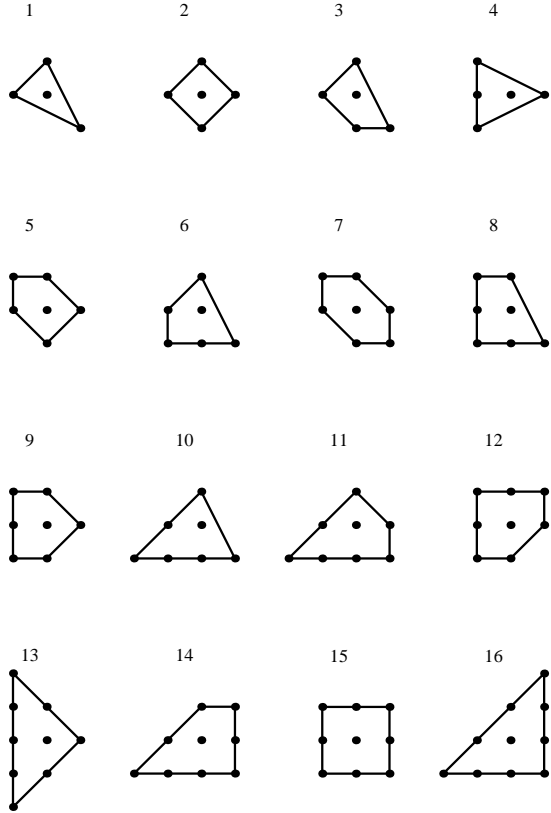}
  \caption{The sixteen two-dimensional reflexive polytopes which
    define the almost Fano toric surfaces via their fine star
    triangulations.}
  \label{fig:2dpoly}
\end{figure}

\section{Finiteness of Magnetized D9- \& D5-brane Configurations}
\label{sec:finiteness}

In this section we bound the number of possible gauge sectors
arising in the considered compactifications of Type IIB string theory. 

As emphasized in section \ref{sec:background}, the number $N^\alpha$ of branes in a stack and
their associated magnetic fluxes $F^\alpha$ are subject to the consistency conditions imposed by
tadpole cancellation conditions \eqref{Tadpoles}, \eqref{eq:D9tadpole} and
the SUSY conditions \eqref{eq:SUSYcondition}. Since the numbers $N^\alpha$
of D9-branes are bounded by \eqref{eq:D9tadpole}, it is
therefore the goal of this proof to bound the flux quanta $F^\alpha$ and the number of D5-branes
in $\Sigma^{\text{D5}}$.

Most of the proofs in this work have the same basic structure. The key point is to find
a  bound on the number of different flux configurations $F^\alpha$ and D5-branes $\Sigma^{\text{D5}}$
at an arbitrary point in the large volume regions of K\" ahler moduli space, i.e.~a bound
that is independent of the K\"ahler moduli. As we will see, proving this requires an intriguing 
interplay between both the tadpole conditions \eqref{Tadpoles}, \eqref{eq:D9tadpole} and the SUSY conditions
\eqref{eq:SUSYcondition}, a general rubic which was also used in the proof of \cite{Douglas:2006xy}\footnote{The interplay 
between SUSY and tadpole conditions has also been used in 
\cite{Blumenhagen:2004xx,Gmeiner:2005vz} for rigid 
$\mathbb{Z}_2\times\mathbb{Z}_2$-orientifolds and for other models in \cite{Gmeiner:2007zz,Honecker:2012qr}.}.  In addition, 
the following proof applies if  a list of geometrical properties, listed at the beginning of Section \ref{sec:generic proof}, are satisfied. 
These are obeyed for the considered examples $B=\mathbb{F}_k$, $dP_n$ and the toric surfaces. 

Before delving into the details of the proof, let us introduce a very important notation.
Because of their fundamentally different contributions to \eqref{Tadpoles}, \eqref{eq:D9tadpole} and \eqref{eq:SUSYcondition} it
is useful to split D9-brane stacks into 
to qualitatively different types according to their flux quanta. We denote D9-brane stacks with
$m^\alpha_0 \ne 0$ as \emph{$\beta$-branes}, and those with
$m^\alpha_0 = 0$ as \emph{$\gamma$-branes}:
\begin{equation}
\label{eq:BetaGammaBranes}
	\xymatrix @C=-1in {
	&\parbox{5cm}{\centering  D9-branes} \ar[rd]\ar[ld]  & \\
	\parbox{6.5cm}{\centering$\beta$-branes:  $m^\beta_0\neq 0$}&  &\parbox{6.5cm}{\centering  $\gamma$-branes: $m^\gamma_0= 0$}
	}
\end{equation}
In addition, in the rest of this section we label fluxes of a $\beta$- and $\gamma$-brane by $m^\beta_I$ and $m^\gamma_I$, 
respectively. 

We begin in Section \ref{sec:genericDiscussion} by preparing for the general finiteness proof by writing out the tadpoles and SUSY 
conditions of Section \ref{sec:background} for elliptically fibered Calabi-Yau threefolds $X$. 
We also make certain definitions and deduce a number 
of simple inequalities and bounds, that will be essential for the later discussion. Then, in Section \ref{sec:P2} we prove finiteness
for the special base $B=\mathbb{P}^2$, which will demonstrate the usefulness of the definitions of the previous section and serve as a 
warm-up for the general proof in Section \ref{sec:generic proof}.

\subsection{Prerequisites: Definitions \& Basic Inequalities}
\label{sec:genericDiscussion}

In this section we make some general definitions and observations 
necessary to formulate and organize the proof in Section \ref{sec:generic proof}. 

As a starting point, we observe that the SUSY conditions 
\eqref{eq:SUSYcondition} must be satisfied by each brane stack, but only involve the direction along the K\"ahler class $j$, 
whereas the tadpole conditions (\ref{Tadpoles}) have to be obeyed for each divisor $D_I$, but are summed across brane stacks. 
Thus, in order to bound each component $m^\alpha_I$ of every flux vector $m^\alpha$, labelled by the brane stack $\alpha$, it 
is crucial to identify quantities, that enter both types of constraints, when rewritten in a particular form. 

To this end, we write out the tadpole conditions explicitly in the basis of divisors \eqref{eq:D_I}.
The conditions (\ref{Tadpoles}) for $I=0$, to which 
we will refer in the future as the $0^{\text{th}}$-tadpole, reads
\begin{equation}\label{T0}
0^{\text{th}}\text{-tadpole}:\,\,\,n_0^{D5}-T_0 = \underbrace{\sum_{\beta}N^{\beta}C\left(b+\frac{m^{\beta}}{m^{\beta}_0}, b+\frac{m^{\beta}}{m^{\beta}_0}\right)(m_0^{\beta})^2}_{\beta\text{~brane contributions}}+\underbrace{\sum_{\gamma}N^{\gamma}C(m^{\gamma},m^{\gamma})}_{\gamma \text{~brane contributions}}
\end{equation}
where we  used \eqref{eq:CmatOnB} and
(\ref{eq:vectors}) and emphasized the respective contributions from $\beta$-branes 
and $\gamma$-branes. For $I=i$, to which we will refer as the $i^{\text{th}}$-tadpole, the tadpole (\ref{Tadpoles}) reads
\begin{equation}\label{eq:Ti}
i^{\text{th}}\text{-tadpole}:\,\,\,n_i^{D5}-T_i=\underbrace{\sum_{\beta}N^{\beta}t_i^{\beta}(m_0^{\beta})^2}_{\beta \text{~brane contributions}} \qquad \text{with} \qquad t_i^{\beta}\equiv 2\sum_{j=1}^pC_{ij}\left(\frac{b_j}{2}+\frac{m^{\beta}_j}{m^{\beta}_0}\right)\,.
\end{equation}
We note that the first term in $t^\beta_k$ can be written as $\sum_j b_jC_{ij}=\mathcal{K}_{00i}$ which is an integer by 
\eqref{eq:tripleIntsFk}, \eqref{eq:C00iC000dPn} and Table \ref{table:toric surface}.
The quantities $t_i^\beta$ can be defined for $\beta$-branes and play an 
important in the proof, because they naturally appear in 
the SUSY constraints. 
We emphasize that while both $\beta$-branes and $\gamma$-branes 
contribute to the $0^{\text{th}}$-tadpole condition, only $\beta$-branes
contribute to the $i^{\text{th}}$-tadpole as is indicated by the braces in \eqref{T0}, \eqref{eq:Ti}.

We note that one can immediately deduce a lower bound on the left hand side of 
\eqref{T0} and \eqref{eq:Ti} by setting the positive numbers $n_I^{\text{D5}}=0$:
\beq \label{eq:T_ILowerbounds}
	-T_0 \leq \sum_{\beta}N^{\beta}C\left(b+\frac{m^{\beta}}{m^{\beta}_0}, b+\frac{m^{\beta}}{m^{\beta}_0}\right)(m_0^{\beta})^2+\sum_{\gamma}N^{\gamma}C(m^{\gamma},m^{\gamma})\,,\quad
	-T_i\leq\sum_{\beta}N^{\beta}t_i^{\beta}(m_0^{\beta})^2\,.
\eeq
These lower bounds on the $i^{\text{th}}$-tadpoles imply, as we will see, 
that if the $t_i^\beta$ are bounded above, then they are automatically 
bounded below. This can be seen by bringing the bounded positive contribution to the left hand side of \eqref{eq:T_ILowerbounds}.

For $\beta$-branes, which have $m_0^\beta\ne 0$, it useful to divide the SUSY equality \eqref{eq:SUSYcondition} by
$m_0^\beta$. Using again \eqref{eq:CmatOnB} and
(\ref{eq:vectors}), we write the first condition in \eqref{eq:SUSYcondition} to
obtain 
\begin{equation}\label{betaSUSYeq}
\left[3C\left(\frac{j}{j_0}, \frac{j}{j_0}\right)+6C\left(\frac{b}{2}+\frac{j}{j_0}, b+\frac{m^{\beta}}{m^{\beta}_0}\right)\right]j_0^2 = \left[\frac{1}{4}\mathcal{K}_{000}+3C\left(\frac{b}{2}+\frac{m^{\beta}}{m_0^{\beta}}, \frac{b}{2}+\frac{m^{\beta}}{m_0^{\beta}}\right)\right](m^{\beta}_0)^2.
\end{equation}
The SUSY inequality in \eqref{eq:SUSYcondition} for $\beta$-branes can be combined with the 
SUSY equality \eqref{betaSUSYeq} as follows.  By dividing the SUSY inequality in \eqref{eq:SUSYcondition} by $j_0$
and subtracting the SUSY equality \eqref{betaSUSYeq}, we obtain after a few lines of algebra the following inequality:
\begin{equation}\label{StarWithC000}
0>\frac{1}{2}\mathcal{K}_{000}+6C\left(\frac{b}{2}+\frac{m^{\beta}}{m^{\beta}_0}, \frac{b}{2}+\frac{j}{j_0}\right)\,.
\end{equation}
This can equivalently be written in the form
\begin{equation}\label{StarOldFormWithC000}
0>\frac{1}{2}\mathcal{K}_{000}+3 \sum_i^p t_i^{\beta}\left(\frac{b_i}{2}+\frac{j_i}{j_0}\right)
\end{equation}
and we see that the expression $t_i^\beta$, which explicitly appears in 
the $i^{\text{th}}$-tadpole conditions in \eqref{eq:T_iEvaluated}, 
appears also in this manipulation of the SUSY constraints.

We note that \eqref{StarOldFormWithC000} can be related to the tadpole 
conditions. By multiplying \eqref{StarOldFormWithC000} by $N^{\beta}(m^{\beta}_0)^2$ and summing over $\beta$, we employ the right hand side of \eqref{eq:Ti} to obtain
\begin{equation} \label{sumstar}
0>\frac{1}{2}\mathcal{K}_{000}\sum_{\beta}N^{\beta}(m_0^{\beta})^2+3\sum_{i=1}^p(n_i^{\text{D5}}-T_i)\Big(\frac{b_i}{2}+\frac{j_i}{j_0}\Big)\geq \frac{1}{2}\mathcal{K}_{000}\sum_{\beta}N^{\beta}(m_0^{\beta})^2-3\sum_{i=1}^p T_i\Big(\frac{b_i}{2}+\frac{j_i}{j_0}\Big)\,,
\end{equation}
where we set $n_i^{\text{D5}}=0$ in the last inequality.
This condition is used throughout the proof.

Next, we demonstrate that it is possible to also rewrite the SUSY 
equality \eqref{betaSUSYeq} and the $0^{\text{th}}$-tadpole (\ref{T0}) 
in a form that  manifestly contains the 
quantities $t^\beta_i$.
To this end, we first define for each distinct pair of indices $\{i,k\}$, 
$i\neq k$, the matrix $M_{\{i,k\}}$ whose $(j,l)$-th entry in the basis $D_i$ is:
\begin{equation}\label{M}
(M_{\{i,k\}})_{jl}=x_{\{i,k\}}C_{ij}C_{kl}+x_{\{i,k\}}C_{il}C_{kj}-C_{jl}
\end{equation}
where $x_{\{i,k\}}\in \mathbb{Q^+}$ is a non-negative rational
number. This number has to be chosen such that its corresponding 
$M_{\{i,k\}}$ is positive semi-definite. We note, that the
matrices $M_{\{i,k\}}$ resemble the stress energy tensor of a system of
free particles, c.f.~Appendix \ref{app:CveticTheorem}. We use this
to show that, if the first condition in Section \ref{sec:generic proof}
is met,  there always exists an $x_{\{i,k\}}$ so that these matrices are 
positive semi-definite, see Appendices \ref{app:MoriKaehlerFano}
and \ref{app:CveticTheorem}. Thus, throughout the rest of this proof we 
assume that all matrices $M_{\{i,k\}}$ are positive semi-definite.

With this definition, the SUSY
equality (\ref{betaSUSYeq}) and $0^{\text{th}}$-tadpole (\ref{T0}) can 
be written as
\bea\label{betaSUSYeqWithM}
\Big[3C\Big(\tfrac{j}{j_0}, \tfrac{j}{j_0}\Big)\!+\!6C\Big(\tfrac{b}{2}\!+\!\tfrac{j}{j_0}, b\!+\!\tfrac{m^{\beta}}{m^{\beta}_0}\Big)\Big]j_0^2\!=\!\Big[\!\tfrac{1}{4}\mathcal{K}_{000}\!+\!\tfrac{3}{2}x_{\{i,k\}}t_i^{\beta}t_k^{\beta}\!-\!3M_{\{i,k\}}\Big(\tfrac{b}{2}\!+\!\tfrac{m^{\beta}}{m^{\beta}_0}, \tfrac{b}{2}\!+\!\tfrac{m^{\beta}}{m^{\beta}_0}\Big)\!\Big] (m^{\beta}_0)^2\nonumber
\eea
and 
\begin{equation}\label{T_0WithM}
n_0^{D5}-T_0 = \underbrace{\sum_{\beta}N^{\beta}\left[\frac{1}{2}x_{\{i,k\}}\tilde{t}_i^\beta \tilde{t}_k^{\beta}-M_{\{i,k\}}\left(b+\frac{m^{\beta}}{m^{\beta}_0}, b+\frac{m^{\beta}}{m^{\beta}_0}\right)\right](m_0^{\beta})^2}_{\beta-\text{~brane contributions}} 
+\underbrace{\sum_{\gamma}N^{\gamma}C(m^{\gamma},m^{\gamma})}_{\gamma-\text{~brane contributions}}\,,
\end{equation}
respectively, where we indicated the contributions from $\beta-$ and $\gamma$-branes by braces and used the short hand 
notation
\beq \label{eq:tildet_i}
	\tilde{t}^\beta_i=\mathcal{K}_{00i}+t_i^{\beta}\,.
\eeq

As we will see, the proof
of Section \ref{sec:generic proof} applies whenever the $M$-matrices in \eqref{M} are all
positive semi-definite. In fact, for all the bases $B$ of the threefold 
$X$ considered, this matrix is positive semi-definite.  
For $\bP^2,\bP^1 \times \bP^1, dP_1, dP_2,$ and $\bF_2$ the $M$-matrix
can be readily computed in the K\" ahler cone basis, and indeed, it is
positive semi-definite. However, for $dP_n$ with $n\ge 3$ there exists
a significant complication since in these examples, the K\" ahler cone 
is non-simplicial, as mentioned in Section \ref{sec:background}.
In these cases, we cover the K\"ahler cone by simplicial subcones 
consisting of $h^{(1,1)}$ generators and compute
the $M$-matrix \eqref{M} for this choice. As demonstrated 
in Appendix \ref{app:MoriKaehlerFano}, for $dP_n$, $n<9$, the 
$M$-matrices are positive semi-definite for all such subcones. For the toric surfaces, we 
refer to Appendix \ref{app:explicitdata} for positive semi-definiteness of the matrices \eqref{M}.  
Thus, for the rest of the paper we can assume that all $M_{\{i,k\}}$ are positive semi-definite for these bases.

\subsection{Warm Up: Finiteness for Elliptic Fibrations over $\bP^2$}
\label{sec:P2}

Before proceeding on to more difficult examples, let us prove 
finiteness in the simplest example of $B=\mathbb{P}^2$. In particular, 
in this example we will demonstrate the usefulness of the derived 
inequality \eqref{StarWithC000} and \eqref{sumstar}.

For an elliptically fibered Calabi-Yau threefold $X$ over $B=\bP^2$, 
the relevant geometrical data following from \eqref{eq:CherndPn}, \eqref{eq:C00iC000dPn}, \eqref{eq:T_iEvaluateddPn} and 
\eqref{eq:tripleIntsdPnsimplicial} is:
\begin{equation}\mathcal{K}_{000}=9\,, \quad \mathcal{K}_{001}=3\,,\quad \mathcal{K}_{011}\equiv C_{11}=1\,,\quad
  b_1=3\,,\quad T_1=36\,.
\end{equation}
Using this the inequality (\ref{StarWithC000})
reduces to
\begin{equation}\label{starineqforP^2}
0>\mathcal{K}_{001}(m_0^{\beta})^2+2\mathcal{K}_{011}m_0^{\beta}m_1^{\beta}\,.
\end{equation}
The tadpole for $D_1$ reads
\begin{equation}\label{T_1forP^2}
n_1^{D5}-T_1=\sum_{\beta}N^{\beta}\left[\mathcal{K}_{001}\,(m_0^{\beta})^2+2\,\mathcal{K}_{011}\,m_0^{\beta}m_1^{\beta}\right]\,.
\end{equation}

By (\ref{starineqforP^2}), the right hand side of (\ref{T_1forP^2}) 
must be negative. Thus we have a bound for $n_1^{\text{D5}}$, given by

\begin{equation}
n_1^{\text{D5}}<T_1.
\end{equation}

In addition, for each $\beta$-brane we deduce  
from \eqref{starineqforP^2} that
\bea
  0 &<&
  |m_0^{\beta}||\mathcal{K}_{001}m_0^{\beta}+2\mathcal{K}_{011}m_1^{\beta}|=|\mathcal{K}_{001}(m_0^{\beta})^2+2\mathcal{K}_{011}m_0^{\beta}m_1^{\beta}|
  \nonumber \\ &\leq&
  \sum_{\beta}N^{\beta}|\mathcal{K}_{001}(m_0^{\beta})^2+2\mathcal{K}_{011}m_0^{\beta}m_1^{\beta}|\leq
  T_1\,.
\eea
Notice that $|\mathcal{K}_{001}m_0^{\beta}+2\mathcal{K}_{011}m_1^{\beta}|$
is a non-zero integer by virtue of the strict inequality
(\ref{starineqforP^2}).  This implies the bound
\beq
|m_0^{\mathnormal{\beta}}|\leq T_1\,.
\eeq  
Next, since $|\mathcal{K}_{001}m_0^{\beta}+2\mathcal{K}_{011}m_1^{\beta}| \leq
T_1/|m_0^{\beta}|$ and $|m_0^\beta|$ is bounded, $m_1^\beta$ is also
bounded as 
\beq
|m_1^{\mathnormal{\beta}}|\leq
\frac{1}{2\mathcal{K}_{011}}\Big(\frac{T_1}{|m_0^{\mathnormal{\beta}}|}+\mathcal{K}_{001}|m_0^{\mathnormal{\beta}}|\Big)\,.
\eeq
Thus we have shown that the magnetic flux quanta $m^\beta$ associated 
to $\beta$-branes are bounded.

A bound on the flux quanta of $\gamma$-branes is straightforward to
obtain. The SUSY equality in \eqref{eq:SUSYcondition} for each 
$\gamma$-brane is
$\mathcal{K}_{011}\left(\frac{b_1}{2}+\frac{j_1}{j_0}\right)m_1^{\gamma}=0$. 
Since $\mathcal{K}_{011}\neq 0$ and $\left(\frac{b_1}{2}+\frac{j_1}{j_0}\right)$ 
is strictly positive, we must have $m_1^{\gamma}=0$. Since a $\gamma$-brane by 
definition has $m_0^\gamma=0$, the flux quanta of 
$\gamma$-branes are trivially bounded. This completes the 
proof for $B=\mathbb{P}^2$.

\subsection{Proving Finiteness for Two-Dimensional Almost Fano Bases}
\label{sec:generic proof}

In this section we present the general proof of the finiteness of the number of consistent
Type IIB compactification with magnetized D9-branes on smooth elliptically fibered Calabi-Yau threefolds.
As discussed before the bases $B$ for which the presented proof has been developed are the 
two-dimensional almost Fano varieties. These are the del Pezzo surfaces
$dP_n$, $n=0,\ldots,8$, with the case of $dP_0=\mathbb{P}^2$ discussed in the previous section \ref{sec:P2}, 
the Hirzebruch surfaces $\mathbb{F}_k$, $k=0,1$, including the almost Fano $\mathbb{F}_2$, as well as the toric surfaces.

The geometrical properties  that are essential for the following proof are the smoothness of the generic
elliptic Calabi-Yau fibration over them, as well as the following list of properties:
\begin{enumerate}
	\item[(1)] all K\"ahler cone generators of $B$ are time- or 
	light-like vectors in the same light-cone.
	\item[(2)] positivity of the coefficients $b_i$ in \eqref{eq:basisexp}, i.e.~$b_i\geq 0$ for all $i$.
	\item[(3)] positivity and integrality of $\mathcal{K}_{00i}$ as defined in \eqref{eq:C_IJKrels}, i.e.~$\mathcal{K}_{00i}\in\mathbb{Z}_{\geq0}$ 
	for all $i$.
	\item[(4)] the signature of the matrix $C_{ij}$ defined in \eqref{eq:CmatOnB} is $(1,n)$, where $n+1=h^{(1,1)}(B)$, i.e.~has one positive and $n$ negative eigenvalues. 
	\item[(5)] positivity of the K\"ahler parameters $j_i$ and validity of the large volume approximation, i.e.~$j_i\gg 1$ for all $i$.
\end{enumerate}
We claim that the proof presented below applies to all bases $B$ that obey these conditions. 

We note that properties (4) and (5) are automatically satisfied for all the surfaces we consider: the signature of the matrix $C_{ij}$ defined in 
\eqref{eq:CmatOnB} is $(1,n)$, cf.~Section \ref{sec:B2geometries}, and  $j_i\gg 1$ always holds in the K\"ahler cone basis at large volume for any $B$. 
The validity of properties (1)-(3) for the considered bases is shown in the Appendices \ref{app:MoriKaehlerFano} and \ref{app:explicitdata}. 
As discussed there, the only subtlety arises for the
higher del Pezzos $dP_n$, $n>2$, which have non-simplicial K\"ahler cones. In this case, the indices $i$ refer to the 
generators of a suitably chosen simplicial subcone, such that properties (1)-(3) hold. As argued in appendix
\ref{app:MoriKaehlerFano} there always exists a covering of the K\"ahler cones of the $dP_n$ by simplicial subcones, such that for each 
subcone in the covering properties (1)-(3) hold. 

The following proof is organized as follows. We already introduced the two types of branes, denoted $\beta$- and $\gamma$-branes,
to distinguish between branes with and without fluxes along the fiber $\mathcal{E}$, i.e.~$\int_\mathcal{E} F^\beta\neq 0$ 
and $\int_\mathcal{E} F^\gamma=0$, respectively. First we prove in Section \ref{sec:BoundsBeta} that there is only a finite
number of flux configurations on $\beta$-branes. Then in Section \ref{sec:BoundsD5} we show finiteness of the numbers of 
D5-branes $n_I^{\text{D5}}$. Finally, we conclude the proof in Section \ref{sec:BoundsGamma} by showing  finiteness of 
the number of flux configurations on $\gamma$-branes.

\subsubsection{Bounds on $\beta$-branes}
\label{sec:BoundsBeta}

\subsubsection*{Bounds on $m^{\beta}_0$}

In the following we obtain a bound on the flux component $m_0^\beta$ for all $\beta$-branes. The result is
\beq \label{eq:m0alphabound}
	|m_0^\beta|\leq \text{max}(T_i)\,,
\eeq 
where the maximum is taken over all generators of the specific subcone 
of the K\"ahler cone. We note that here and in the rest of the paper, 
all minima and maxima on $T_i$ and $x_{\{i,k\}}$ are taken across 
generators of the specific subcone we are in. However, except the 
minimum on $T_i$ in theorem \ref{th:boundn0D5}, the reader is free to 
take all other maxima and minima across all generators of the entire 
K\"ahler cone, for easy computation purposes.
For del Pezzo surfaces this yields 
$\text{max}(T_i)=36$, for the Hirzebruch surfaces 
$\mathbb{F}_k$ it is $\text{max}(T_i)=24+12k$ and for the toric surfaces we can read off this bound from Table \ref{table:toric surface}.

We begin by considering inequality \eqref{StarOldFormWithC000}. In fact, since $\mathcal{K}_{000}\geq 0$,  \eqref{StarOldFormWithC000} 
implies 
\beq \label{eq:StarWithoutC000}
	0 > \sum_i^p t_i^\beta(\frac{b_i}{2}+\frac{j_i}{j_0})
\eeq
Next we multiply this by $N^\beta (m_0^\beta)^2$ and sum over $\beta$ to obtain, using \eqref{eq:Ti},
\beq \label{eq:StarWithLowerBoundPart1}
	0>\sum_\beta\sum_i^p N^\beta t_i^{\beta}(m_0^{\beta})^2\left(\frac{b_i}{2}+\frac{j_i}{j_0}\right)=\sum_i^p\left(n_i^{\text{D5}}-T_i\right)\left(\frac{b_i}{2}+\frac{j_i}{j_0}\right)\geq\sum_i^p\left(-T_i\right)\left(\frac{b_i}{2}+\frac{j_i}{j_0}\right) 
\eeq
where we set the positive $n_i^{\text{D5}}=0$ for all $i$ in the last inequality.
This  lower bound on the sum over $\beta$ also implies 
\begin{equation}\label{eq:StarWithLowerBound}
0>\sum_{i=1}^p N_\beta t_i^{\beta}(m_0^{\beta})^2\left(\frac{b_i}{2}+\frac{j_i}{j_0}\right) \geq \sum_{i=1}^p(-T_i)\left(\frac{b_i}{2} + \frac{j_i}{j_0}\right)\,.
\end{equation}
because by \eqref{eq:StarWithoutC000} all summands are negative. This motivates the following definition:
\begin{definition}
\label{def:specialmixedbranes}
A \textbf{special brane} is a  $\beta$-brane with $t_i^{\beta}<0 $ for all $i$. A \textbf{mixed brane} is  a $\beta$-brane which is not a special 
brane (i.e.~there exists an $i$ such that $t_i^\beta\geq 0$).
\end{definition}

\begin{remark}
\label{remark1}
\normalfont{By \eqref{eq:StarWithoutC000}, there does not exist a mixed brane with $t_i\geq 0$ $\forall i$, since $b_i, j_i\geq 0$.} Hence for a mixed brane, we cannot have $t_i$ of the same sign $\forall i$, they must be of mixed signs. This motivates its name.
\end{remark}

For special branes, we immediately conclude from \eqref{eq:StarWithLowerBound} that 
\bea
\text{max}(T_i)\sum_{i=1}^p\left(\frac{b_i}{2} + \frac{j_i}{j_0}\right) & \geq& \sum_{i=1}^pT_i\left(\frac{b_i}{2} + \frac{j_i}{j_0}\right) 
\geq \sum_{i=1}^p N_\beta|t_i^{\beta}|(m_0^{\beta})^2\left(\frac{b_i}{2}+\frac{j_i}{j_0}\right) \\
& =&
\sum_{i=1}^p \underbrace{N_\beta|t_i^{\beta}m_0^{\beta}|}_{\text{$\in \mathbb{N}$, $\geq 1$}}|m_0^{\beta}|\left(\frac{b_i}{2}+\frac{j_i}{j_0}\right) \geq
|m_0^{\beta}| \sum_{i=1}^p \left(\frac{b_i}{2}+\frac{j_i}{j_0}\right)\,.\nonumber
\eea
Here we have used \eqref{eq:StarWithLowerBound} in the second inequality, and that $ t_i^\beta m_0^\beta=\sum_j C_{ij}(b_jm_0^\beta+2m_j^\beta)$ is a non-zero positive integer, cf.~\eqref{eq:Ti} in the last inequality: it is an integer because both its first term, $\mathcal{K}_{00i}m_0^\beta$, and the second term, the flux $F^\beta$ integrated over the integral class $D_i$, are integers by \eqref{eq:integralitymalpha}. It is non-zero because $t_i^{\beta}$ is non-zero by the definition of special branes, and $m_0^{\beta}$ is non-zero by the definition of $\beta$-branes. Thus for special branes, the flux quantum $m_0^\beta$ is bounded as
\beq
\label{eq:m0specialbrane}
|m_0^{\beta}|\leq
\text{max}(T_i)\,.
\eeq  
We will show that mixed branes have a even smaller bound for their $|m_0^{\beta}|$.

Let us first make an observation that will facilitate the identification of special branes.
\begin{lemma}\label{lem:largej}
A $\beta$-brane which satisfies $0\leq C\left(\frac{b}{2}+\frac{j}{j_0}, b+\frac{m^{\beta}}{m^{\beta}_0}\right)$ is a special brane.
\end{lemma}

\begin{proof}
For any $\beta$ brane with $0\leq C\left(\frac{b}{2}+\frac{j}{j_0}, b+\frac{m^{\beta}}{m^{\beta}_0}\right)$, consider its SUSY equality (\ref{betaSUSYeqWithM}). Then 
\begin{equation}\label{eq:LHSofSUSYatLeast}
\text{LHS of}~ (\ref{betaSUSYeqWithM}) \geq 3C\left(\frac{j}{j_0}, \frac{j}{j_0}\right)j_0^2\,.
\end{equation}

Suppose it is not a special brane. Then by definition we cannot have $t^{\beta}_i <0 ~\forall i$. Remark \ref{remark1} also forbids 
$t^{\beta}_i \geq 0 ~\forall i$. Thus there exists a pair of $i,k$ such that $t^{\beta}_i$ and $t^{\beta}_k$ are of opposite signs (the following argument still applies if one of them is zero). Writing the RHS of (\ref{betaSUSYeqWithM}) in terms of this particular pair of $t^{\beta}_i,t^{\beta}_k$, we 
observe that 
\begin{equation}\label{eq:RHSofSUSYatMost}
\text{RHS of}~\eqref{betaSUSYeqWithM} \leq \frac{1}{4}\mathcal{K}_{000}(m^{\beta}_0)^2 \leq \frac{1}{4}\mathcal{K}_{000}\sum_{\beta}N^{\beta}(m_0^{\beta})^2 < \frac{3}{2}\sum_{i=1}^p(T_i)\left(\frac{b_i}{2}+\frac{j_i}{j_0}\right)
\end{equation}
where in the first inequality we dropped all negative terms on the RHS of \eqref{betaSUSYeqWithM} 
and in the last inequality we employed the lower bound on \eqref{sumstar}. 
Now \eqref{eq:LHSofSUSYatLeast} shows that the LHS of \eqref{betaSUSYeqWithM} is at
least quadratic in the $j_i$'s and grows as the K\"ahler volume of $B$. However, inequality \eqref{eq:RHSofSUSYatMost}
implies that the RHS of
\eqref{betaSUSYeqWithM} is at most on the order of $j_i/j_0$. In the
limit of all $j_I$ large, which in particular implies large volume of $B$, the LHS of
\eqref{betaSUSYeqWithM} has to be greater than the RHS of \eqref{betaSUSYeqWithM}. Thus, the SUSY
equality \eqref{betaSUSYeqWithM} is violated. Our initial assumption that this $\beta$-brane is not a special brane must be wrong; it must 
be a special brane.
\end{proof}

\begin{remark}
\label{remark:ValueOfji}
\normalfont{The argument in Lemma \ref{lem:largej} about the growth of the two sides of the SUSY equality \eqref{betaSUSYeqWithM}
can be further substantiated for concrete bases $B$.  For all $\bF_k$, we can check that we have LHS of \eqref{betaSUSYeqWithM}$>$ RHS of
\eqref{betaSUSYeqWithM} when $j_I\geq 3 ~\forall I$. This is clearly the case if the 
supergravity approximation is supposed to be valid. For $dP_n$, the matrix $C(\cdot,\cdot)$ has signature $(1,n)$, i.e.~we can
have $C(j,j)=0$ for $j\neq 0$ and the above argument might be invalidated. However,  we can 
only have $C(j,j)=0$  if the K\"ahler form $j_B=\sum_i j_iD_i$ on $B$ is on the boundary of the K\"ahler cone.  This means that the 
K\"ahler volume of $B$ is zero or cycles in $B$ have shrunk to zero which clearly invalidates the supergravity approximation.}
\end{remark}

Thus, it remains to bound $m_0^{\beta}$ for $\beta$-branes satisfying 
$0 \geq C\left(\frac{b}{2}+\frac{j}{j_0},b+\frac{m^{\beta}}{m^{\beta}_0}\right)$. For such $\beta$-branes, we observe 
\beq
0\geq C\left(\frac{b}{2}+\frac{j}{j_0}, b+\frac{m^{\beta}}{m^{\beta}_0}\right)
 = \frac{1}{2}\sum_{i=1}^p\mathcal{K}_{00i}\left(\frac{b_i}{2}+\frac{j_i}{j_0}\right) + \frac{1}{2}\underbrace{\sum_{i=1}^pt_i^{\beta}\left(\frac{b_i}{2}+\frac{j_i}{j_0}\right)}_{<0\text{ by (\ref{eq:StarWithoutC000})}}
\eeq
using (\ref{eq:C_IJKrels}) and the definition of $t_i^{\beta}$ \eqref{eq:Ti}. Next, label all $\beta$-branes with 
$0 \geq C\Big(\frac{b}{2}+\frac{j}{j_0},
  b+\frac{m^{\beta}}{m^{\beta}_0}\Big)$ by $\beta'$, multiply the above inequality by $N^{\beta'}(m_0^{\beta'})^2$ and sum over $\beta'$:
\begin{align}
0 & \geq \frac{1}{2}\sum_{i=1}^p\mathcal{K}_{00i}\left(\frac{b_i}{2}+\frac{j_i}{j_0}\right)\sum_{\beta'}N^{\beta'}(m_0^{\beta'})^2 + \frac{1}{2}\sum_{\beta'}N^{\beta'}(m_0^{\beta'})^2\sum_{i=1}^pt_i^{\beta'}\left(\frac{b_i}{2}+\frac{j_i}{j_0}\right)\nonumber\\
& \geq \frac{1}{2}\sum_{i=1}^p\mathcal{K}_{00i}\left(\frac{b_i}{2}+\frac{j_i}{j_0}\right)\sum_{\beta'}N^{\beta'}(m_0^{\beta'})^2 + \frac{1}{2}\sum_{\beta}N^{\beta}(m_0^{\beta})^2\sum_{i=1}^pt_i^{\beta}\left(\frac{b_i}{2}+\frac{j_i}{j_0}\right)\nonumber\\
& = \frac{1}{2}\sum_{i=1}^p\mathcal{K}_{00i}\left(\frac{b_i}{2}+\frac{j_i}{j_0}\right)\sum_{\beta'}N^{\beta'}(m_0^{\beta'})^2 + \frac{1}{2}\sum_{i=1}^p(n_i^{D5}-T_i)\left(\frac{b_i}{2}+\frac{j_i}{j_0}\right)\nonumber\\
& \geq \frac{1}{2}\sum_{i=1}^p\mathcal{K}_{00i}\left(\frac{b_i}{2}+\frac{j_i}{j_0}\right)\sum_{\beta'}N^{\beta'}(m_0^{\beta'})^2 + \frac{1}{2}\sum_{i=1}^p(-T_i)\left(\frac{b_i}{2}+\frac{j_i}{j_0}\right)\,.
\end{align}
Here in the second line we extended the sum over $\beta'$ to the sum over all $\beta$-branes; by \eqref{eq:StarWithoutC000} each summand 
is negative, thus, extending the sum only decreases it. In the third line we have used \eqref{eq:Ti}. With \eqref{eq:T_iEvaluated}  and
the last line of the above inequality we obtain
\begin{equation}
12\sum_{i=1}^p\mathcal{K}_{00i}\left(\frac{b_i}{2}+\frac{j_i}{j_0}\right) \geq \sum_{\beta'}N^{\beta'}(m_0^{\beta'})^2 \sum_{i=1}^p\mathcal{K}_{00i}\left(\frac{b_i}{2}+\frac{j_i}{j_0}\right)
\end{equation}
Comparing coefficients, we see $\sum_{\beta'}N^{\beta'}(m_0^{\beta'})^2 \leq 12$ which implies the bound
\beq
|m_0^{\beta'}|\leq 3\,.
\eeq 

This is an even smaller bound than \eqref{eq:m0specialbrane} derived previously for special branes satisfying 
$0\leq C\left(\frac{b}{2}+\frac{j}{j_0}, b+\frac{m^{\beta}}{m^{\beta}_0}\right)$ because each $T_i=12\mathcal{K}_{00i}$ is a integer multiple of 12. 
Thus, the overall bound on  $m_0^{\beta}$  for a $\beta$-brane is still $|m_0^{\beta}|\leq \text{max}(T_i)$. 

Recall $\gamma$ branes by definition have $m_0^{\gamma}=0$. Thus we are done bounding $m_0^\alpha$, where 
\eqref{eq:m0alphabound} is the concrete, computable bound. In summary, we have found the  precise bounds in Table 
\ref{tab:m0bounds}.
\begin{table}[ht!]
\begin{tabular}{|c|c|c|c|c|}
\hline
Branes & Special branes with & Special branes with &\! Mixed branes\!&\!\! $\gamma$-branes\!\rule{0pt}{12pt}\\
\!\! & \!\! $0\leq C\left(\frac{b}{2}+\frac{j}{j_0}, b+\frac{m^{\beta}}{m^{\beta}_0}\right)$\!\!& \!\!$0\geq C\left(\frac{b}{2}+\frac{j}{j_0}, b+\frac{m^{\beta}}{m^{\beta}_0}\right)$\!\!  &  & \\ \hline
 $m_0^\alpha$-bound & $|m_0^\beta|\leq \text{max}(T_i)$ & $|m_0^\beta|\leq 3$ & $|m_0^\beta|\leq 3$ & $m_0^\gamma=0$\rule{0pt}{12pt}\\ \hline
\end{tabular}
\caption{Summary of bounds on $m^\alpha_0$.}
\label{tab:m0bounds}
\end{table} 

\subsubsection*{Bounds on the number of Solutions to the Vector $m^{\beta}$}

We begin by noting that \eqref{eq:Ti} can be viewed as the following matrix multiplication equation 
\begin{equation}
t^{\beta}:=
\left( \begin{array}{c}
t^{\beta}_1 \\
\cdot \\
\cdot\\
\cdot \\
t^{\beta}_p  \end{array} \right)= 2C\cdot\left(\frac{b}{2}+\frac{m^{\beta}}{m^{\beta}_0}\right)\,.
\end{equation}
The invertible matrix $2C$ gives a 1-1 correspondence between the
vector $m^{\beta}$ and the vector $t^{\beta}$. Thus, in order to show
that there are finitely many solutions for the vector $m^{\beta}$, we can
equivalently show that there are finitely many solutions for the
vector $t^{\beta}$. 

We can accomplish this by showing each component $t^{\beta}_i$ is bounded. We recall that it suffices to 
prove each $t^{\beta}_i$ is bounded above: since $(m^{\beta}_0)^2$ is bounded as we have just shown,  an upper 
bound also implies a lower bound by the second inequality in 
\eqref{eq:T_ILowerbounds}. Since the $t^\beta_i$ of special branes are by definition bounded above by $0$, see Definition 
\ref{def:specialmixedbranes}, we only have to bound  the $t^\beta_i$ of mixed branes. 

It is important for finding this upper bound on the $t_i^\beta$, to first analyze how each type of branes contribute to the sign of a tadpole.  We obtain the table \ref{tab:TadpoleContributions}, where we have indicated in parenthesis 
where the corresponding result will be proven in this work. 
\begin{table}[ht!]
\begin{center}
\begin{tabular}{|c||c|c|c|c|}
\hline
 & \multicolumn{2}{c|}{Special branes} &\! \!\! Mixed branes\! \!\!& $\gamma$-branes\rule{0pt}{12pt}\\ \hline \hline
\!\!$0^{\text{th}}$-tadpole\rule{0pt}{12pt}\!\! & positive & negative & negative & negative   \\ 
&\!\!($\Rightarrow$  $\forall\,\tilde{t}_i^\beta<0$ by Cor.~\ref{cor:T0+then(C00i+ti)-})\!\!&\!\! \phantom{($\Rightarrow$  $\forall\,
\tilde{t}_i^\beta<0$ by Cor\ref{cor:T0+then(C00i+ti)-})}\!\!&(by Prop.~\ref{prop:T0+specialbrane}) &\!\!\! (by Prop.~\ref{prop:gammaContribute-T0})\!\!\!\\ \hline
\!\!$i^{\text{th}}$-tadpole\!\! & \multicolumn{2}{c|}{negative}  & $\text{sign}(t_i^\beta)$ & $0$ \rule{0pt}{12pt}\\ 
 & \multicolumn{2}{c|}{(by \eqref{eq:Ti} and Def.~\ref{def:specialmixedbranes})} & (by \eqref{eq:Ti})  &\\ \hline
\end{tabular}
\caption{Summary of the contributions of the different types of branes to the different tadpoles.}
\label{tab:TadpoleContributions}
\end{center}
\end{table} 

Next, we proceed with proving the results of this table. We begin with the following 

\begin{proposition}\label{prop:gammaContribute-T0}
$\gamma$-branes only contribute negatively to the $0^{\text{th}}$-tadpole (\ref{T0}). Furthermore, any $\gamma$-brane contributing 
zero to the $0^{\text{th}}$-tadpole is the trivial brane, i.e.~$m^\gamma_I=0$ for all $I$.
\end{proposition}

\begin{proof}
A $\gamma$-brane's contribution to the $0^{\text{th}}$-tadpole is proportional
to $C(m^{\gamma},m^{\gamma})$ by \eqref{T0}. In addition, for $\gamma$-branes, the SUSY equality
in \eqref{eq:SUSYcondition} reads
\begin{equation}\label{gammaSUSY}
C\left(\left(\frac{b}{2}+\frac{j}{j_0}\right), m^{\gamma}\right)=0\,,
\end{equation}
as can be seen by setting $m_0^\gamma=0$ and using the intersection relations \eqref{eq:C_IJKrels}.

We recall that $C$ has Minkowski signature $(1,1)$ for $\mathbb{F}_k$ and $(1,n)$ for $dP_n$ and the toric surfaces. 
The vector $\frac{b}{2}+\frac{j}{j_0}$ is time-like, since
\begin{align}\label{c(v,v)+ive}
C\left(\left(\frac{b}{2}+\frac{j}{j_0}\right), \left(\frac{b}{2}+\frac{j}{j_0}\right)\right) = \frac{1}{4}\mathcal{K}_{000}+\sum_{i=1}^p\mathcal{K}_{00i}\frac{j_i}{j_0}+C\left(\frac{j}{j_0}, \frac{j}{j_0}\right)>0\,.
\end{align}
Here, the first term on the RHS of (\ref{c(v,v)+ive}) is positive because
$\mathcal{K}_{000}=8$ for $\mathbb{F}_n$, $9-n$ for $dP_n$ and Table \ref{table:toric surface} applies for toric surfaces. 
The second term is positive because $j_I>0$ and for $\mathbb{F}_k$, $\mathcal{K}_{001}=2$, $\mathcal{K}_{002}=2+k$; for $dP_n$, 
$\mathcal{K}_{00i}=2,3$; for toric surfaces, all relevant entries in Table \ref{table:toric surface} are positive. Finally, the third term is positive because it is proportional to the volume
of $B$. By \eqref{gammaSUSY} the vector $m^\gamma$ is orthogonal to a time-like vector, thus, it is 
space-like, i.e. $0>C(m^{\gamma},m^{\gamma})$, unless it is the zero vector, which trivially has $C(m^{\gamma},m^{\gamma})=0$. 

%
%
%

\end{proof}

\begin{proposition}\label{prop:T0+specialbrane}
Only special branes contribute positively to the $0^{\text{th}}$-tadpole. This is equivalent to the fact, that mixed branes only contribute
negatively to the $0^{\text{th}}$-tadpole.
\end{proposition}

\begin{proof}
We recall that the $0^{\text{th}}$-tadpole can be written in the form (\ref{T_0WithM}) for arbitrary choices of $\{i,k\}$, $i\neq k$. 
Focusing on its RHS, we
note that the second term is always negative by the positive semi-definiteness of the matrices $M_{\{i,k\}}$. Furthermore,
the third term is always negative by Proposition \ref{prop:gammaContribute-T0}. Thus, the RHS of \eqref{T_0WithM}
can only be positive, if the first term on the RHS is positive. This implies that all $\tilde{t}^\beta_i=\mathcal{K}_{00i}+t_i^\beta$, cf.~\eqref{eq:tildet_i},  have to be of the same sign:
if not, there exists a pair $\tilde{t}_i^\beta$, $\tilde{t}_k^\beta$ of opposite sign. Writing the RHS of \eqref{T_0WithM} in terms of this 
pair, the first term is negative and the entire RHS of \eqref{T_0WithM} would be negative.  

If all $\tilde{t}_i^\beta$ are negative, all $t_i^{\beta}$ have to be strictly negative since each 
$\mathcal{K}_{00i}$ are strictly positive. By Definition \ref{def:specialmixedbranes}, a $\beta$-brane with this property is a special brane. 
If the $\tilde{t}_i^\beta$ are all positive, then we have
  $\frac{1}{2}\sum_{i=1}^p\tilde{t}_i^{\beta}\left(\frac{b_i}{2}+\frac{j_i}{j_0}\right)=C\left(\frac{b}{2}+\frac{j}{j_0},
    b+\frac{m^{\beta}}{m^{\beta}_0}\right)\geq 0$, and by Lemma \ref{lem:largej} it is also a special brane.
\end{proof}

\begin{corollary}\label{cor:T0+then(C00i+ti)-}
A special brane that contributes positively to the $0^{\text{th}}$-tadpole must have $\tilde{t}_i^{\beta}<0$ for all 
$i$.
\end{corollary}

\begin{proof}
Recall from the proof of Proposition \ref{prop:gammaContribute-T0} that a special brane which contributes positively 
to the $0^{\text{th}}$-tadpole must have all $\tilde{t}^\beta_i$ of the same sign. If they are all negative, we are done.
Thus, assume all $\tilde{t}^\beta_i\geq 0$. We prove this is not possible using a similar argument as in the proof
of Lemma \ref{lem:largej}. 

Since $\mathcal{K}_{00i}>0$
$\forall i$ and we are considering a special brane, i.e.~all $t_i^\beta<0$, having 
$\tilde{t}^\beta_i=(\mathcal{K}_{00i}+t_i^{\beta})\geq0$ $\forall i$ means
$|t_i^{\beta}|\leq \mathcal{K}_{00i}$ $\forall i$. Now consider the SUSY equality
(\ref{betaSUSYeqWithM}). Since the $M$-matrix is positive
semi-definite, the RHS of (\ref{betaSUSYeqWithM}) is at most
$\left[\frac{1}{4}\mathcal{K}_{000}+\frac{3}{2}x^{\{i,k\}}\mathcal{K}_{00i}\mathcal{K}_{00k}\right](m_0^{\beta})^2$. Also, by the last inequality in (\ref{sumstar}), we have
\beq
6\sum_{i=1}^p T_i\left(\frac{b_i}{2}+\frac{j_i}{j_0}\right)> \mathcal{K}_{000}\sum_{\beta}N_{\beta}(m_0^{\beta})^2\,,
\eeq
i.e.~$(m_0^\beta)^2$ is smaller than a linear combination of $j_i/j_0$, so is
$\left[\frac{1}{4}\mathcal{K}_{000}+\frac{3}{2}x^{\{i,k\}}\mathcal{K}_{00i}\mathcal{K}_{00k}\right](m_0^{\beta})^2$, since the prefactor 
$\left[\frac{1}{4}\mathcal{K}_{000}+\frac{3}{2}x^{\{i,k\}}\mathcal{K}_{00i}\mathcal{K}_{00k}\right]\sim \mathcal{K}_{000}$. However, $\tilde{t}_i^{\beta}>0$
for all $i$ means
$\frac{1}{2}\sum_{i=1}^p\tilde{t}_i^{\beta}\left(\frac{b_i}{2}+\frac{j_i}{j_0}\right)=C\left(\frac{b}{2}+\frac{j}{j_0},
  b+\frac{m^{\beta}}{m^{\beta}_0}\right)\geq 0$, which implies that the LHS of
(\ref{betaSUSYeqWithM}) is at least $3C(j,j)$ which is quadratic in the
$j_i$. 

Thus, in the limit that all $j_I$ are large, the LHS of
(\ref{betaSUSYeqWithM}) will always be greater than its RHS, thus
violating the SUSY equality.\footnote{The precise value of the $j_I$ at which the SUSY equality is violated can be computed as mentioned 
in Remark \ref{remark:ValueOfji}. For example, for $\mathbb{F}_k$, we find that the SUSY equality is violated for 
$j_I\geq 10$ $\forall I$.}
\end{proof}

This concludes the proof of the results in Table 
\ref{tab:TadpoleContributions}.  We prove three more important Lemmas before we finally derive the bounds on $t_i^\beta$. 

For the rest of the proof, we will label special branes that contribute positively to the $0^{\text{th}}$-tadpole by $\beta_s$, and mixed branes by $\beta_m$. We also use the simplified notation 
\begin{equation}
\sum_{\beta_m,\,+} \equiv \sum_{\beta_m, \, \tilde{t}_i^{\beta_m} \geq 0}\,.
\end{equation}
The index $i$ is omitted in this simplified notation when it is clear from the context to which $i$ we are referring. 

\begin{lemma}\label{lemma:ShiftedTadpole}
For any index $i$, we have the following inequality:
\begin{equation} \label{eq:inequalityBetamBetas}
\sum_{\beta_s }N^{\beta_s}
\Big|\tilde{t_i}^{\beta_s}\Big|(m_0^{\beta_s})^2 < \sum_{\beta_m, \, +}N^{\beta_m}\tilde{t}_i^{\beta_m}(m_0^{\beta_m})^2\ + T_i\,.
\end{equation}
\end{lemma}

\begin{proof}
By \eqref{eq:T_ILowerbounds}, we have a lower bound for the 
$i^{\text{th}}$-tadpole. Thus, we have the following inequality for the 
$i^{\text{th}}$-tadpole:
\begin{align}\label{eq:TiShifted}
T_i \geq & \!\!\! \sum_{\beta, \, t_i^{\beta}<0} \!\!\! N^{\beta}
\Big|t_i^{\beta}\Big|(m_0^{\beta})^2 -\!\!\!\!\! \sum_{\beta_m, \, t_i^{\beta_m} \geq 0}\!\!\!\!\!N^{\beta_m}t_i^{\beta_m}(m_0^{\beta_m})^2
\geq \sum_{\beta_s }N^{\beta_s}
\Big|t_i^{\beta_s}\Big|(m_0^{\beta_s})^2 -\!\!\!\!\! \sum_{\beta_m, \, t_i^{\beta_m} \geq 0}\!\!\!\!\!N^{\beta_m}t_i^{\beta_m}(m_0^{\beta_m})^2 \nonumber\\
> & \sum_{\beta_s }N^{\beta_s}
\Big|t_i^{\beta_s}\Big|(m_0^{\beta_s})^2 -\!\!\!\!\! \sum_{\beta_m, \, t_i^{\beta_m} \geq 0}\!\!\!\!\!N^{\beta_m}t_i^{\beta_m}(m_0^{\beta_m})^2 -\sum_{\beta_s}N^{\beta_s}\mathcal{K}_{00i}(m_0^{\beta_s})^2- \!\!\!\!\!\!\!\! \sum_{\beta_m, \, 
t_i^{\beta_m} \geq 0} \!\!\!\!\!\!\!\! N^{\beta_m}\mathcal{K}_{00i}(m_0^{\beta_m})^2 \nonumber\\
& -\!\!\!\!\!\!\!\!\!\!\!\!\!\!\! \sum_{\beta_m, \, t_i^{\beta_m} < 0, \, \tilde{t}_i^{\beta_m} \geq 0} \!\!\!\!\!\!\!\!\!\!\!\!\!\! N^{\beta_m}\tilde{t}_i^{\beta_m}(m_0^{\beta_m})^2 \nonumber\\
= & \sum_{\beta_s }N^{\beta_s}
\Big|\tilde{t_i}^{\beta_s}\Big|(m_0^{\beta_s})^2 - \!\!\!\!\!\!\sum_{\beta_m, \, t_i^{\beta_m} \geq 0}\!\!\!\!N^{\beta_m}\tilde{t}_i^{\beta_m}(m_0^{\beta_m})^2 -\!\!\!\!\!\!\!\!\!\!\!\!\!\!\! \sum_{\beta_m, \, t_i^{\beta_m} < 0, \, \tilde{t}_i^{\beta_m} \geq 0} \!\!\!\!\!\!\!\!\!\!\!\!\!\! N^{\beta_m}\tilde{t}_i^{\beta_m}(m_0^{\beta_m})^2 \nonumber\\
= & \sum_{\beta_s }N^{\beta_s}
\Big|\tilde{t_i}^{\beta_s}\Big|(m_0^{\beta_s})^2 - \sum_{\beta_m, \,+}N^{\beta_m}\tilde{t}_i^{\beta_m}(m_0^{\beta_m})^2\, ,
\end{align}
where in the first inequality, we 
split terms in the sum of \eqref{eq:T_ILowerbounds} into positive and negative contributions, as indicated in the summation by $t_i^{\beta_m} \geq 0$ and $t_i^{\beta_m} < 0$. In 
the second inequality, in the first term, we only kept those special 
branes in the sum that contribute positively to the 
$0^{\text{th}}$-tadpole, which are labelled by 
$\beta_s$. In the second line, we added 
three more negative terms and 
in the next equality, we combined them into three sums using 
\eqref{eq:tildet_i}, that yield the two sums  in the last 
line.
\end{proof}

\begin{lemma}\label{lemma:1/6}
	For any pair of a special brane that contributes positively to the $0^{\text{th}}$-tadpole and a mixed brane, there exists an index $k$ such that $\tilde{t}_k^{\beta_m}$ is strictly negative and $|\tilde{t}_k^{\beta_m}|>|\tilde{t}_k^{\beta_s}|$. In particular 
	\beq \label{eq:1/6}
		|\tilde{t}_k^{\beta_m}|-|\tilde{t}_k^{\beta_s}|\geq \frac13\,.
	\eeq	
\end{lemma}

\begin{proof}
Suppose the converse is true, i.e.~for some pair of a special brane 
that contributes positively to the $0^{\text{th}}$-tadpole and a mixed 
brane, there does not exist an index $k$ such that 
$\tilde{t}_k^{\beta_m}$ is strictly negative and 
$|\tilde{t}_k^{\beta_m}|>|\tilde{t}_k^{\beta_s}|$. Then, consider 
the difference of the SUSY equality \eqref{betaSUSYeqWithM} for 
the mixed brane and for the special brane:
\bea\label{LHSdif=RHSdif}
&& \text{LHS of \eqref{betaSUSYeqWithM} for the mixed brane - LHS of \eqref{betaSUSYeqWithM} for the special brane} \nonumber\\
&= &
\text{RHS of \eqref{betaSUSYeqWithM} for the mixed brane - RHS of \eqref{betaSUSYeqWithM} for the special brane}
\eea

We will show that \eqref{LHSdif=RHSdif} will be violated. To simplify 
our notation, we will in the following denote the difference of the LHS 
and RHS in \eqref{LHSdif=RHSdif} by $\Delta_{\text{LHS}}$ and 
$\Delta_{\text{RHS}}$, respectively. First 
consider the difference $\Delta_{\text{LHS}}$. The first term, 
$3C\left(\frac{j}{j_0}, \frac{j}{j_0}\right)j_0^2$, is the same for 
both branes. Thus, by expanding everything out and using \eqref{eq:Ti} 
and \eqref{eq:tildet_i}, we obtain
\bea \label{eq:DeltaLHS}
\Delta_{\text{LHS}}&=&  6C\left(\frac{b}{2}+\frac{j}{j_0}, b+\frac{m^{\beta_m}}{m^{\beta_m}_0}\right)j_0^2- 6C\left(\frac{b}{2}+\frac{j}{j_0}, b+\frac{m^{\beta_s}}{m^{\beta_s}_0}\right)j_0^2 \nonumber\\
&=&  3\sum_{i=1}^p\tilde{t}_i^{\beta_m}\left(\frac{b_i}{2}+\frac{j_i}{j_0}\right)j_0^2-3\sum_{i=1}^p\tilde{t}_i^{\beta_s}\left(\frac{b_i}{2}+\frac{j_i}{j_0}\right)j_0^2\,.
\eea

By Corollary \ref{cor:T0+then(C00i+ti)-}, since the special brane 
contributes positively to the $0^{\text{th}}$-tadpole, 
$\tilde{t}_i^{\beta_s}<0$ for all $i$. Also notice that the mixed brane 
must have at least one $i$ for which $\tilde{t}_i^{\beta_m} > 0$, 
because by definition, a mixed brane must have at least one $i$ for 
which $t_i^{\beta_m} \geq 0$, and for this $i$, by \eqref{eq:tildet_i} 
and the positivity of $\mathcal{K}_{00i}$, $\tilde{t}_i^{\beta_m} > 0$. 
Labelling those $i$ for which $\tilde{t}_i^{\beta_m} > 0$ as $i+$, and 
those $i$ for which $\tilde{t}_i^{\beta_m} \leq 0$ as $i-$, 
\eqref{eq:DeltaLHS} becomes
\begin{align} \label{eq:LHSSUSYDifference}
\Delta_{\text{LHS}}= & 3j_0^2\Big[\sum_{i+}\tilde{t}_{i+}^{\beta_m}\Big(\frac{b_{i+}}{2}+\frac{j_{i+}}{j_0}\Big)
+ \sum_{i-}(|\tilde{t}_{i-}^{\beta_s}|-|\tilde{t}_{i-}^{\beta_m}|)\Big(\frac{b_{i-}}{2}+\frac{j_{i-}}{j_0}\Big)  + \sum_{i+}|\tilde{t}_{i+}^{\beta_s}|\Big(\frac{b_{i+}}{2}+\frac{j_{i+}}{j_0}\Big)\Big]\nonumber\\
\geq & 3j_0^2\left[\sum_{i+}\tilde{t}_{i+}^{\beta_m}\left(\frac{b_{i+}}{2}+\frac{j_{i+}}{j_0}\right) 
+ \sum_{i+}|\tilde{t}_{i+}^{\beta_s}|\left(\frac{b_{i+}}{2}+\frac{j_{i+}}{j_0}\right)\right]\,.
\end{align}
Here in the last step we dropped the second sum, which is positive, because by assumption there does not exist an index $k$ such that $\tilde{t}_k^{\beta_m}$ is strictly negative and $|\tilde{t}_k^{\beta_m}|>|\tilde{t}_k^{\beta_s}|$.

Notice, by \eqref{eq:Ti}, $\tilde{t}_i^{\beta_s}$, $\tilde{t}_i^{\beta_m}$ are rational numbers 
$\frac{a}{m_0^{\beta_s}\rule{0pt}{10pt}}$, $\frac{b}{m_0^{\beta_m}\rule{0pt}{10pt}}$ with $a$, $b \in 2\mathbb{Z}$.\footnote{By \eqref{eq:Ti}, we have $\tilde{t}_k^{\beta_s}=2\mathcal{K}_{00k}+2\sum_j C_{kj}\frac{m_j^{\beta_s}}{m_0^{\beta_s}}= \frac{a}{m_0^{\beta_s}}$, 
$\tilde{t}_k^{\beta_m}=2\mathcal{K}_{00k}+2\sum_j C_{kj}\frac{m_j^{\beta_m}}{m_0^{\beta_m}}= \frac{b}{m_0^{\beta_m}}$ for $a,b \in 2\mathbb{Z}$.} By Table \ref{tab:m0bounds}, we have 
$|m_0^{\beta_m}|\leq3$. For the special brane, since $\tilde{t}_i^{\beta_s}<0$ for all $i$, we have 
\beq 
	0>\frac12\sum_{i=1}^p\tilde{t}_i^{\beta_s}\left(\frac{b_i}{2}+\frac{j_i}{j_0}\right)=C\left(\frac{b}{2}+\frac{j}{j_0}, b+\frac{m^{\beta_s}}{m^{\beta_s}_0}\right)\,.
\eeq
Thus, the bound $|m_0^{\beta_s}|\leq 3$ in the third column of Table \ref{tab:m0bounds} applies. This implies both $|\tilde{t}_i^{\beta_s}|>0$, 
$\tilde{t}_{i+}^{\beta_m}>0$ are either integers or a third of integers: 
\begin{equation}
3\tilde{t}_{i+}^{\beta_m} > 1\,, \qquad 3|\tilde{t}_{i}^{\beta_s}|> 1\,. 
\end{equation}
Hence, \eqref{eq:LHSSUSYDifference} becomes 
\bea\label{LHSDifLarge}
\Delta_{\text{LHS}}&\geq & \sum_{i+}3\tilde{t}_{i+}^{\beta_m}\left(\frac{b_{i+}}{2}+\frac{j_{i+}}{j_0}\right) j_0^2
+ \sum_{i+}3|\tilde{t}_{i+}^{\beta_s}|\left(\frac{b_{i+}}{2}+\frac{j_{i+}}{j_0}\right)j_0^2\geq  2\sum_{i+}\left(\frac{b_{i+}}{2}+\frac{j_{i+}}{j_0}\right) j_0^2
\nonumber\\
&\geq & 2\sum_{i+}j_{i+}j_0\,,
\eea
where in the last step we dropped the  term containing the positive 
$b_i$. We have discussed that at least one index 
$i+$ exists. With $j_I\gg 1$ for all $I$, \eqref{LHSDifLarge} shows 
that the difference between the LHS of \eqref{betaSUSYeqWithM} for the 
two branes is large.

Next, we show that the difference $\Delta_{\text{RHS}}$ between the RHS 
of \eqref{betaSUSYeqWithM} for the two branes is much smaller.
Starting from the RHS of \eqref{T_0WithM} for the special brane we
note the identity
\begin{align} \label{RHS+ive}
& 3\left[\tfrac{1}{2}x_{\{i,k\}}\tilde{t}_i^{\beta}\tilde{t}_k^{\beta}-M_{\{i,k\}}\left(b+\tfrac{m^{\beta}}{m^{\beta}_0}, 
b+\tfrac{m^{\beta}}{m^{\beta}_0}\right)\right](m_0^{\beta})^2
\\
= & \Big[\tfrac{1}{4}\mathcal{K}_{000}+\tfrac{3}{2}x_{\{i,k\}}t_i^{\beta}t_k^{\beta}-3M_{\{i,k\}}\Big(\tfrac{b}{2}
+\tfrac{m^{\beta}}{m^{\beta}_0}, \tfrac{b}{2}+\tfrac{m^{\beta}}{m^{\beta}_0}\Big)\Big](m^{\beta}_0)^2
\!+ \tfrac{3}{2}\sum_{i=1}^pb_i\tilde{t}_i^{\beta}(m_0^{\beta})^2 - \mathcal{K}_{000}(m_0^{\beta})^2\,.\nonumber
\end{align}
Since the special brane contributes positively to the $0^{\text{th}}$-tadpole, the LHS of (\ref{RHS+ive}) is positive. We also recall that $\tilde{t}_i^{\beta}<0$ for all $i$
by Corollary \ref{cor:T0+then(C00i+ti)-}, which implies that the  second last term on
the RHS of (\ref{RHS+ive}) is strictly negative, as $b_i\geq 0$. In addition, the last term on
the RHS is always negative for the bases $B$ we consider. 
Thus, the term in square brackets on the RHS of \eqref{RHS+ive}, which 
is the RHS of (\ref{betaSUSYeqWithM}), must be strictly positive. In 
particular, it must have a bigger magnitude than that of (the next to 
last term and) the last term:
\begin{equation*}
\text{RHS of \eqref{betaSUSYeqWithM} for the special brane}>\mathcal{K}_{000}(m_0^{\beta_s})^2.
\end{equation*}

Next, consider the RHS of \eqref{betaSUSYeqWithM} for the mixed brane. Since it is a mixed brane, we can pick a pair of $t_i^{\beta_m}$, $t_k^{\beta_m}$ of opposite signs to make the second term of the RHS of \eqref{betaSUSYeqWithM} negative. By the positive semi-definiteness of the $M$-matrix, the third term of the RHS of \eqref{betaSUSYeqWithM} is always negative. Thus 
\begin{equation*}
\text{RHS of \eqref{betaSUSYeqWithM} for the mixed brane}\leq \frac{1}{4}\mathcal{K}_{000}(m_0^{\beta_m})^2.
\end{equation*}
Hence, we obtain, using again the bounds on $m_0^\beta$ from Table 
\ref{tab:m0bounds},
\beq\label{RHSDifSmall}
\Delta_{\text{RHS}}<  \frac{1}{4}\mathcal{K}_{000}(m_0^{\beta_m})^2 - \mathcal{K}_{000}(m_0^{\beta_s})^2
\leq  \frac{1}{4}\mathcal{K}_{000}(3)^2 - \mathcal{K}_{000}(1)^2 
=  \frac{5}{4}\mathcal{K}_{000}\,.
\eeq
By comparison of \eqref{LHSDifLarge} and \eqref{RHSDifSmall}, using the 
property $j_I\gg 1$ for all $I$, we see that we will always have 
\beq
\Delta_{\text{LHS}}> \Delta_{\text{RHS}}\,,
\eeq
which clearly violates  \eqref{LHSdif=RHSdif}. 

Finally we prove \eqref{eq:1/6}. Recall both $\tilde{t}_k^{\beta_s}$, 
$\tilde{t}_k^{\beta_m}$ are either integers or a third 
of an integer. Since $|\tilde{t}_k^{\beta_m}|>|\tilde{t}_k^{\beta_s}|$, 
their difference is at least a non-zero integer divided by their common 
denominator, which is $3$, i.e.~\eqref{eq:1/6} applies. 
\end{proof}

We make two useful definitions for the next lemma before stating it. Recall that the contribution 
of a mixed brane to the $0^{\text{th}}$-tadpole is negative, cf.~Table 
\ref{tab:TadpoleContributions}, and is 
given by the first term in \eqref{T_0WithM}: 
\begin{equation}\label{MixedContributiontoT0}
0 \geq R^{\beta_m} \equiv N^{\beta_m}\left[\frac{1}{2}x_{\{i,k\}}\tilde{t}_i^{\beta_m}\tilde{t}_k^{\beta_m}-M_{\{i,k\}}\left(b+\frac{m^{\beta_m}}{m^{\beta_m}_0}, b+\frac{m^{\beta_m}}{m^{\beta_m}_0}\right)\right](m_0^{\beta_m})^2\,.
\end{equation} 
Similarly, for a special brane that contributes positively to the 
$0^{\text{th}}$-tadpole, its contribution is also given by the first 
term in \eqref{T_0WithM}: 
\begin{equation}\label{SpecialContributiontoT0}
0 \leq S^{\beta_s} \equiv N^{\beta_s} \left[\frac{1}{2}x_{\{i,k\}}\tilde{t}_i^{\beta_s}\tilde{t}_k^{\beta_s}-M_{\{i,k\}}\left(b+\frac{m^{\beta_s}}{m^{\beta_s}_0}, b+\frac{m^{\beta_s}}{m^{\beta_s}_0}\right)\right](m_0^{\beta_s})^2\,. 
\end{equation}
Thus, the total positive contribution to the $0^{\text{th}}$-tadpole, 
and part of the total negative contributions to the $0^{\text{th}}$-
tadpole from mixed branes with $\tilde{t}_i^{\beta_m} \geq 
0$ are
\begin{equation}
\sum_{\beta_s}S^{\beta_s}\geq 0\,, \qquad\qquad \sum_{\beta_m, \, +} R^{\beta_m}\leq 0\,.
\end{equation}

\begin{lemma}\label{lemma:1/6Difference}
Given $h_1,h_2 \in \mathbb{Q}^+$, $0< h_1,h_2 \leq 1$, so that for some index $i$
\begin{equation}\label{h1h2}
h_1 \sum_{\beta_m, \, +}N^{\beta_m}\tilde{t}_i^{\beta_m}(m_0^{\beta_m})^2\ = h_2 \sum_{\beta_s }N^{\beta_s}
\Big|\tilde{t_i}^{\beta_s}\Big|(m_0^{\beta_s})^2 
\end{equation}
holds, then $h_1\sum_{\beta_m, \, +} |R^{\beta_m}| > h_2\sum_{\beta_s}S^{\beta_s}$. In particular,
\begin{equation}
h_1\sum_{\beta_m, \, +} |R^{\beta_m}| - h_2\sum_{\beta_s}S^{\beta_s} \geq \frac{1}{6} \text{\normalfont{min}}\Big(x_{\{i,k\}}\Big) h_1 \sum_{\beta_m, +}  N^{\beta_m} \tilde{t}_i^{\beta_m} (m_0^{\beta_m})^2\,,
\end{equation}
where the minimum and maximum is taken over all pairs $\{i,k\}$ of indices of K\"ahler cone generators in the subcone, but can also be taken 
	across the entire K\"ahler cone.
\end{lemma}

\begin{proof}
We introduce a partition of unity
$\{f^{\beta_s}\}_{\beta_s}$, and, for every index $\beta_s$, a partition of unity $\{g^{\beta_s,\beta_m}\}_{\beta_m,\,+}$,\footnote{We emphasize that the index $\beta_m$  on $g^{\beta_s,\beta_m}$   is only limited to mixed branes with $\tilde{t}_i^{\beta_m}\geq 0$.} i.e.
\beq\label{eq:fPartition}
\sum_{\beta_s}f^{\beta_s}=1\,, \qquad \sum_{\beta_m,+} g^{\beta_s,\beta_m}=1\,,  \qquad f^{\beta_s},\,\, g^{\beta_s,\beta_m} \in \mathbb{Q}^+\,, \qquad 0 < f^{\beta_s}, g^{\beta_s,\beta_m} \leq 1\,,
\eeq
defined by the property
\beq\label{eq:fPartition1}
f^{\beta_s}h_1N^{\beta_m}\tilde{t}_i^{\beta_m}(m_0^{\beta_m})^2
 = g^{\beta_s,\beta_m}h_2N^{\beta_s} \left\lvert \tilde{t}_i^{\beta_s} \right\lvert (m_0^{\beta_s})^2
\,. 
\eeq

Inserting unity as  $1=\sum_{\beta_s} f^{\beta_s}=\sum_{\beta_m,+}, g^{\beta_s,\beta_m}$,  we obtain the obvious identity 
\begin{align}\label{eq:1/6DifferencePart1}
& h_1\sum_{\beta_m, +}|R^{\beta_m}| - h_2 \sum_{\beta_s}S^{\beta_s} =  \Big(\sum_{\beta_s} f^{\beta_s} \Big) h_1 \sum_{\beta_m, +}|R^{\beta_m}| - h_2 \sum_{\beta_s}\Big(\sum_{\beta_m,+} g^{\beta_s,\beta_m} \Big) S^{\beta_s}\nonumber\\
= & \sum_{\beta_s} \sum_{\beta_m, +} \Big(f^{\beta_s} h_1 |R^{\beta_m}| - g^{\beta_s,\beta_m} h_2 S^{\beta_s} \Big)\,.
\end{align}
For each summand in the sum of the last line of \eqref{eq:1/6DifferencePart1}, we have 
\begin{align}\label{eq:ChoiceOfik}
f^{\beta_s} h_1 |R^{\beta_m}\!| \!-\! g^{\beta_s,\beta_m} h_2 S^{\beta_s} 
\!\!=\! & f^{\beta_s} h_1 N^{\beta_m} \Big\lvert \tfrac{1}{2}x_{\{i,k\}} \underbrace{\tilde{t}_i^{\beta_m}}_{\geq 0} \underbrace{\tilde{t}_k^{\beta_m}}_{<0}-M_{\{i,k\}}\Big(b+\tfrac{m^{\beta_m}}{m^{\beta_m}_0}, b+\tfrac{m^{\beta_m}}{m^{\beta_m}_0}\Big) \Big\lvert (m_0^{\beta_m})^2 \nonumber\\
- g^{\beta_s,\beta_m}&  h_2 N^{\beta_s} \Big[\frac{1}{2}x_{\{i,k\}} \underbrace{\tilde{t}_i^{\beta_s}}_{<0} \underbrace{\tilde{t}_k^{\beta_s}}_{<0}-M_{\{i,k\}}\Big(b+\tfrac{m^{\beta_s}}{m^{\beta_s}_0}, b+\tfrac{m^{\beta_s}}{m^{\beta_s}_0}\Big)\Big](m_0^{\beta_s})^2.
\end{align}
Here, the pair $\{i,k\}$ is chosen so that the index $i$ is the one for which \eqref{h1h2} holds, and the index $k$ is chosen such that the inequality $|\tilde{t}_k^{\beta_m}|-|\tilde{t}_k^{\beta_s}|\geq \frac13$ of Lemma \ref{lemma:1/6} holds for the pair $(\beta_s,\beta_m)$ of special and mixed brane in \eqref{eq:ChoiceOfik}. We emphasize that the choice of this index $k$ depends on the brane pair $(\beta_s,\beta_m)$ and thus might be different for each summand in \eqref{eq:1/6DifferencePart1}. Next, we drop the positive semi-definite $M$-matrix terms in \eqref{eq:ChoiceOfik} to get
\begin{align}\label{eq:EachSummandDone}
f^{\beta_s} h_1 |R^{\beta_m}|\! - \! g^{\beta_s,\beta_m} h_2 S^{\beta_s} 
\!\!\geq & \frac{1}{2}x_{\{i,k\}}\!\! \Big[\! f^{\beta_s} h_1 N^{\beta_m} \tilde{t}_i^{\beta_m} (m_0^{\beta_m})^2 \Big\lvert \tilde{t}_k^{\beta_m}\Big\lvert 
\!-\!  g^{\beta_s,\beta_m} h_2 N^{\beta_s} \Big\lvert \tilde{t}_i^{\beta_s} \Big\lvert (m_0^{\beta_s})^2  \Big\lvert \tilde{t}_k^{\beta_s}\Big\lvert \Big] \nonumber\\
= & \frac{1}{2}x_{\{i,k\}} \left( f^{\beta_s} h_1 N^{\beta_m} \tilde{t}_i^{\beta_m} (m_0^{\beta_m})^2 \right) \underbrace{\left(\Big\lvert \tilde{t}_k^{\beta_m}\Big\lvert -  \Big\lvert \tilde{t}_k^{\beta_s}\Big\lvert \right)}_{\geq 1/3} \nonumber\\
\geq & \frac{1}{6}\text{min}(x_{\{i,k\}}) f^{\beta_s} h_1 N^{\beta_m} \tilde{t}_i^{\beta_m} (m_0^{\beta_m})^2\,, 
\end{align}
where we have used that the coefficients of $|\tilde{t}_k^{\beta_m}|$,
$|\tilde{t}_k^{\beta_s}|$ in the first line are equal by  \eqref{eq:fPartition1}. In addition,
we have removed the aforementioned implicit dependence of the index $k$ on $(\beta_s,\beta_m)$ by taking the minimum over all $\{i,k\}$.

Thus, plugging \eqref{eq:EachSummandDone} into \eqref{eq:1/6DifferencePart1} we obtain 
\beq
h_1\sum_{\beta_m, +}|R^{\beta_m}| - h_2\sum_{\beta_s} S^{\beta_s} 
\geq \frac{1}{6} \text{min}\Big(x_{\{i,k\}}\Big) h_1 \sum_{\beta_m, +}  N^{\beta_m} \tilde{t}_i^{\beta_m} (m_0^{\beta_m})^2\,,
\eeq
where we performed the sum over $\beta_s$ and used  $\sum_{\beta_s} 
f^{\beta_s}=1$, cf.~\eqref{eq:fPartition}.
\end{proof}

Now we are finally ready to show that every $t_i^{\beta}$ has an upper 
bound.
\begin{theorem}\label{thm:t_iBound}
	For all $i$ and $\beta$, $t^\beta_i$ are bounded from above as
	\beq \label{eq:tbetabound}
		t_i^{\beta}\leq \frac{6T_0+3 T_i\cdot \text{\normalfont max}
		(x_{\{i,k\}})\cdot \text{\normalfont{max}}(T_l)}{
		\text{\normalfont min}(x_{\{i,k\}})}\,,
	\eeq
	where the minimum and maximum is taken over all pairs $\{i,k\}$ of indices of K\"ahler cone generators in the subcone, but can also be taken 
	across the entire K\"ahler cone.
\end{theorem}

\begin{proof}
We derive the above bound for $t_i^{\beta}$ for an arbitrary index $i$. 
By Lemma \ref{lemma:ShiftedTadpole}, we either have $\sum_{\beta_s }N^{\beta_s}
\Big|\tilde{t_i}^{\beta_s}\Big|(m_0^{\beta_s})^2 \leq \sum_{\beta_m, \, +}N^{\beta_m}\tilde{t}_i^{\beta_m}(m_0^{\beta_m})^2$, or $\sum_{\beta_m, \, +}N^{\beta_m}\tilde{t}_i^{\beta_m}(m_0^{\beta_m})^2 < \sum_{\beta_s }N^{\beta_s}
\Big|\tilde{t_i}^{\beta_s}\Big|(m_0^{\beta_s})^2 < \sum_{\beta_m, \, +}N^{\beta_m}\tilde{t}_i^{\beta_m}(m_0^{\beta_m})^2 + T_i$. We consider each case separately:

\noindent \textbf{Case 1:} $\sum_{\beta_s }N^{\beta_s}
\Big|\tilde{t_i}^{\beta_s}\Big|(m_0^{\beta_s})^2 \leq \sum_{\beta_m, \, +}N^{\beta_m}\tilde{t}_i^{\beta_m}(m_0^{\beta_m})^2$.

In other words, we have a relation as in \eqref{h1h2} with $h_1\leq 1$, 
$h_2=1$,
\begin{equation}\label{eq:h1Definition}
h_1 \sum_{\beta_m, \, +}N^{\beta_m}\tilde{t}_i^{\beta_m}(m_0^{\beta_m})^2\ = \sum_{\beta_s }N^{\beta_s}
\Big|\tilde{t_i}^{\beta_s}\Big|(m_0^{\beta_s})^2\,.
\end{equation}
Starting with the first inequality in \eqref{eq:T_ILowerbounds} and 
employing \eqref{MixedContributiontoT0},
\eqref{SpecialContributiontoT0}, we obtain
\begin{align}\label{1/6DifferenceSpecialCase}
T_0 \geq & \sum_{\beta_m, +}|R^{\beta_m}| - \sum_{\beta_s} S^{\beta_s}
=(1-h_1)\sum_{\beta_m, +}|R^{\beta_m}|+h_1\sum_{\beta_m, +}|R^{\beta_m}| - \sum_{\beta_s} S^{\beta_s} \nonumber\\
\geq & (1-h_1)\sum_{\beta_m, +}|R^{\beta_m}|+\frac{1}{6} \text{min}(x_{\{i,k\}}) h_1 \sum_{\beta_m, +}  N^{\beta_m} \tilde{t}_i^{\beta_m} (m_0^{\beta_m})^2 \nonumber\\
\geq&(1-h_1)\sum_{\beta_m,\,+}N^{\beta_m}\frac{1}{2}x_{\{i,k\}} \tilde{t}_i^{\beta_m} \underbrace{|\tilde{t}_k^{\beta_m}|}_{\geq 1/3}(m_0^{\beta_m})^2+\frac{1}{6} \text{min}(x_{\{i,k\}}) h_1 \sum_{\beta_m, +}  N^{\beta_m} \tilde{t}_i^{\beta_m} (m_0^{\beta_m})^2\nonumber\\
\geq&(1-h_1)\frac{1}{6}\text{min}(x_{\{i,k\}})\sum_{\beta_m,\,+}N^{\beta_m} \tilde{t}_i^{\beta_m} (m_0^{\beta_m})^2+h_1 \frac{1}{6} \text{min}(x_{\{i,k\}}) \sum_{\beta_m, +}  N^{\beta_m} \tilde{t}_i^{\beta_m} (m_0^{\beta_m})^2 \nonumber\\
=&\frac{1}{6}\text{min}(x_{\{i,k\}})\sum_{\beta_m,\,+}N^{\beta_m} \tilde{t}_i^{\beta_m} (m_0^{\beta_m})^2\,,
\end{align}
where in the first inequality we only kept negative contributions to the 
$0^{\text{th}}$-tadpole from mixed branes with 
$\tilde{t}_i^{\beta_m}\geq 0$ (see Table \ref{tab:TadpoleContributions}). 
In the second line we used Lemma \ref{lemma:1/6Difference}. In the 
third line we plugged in the definition \eqref{MixedContributiontoT0} 
of $R^{\beta_m}$, where we picked our choice of the pair 
$\{i,k\}$ so that $i$ is the same 
index $i$ that we want to derive a bound for $t_i^{\beta}$, and  
$k$ such that $|\tilde{t}_k^{\beta_m}|\geq \frac{1}{3}$ \footnote{Indeed, 
since a non-zero $\tilde{t}_k^{\beta_m}$ is at least a third of an 
integer, we only have to argue that a $k$ with a non-zero $\tilde{t}_k^{\beta_m}$ exists. But this is true since otherwise  $C\left(\frac{b}{2}+\frac{j}{j_0}, b+\frac{m^{\beta_m}}{m^{\beta_m}_0}\right)=\frac12\sum_{i=1}^p\tilde{t}_i^{\beta_m}\left(\frac{b_i}{2}+\frac{j_i}{j_0}\right)\geq 0$, which by Lemma \ref{lem:largej} 
implies that this brane would be a special, not a mixed brane.}, and dropped the M-matrix term. The remaining two lines of \eqref{1/6DifferenceSpecialCase} are just algebra. Thus, we have the following bound on $t_i$:
\begin{align}\label{eq:SpecialCaseBound}
T_0 \geq & \frac{1}{6} \text{min}(x_{\{i,k\}}) \sum_{\beta_m,+}N^{\beta_m} \tilde{t}_i^{\beta_m} (m_0^{\beta_m})^2, 
 \,\,\,   \Longrightarrow \,\,\,  t_i^{\beta} < \sum_{\beta_m,+}N^{\beta_m} \tilde{t}_i^{\beta_m} (m_0^{\beta_m})^2 \leq \frac{6T_0}{\text{min}(x_{\{i,k\}})}\,.
\end{align}

\noindent \textbf{Case 2:} $\sum_{\beta_m, \, +}N^{\beta_m}\tilde{t}_i^{\beta_m}(m_0^{\beta_m})^2 < \sum_{\beta_s }N^{\beta_s}
\Big|\tilde{t_i}^{\beta_s}\Big|(m_0^{\beta_s})^2 < \sum_{\beta_m, \, +}N^{\beta_m}\tilde{t}_i^{\beta_m}(m_0^{\beta_m})^2 + T_i$. 

In this case we are in a special case of \eqref{h1h2} with $h_1=1$, $h_2\leq 1$ and
\begin{equation}\label{eq:ExcessSpecial<Ti}
\sum_{\beta_m, \, +}N^{\beta_m}\tilde{t}_i^{\beta_m}(m_0^{\beta_m})^2\ = h_2 \sum_{\beta_s }N^{\beta_s}
\Big|\tilde{t_i}^{\beta_s}\Big|(m_0^{\beta_s})^2 \qquad (1-h_2) \sum_{\beta_s} N^{\beta_s} \Big\lvert \tilde{t}_i^{\beta_s}\Big\lvert (m_0^{\beta_s})^2 < T_i \,.
\end{equation}
Analogous to \eqref{1/6DifferenceSpecialCase} of Case 1, we obtain
\begin{align}\label{1/6DifferenceGeneralCase}
T_0 \geq & \sum_{\beta_m, +}|R^{\beta_m}| - \sum_{\beta_s} S^{\beta_s}
=\sum_{\beta_m, +}|R^{\beta_m}|- h_2\sum_{\beta_s} S^{\beta_s} - (1-h_2)\sum_{\beta_s} S^{\beta_s} \nonumber\\
\geq & \frac{1}{6} \text{min}\Big(x_{\{i,k\}}\Big) \sum_{\beta_m, +}  N^{\beta_m} \tilde{t}_i^{\beta_m} (m_0^{\beta_m})^2 - (1-h_2)\sum_{\beta_s} S^{\beta_s}\,.
\end{align}

We digress to consider the following inequality:
\begin{align} \label{eq:Digression}
& \sum_{i}^p \sum_{\beta_s} N^{\beta_s} \tilde{t_i}^{\beta_s} (m_0^{\beta_s})^2\Big( \frac{b_i}{2}+\frac{j_i}{j_0} \Big) = 
\sum_{\beta_s} \sum_{i}^p N^{\beta_s} \Big(t_i^{\beta_s}+\mathcal{K}_{00i}\Big) (m_0^{\beta_s})^2\Big( \frac{b_i}{2}+\frac{j_i}{j_0} \Big) \\
\geq & \sum_{\beta_s} \sum_{i}^p N^{\beta_s} t_i^{\beta_s} (m_0^{\beta_s})^2\Big( \frac{b_i}{2}+\frac{j_i}{j_0} \Big) 
\geq \sum_{\beta} \sum_{i}^p N^{\beta} t_i^{\beta} (m_0^{\beta})^2\Big( \frac{b_i}{2}+\frac{j_i}{j_0} \Big) 
\geq \sum_{i}^p (-T_i)\Big( \frac{b_i}{2}+\frac{j_i}{j_0} \Big)\, ,\nonumber
\end{align}
where in the first equality we used \eqref{eq:tildet_i}, in the second 
inequality, we extended the sum across $\beta_s$ to the sum across all 
$\beta$ because each summand is negative by 
\eqref{eq:StarWithLowerBound}, and in the last inequality we used 
\eqref{eq:StarWithLowerBoundPart1}. Comparing coefficients of 
$\frac{b_i}{2}+\frac{j_i}{j_0}$ between the first and last term in 
\eqref{eq:Digression}, we note that there has to exist an index $k$ such that
\begin{equation}\label{eq:IndexkForExcessSpecialBrane}
T_k \geq \sum_{\beta_s} N^{\beta_s} \Big\lvert \tilde{t_k}^{\beta_s} \Big\lvert (m_0^{\beta_s})^2 \geq \Big\lvert \tilde{t_k}^{\beta_s} \Big\lvert\, .
\end{equation} 

If the index $i$ for which we want to bound $t_i^{\beta}$ coincides with such an index $k$, we have an obvious bound on $t_i^{\beta}$
\begin{equation}
T_i \geq \sum_{\beta_s} N^{\beta_s} \Big\lvert \tilde{t_i}^{\beta_s} \Big\lvert (m_0^{\beta_s})^2 > \sum_{\beta_m,+}N^{\beta_m}\tilde{t}_i^{\beta_m}(m_0^{\beta_m})^2 > t_i^{\beta}\,,
\end{equation}
where in the first inequality, we used \eqref{eq:IndexkForExcessSpecialBrane} with $k=i$, and in the second inequality we used the assumption that $\sum_{\beta_s}N^{\beta_s} \left\lvert \tilde{t}_i^{\beta_s} \right\lvert (m_0^{\beta_s})^2 > \sum_{\beta_m,+}N^{\beta_m}\tilde{t}_i^{\beta_m}(m_0^{\beta_m})^2$. 

Thus we only need to consider $i\neq k$ with $k$ satisfying \eqref{eq:IndexkForExcessSpecialBrane}. Then, the last term on the second line of \eqref{1/6DifferenceGeneralCase} becomes
\begin{align}\label{Bound2Part1}
& (1-h_2)\sum_{\beta_s}S^{\beta_s} 
\leq (1-h_2) \sum_{\beta_s} N^{\beta_s}\frac{1}{2}x_{\{i,k\}} \Big\lvert \tilde{t}_i^{\beta_s}\Big\lvert (m_0^{\beta_s})^2 \Big\lvert \tilde{t}_k^{\beta_s}\Big\lvert \nonumber\\
\leq & \frac{1}{2}x_{\{i,k\}} \underbrace{(1-h_2) \sum_{\beta_s} N^{\beta_s} \Big\lvert \tilde{t}_i^{\beta_s}\Big\lvert (m_0^{\beta_s})^2}_{< T_i} \cdot {T_k} 
< \frac{1}{2}\text{max}(x_{\{i,k\}})\cdot T_i \cdot T_k
\end{align}
where in the first inequality we plugged in the definition 
\eqref{SpecialContributiontoT0} of $S^{\beta_s}$ and picked the pair 
$\{i,k\}$ such that $i$  is the index for which we want to show 
boundedness for $t_i^{\beta}$, $k$ is the index such that 
\eqref{eq:IndexkForExcessSpecialBrane} is satisfied and dropped the 
negative $M$-matrix term. In the second inequality we used \eqref{eq:IndexkForExcessSpecialBrane}   
for $\tilde{t}_k^{\beta_s}$, as well as the second inequality in \eqref{eq:ExcessSpecial<Ti}. Combining \eqref{1/6DifferenceGeneralCase} and \eqref{Bound2Part1},  we obtain
\beq\label{tiFinalBoundPart1}
T_0 >\sum_{\beta_m, +}|R^{\beta_m}| - \sum_{\beta_s}S^{\beta_s} 
\geq \frac{1}{6} \text{min}(x_{\{i,k\}}) \sum_{\beta_m,+}N^{\beta_m} \tilde{t}_i^{\beta_m} (m_0^{\beta_m})^2 - \frac{1}{2}\text{max}\big(x_{\{i,k\}}\big)T_i \cdot \underbrace{T_k}_{\leq \text{max}(T_l)}
\eeq
and arrive at the final bound
\beq\label{tiFinalBoundPart2}
t_i^{\beta} <  \sum_{\beta_m,+}N^{\beta_m} \tilde{t}_i^{\beta_m}(m_0^{\beta_m})^2 < \frac{6T_0+3 T_i\cdot \text{max}(x_{\{i,k\}}) \cdot\text{max}(T_l)}{\text{min}(x_{\{i,k\}})}.
\eeq
\end{proof}

\subsubsection{Bounds on $n_I^{D5}$}
\label{sec:BoundsD5}

In this section, we employ the results from the previous Section \ref{sec:BoundsBeta} to derive bounds on the numbers $n_I^{\rm D5}$
of D5-branes. These bounds are formulated in two theorems.
\begin{theorem}
\label{th:boundniD5}
For all $i$ we have the following bound on $n_i^{\rm D5}$:
		\beq \label{eq:nD5ibound}
		n_i^{D5} < \frac{6T_0}{\text{\normalfont{min($x_{\{i,k\}}$)}}}+T_i\,,
	\eeq
	where the minimum is taken over all pairs $\{i,k\}$ of indices of K\"ahler cone generators in the subcone, but can also be taken 
	across the entire K\"ahler cone.
\end{theorem}

\begin{proof}
From \eqref{eq:Ti}, we obtain
\begin{equation}\label{eq:niD5titildeRelationship}
n_i^{D5} = \sum_{\beta}N^{\beta}t_i^{\beta}(m_0^{\beta})^2 + T_i < \sum_{\beta}N^{\beta}t_i^{\beta}(m_0^{\beta})^2 + T_i + \sum_{\beta}N^{\beta}\mathcal{K}_{00i}(m_0^{\beta})^2 =  \sum_{\beta}N^{\beta}\tilde{t}_i^{\beta}(m_0^{\beta})^2 + T_i\,,
\end{equation}
where in the last equality we used \eqref{eq:tildet_i}. If $\sum_{\beta}N^{\beta}\tilde{t}_i^{\beta}(m_0^{\beta})^2 \leq 0$, then we have the obvious bound $n_i^{D5}<T_i$. Conversely if $0<\sum_{\beta}N^{\beta}\tilde{t}_i^{\beta}(m_0^{\beta})^2$, we have 
\begin{equation}\label{eq:niD5Split}
0 < \sum_{\beta}N^{\beta}\tilde{t}_i^{\beta}(m_0^{\beta})^2  \leq \sum_{\beta_m,+}N^{\beta_m}\tilde{t}_i^{\beta_m}(m_0^{\beta_m})^2 - \sum_{\beta_s}N^{\beta_s}|\tilde{t}_i^{\beta_s}|(m_0^{\beta_s})^2\,,
\end{equation}
where we dropped negative terms in the last inequality. Thus, we are in case 1 in the proof of Theorem \ref{thm:t_iBound}, i.e.~$\sum_{\beta_s}N^{\beta_s} \left\lvert \tilde{t}_i^{\beta_s} \right\lvert (m_0^{\beta_s})^2 \leq \sum_{\beta_m,+}N^{\beta_m}\tilde{t}_i^{\beta_m}(m_0^{\beta_m})^2$ , and can use results derived previously for that case. Using the fraction $h_1$ defined in \eqref{eq:h1Definition}, \eqref{eq:niD5Split} becomes
\begin{equation}\label{eq:niD5Split1}
\sum_{\beta}N^{\beta}\tilde{t}_i^{\beta}(m_0^{\beta})^2  \leq \sum_{\beta_m,+}N^{\beta_m}\tilde{t}_i^{\beta_m}(m_0^{\beta_m})^2 - \sum_{\beta_s}N^{\beta_s}|\tilde{t}_i^{\beta_s}|(m_0^{\beta_s})^2 = 
(1-h_1)\sum_{\beta_m,+}N^{\beta_m}\tilde{t}_i^{\beta_m}(m_0^{\beta_m})^2\,. 
\end{equation}
By the third line of \eqref{1/6DifferenceSpecialCase}, we obtain
\begin{align}
T_0 \geq & (1-h_1)\sum_{\beta_m,\,+}N^{\beta_m}\frac{1}{2}x^{\{i,k\}} \tilde{t}_i^{\beta_m} \underbrace{|\tilde{t}_k^{\beta_m}|}_{\geq 1/3}(m_0^{\beta_m})^2+ \frac{1}{6} \text{min}\Big(x_{\{i,k\}}\Big) h_1 \sum_{\beta_m, +}  N^{\beta_m} \tilde{t}_i^{\beta_m} (m_0^{\beta_m})^2\nonumber\\
\geq & \frac{1}{6}\text{min}(x_{\{i,k\}})(1-h_1)\sum_{\beta_m,\,+}N^{\beta_m} \tilde{t}_i^{\beta_m}(m_0^{\beta_m})^2\,
\end{align}
by dropping the second term on the RHS of the first line.
By rearranging and combining with \eqref{eq:niD5Split1}, we arrive at 
\begin{equation}
\sum_{\beta}N^{\beta}\tilde{t}_i^{\beta}(m_0^{\beta})^2 \leq (1-h)\sum_{\beta_m,\,+}N^{\beta_m} \tilde{t}_i^{\beta_m}(m_0^{\beta_m})^2 \leq \frac{6T_0}{\text{min}(x_{\{i,k\}})}\,, 
\end{equation}
which in combination with \eqref{eq:niD5titildeRelationship} gives the desired bound \eqref{eq:nD5ibound}.
\end{proof}

\begin{remark}
\normalfont{We note also, that  
the first inequality of \eqref{sumstar} forbids $(n_i^{D5}-T_i)\geq 0$ 
for all $i$, i.e.~although each $n_i^{\text{D5}}$ is bounded above by \eqref{eq:nD5ibound}, 
together they are further constrained by this condition.}
\end{remark}

\begin{theorem}	
\label{th:boundn0D5}
	We have the following bound on $n_0^{\rm D5}$:
\beq \label{eq:n0D5bound}
		n_0^{D5} \leq \frac{1}{2}\text{\normalfont{max($x_{\{i,k\}}$)}}
		\cdot\text{\normalfont{min($T_i$)}} \cdot\text{\normalfont{max}}(T_i) +T_0\,,
	\eeq
	where the minimum and maximum is taken over all pairs $\{i,k\}$ of indices of K\"ahler cone generators in the subcone. The maximum 
	can also be taken across the entire K\"ahler cone.
\end{theorem}

\begin{proof}
Using \eqref{T_0WithM}, we obtain
\begin{align}\label{eq:n0D5BoundPart1}
n_0^{D5} = & \sum_{\beta}N^{\beta}\left[\frac{1}{2}x_{\{i,k\}}\tilde{t}_i^\beta \tilde{t}_k^{\beta}-M_{\{i,k\}}\left(b+\frac{m^{\beta}}{m^{\beta}_0}, b+\frac{m^{\beta}}{m^{\beta}_0}\right)\right](m_0^{\beta})^2
+\sum_{\gamma}N^{\gamma}C(m^{\gamma},m^{\gamma})+T_0 \nonumber\\
\leq & \sum_{\beta_s}S^{\beta_s} - \sum_{\beta_m,+}|R^{\beta_m}| + T_0
\end{align}
where we dropped some negative contributions of the first term on the RHS of the first line as well the negative $\gamma$-brane 
contribution and used  $S^{\beta_s}$, $R^{\beta_m}$ as defined in \eqref{SpecialContributiontoT0}, \eqref{MixedContributiontoT0}, 
respectively. We see that the coarsest bound on $n_0^{D5}$ occurs when $\sum_{\beta_s}S^{\beta_s} - \sum_{\beta_m,+}|R^{\beta_m}|$ is maximized. By \eqref{1/6DifferenceSpecialCase}, since its last line is positive, this expression is always negative in case 1 of Theorem \ref{thm:t_iBound}. To maximize it, we look at case 2 of Theorem \ref{thm:t_iBound}. Starting from \eqref{1/6DifferenceGeneralCase} in case 2 of Theorem \ref{thm:t_iBound}, we obtain

\begin{align}\label{eq:SRbound}
& \sum_{\beta_m, +}|R^{\beta_m}| - \sum_{\beta_s} S^{\beta_s} 
\geq \frac{1}{6} \text{min}\Big(x_{\{i,k\}}\Big) \sum_{\beta_m, +}  N^{\beta_m} \tilde{t}_i^{\beta_m} (m_0^{\beta_m})^2 - (1-h_2)\sum_{\beta_s} S^{\beta_s} \nonumber\\
\geq & - (1-h_2)\sum_{\beta_s} S^{\beta_s}
\geq -(1-h_2) \sum_{\beta_s} N^{\beta_s}\frac{1}{2}x_{\{i,k\}} \Big\lvert \tilde{t}_i^{\beta_s}\Big\lvert (m_0^{\beta_s})^2 \Big\lvert \tilde{t}_k^{\beta_s}\Big\lvert \nonumber\\
\geq & -\frac{1}{2}\text{max}(x_{\{i,k\}}) \underbrace{(1-h_2) \sum_{\beta_s} N^{\beta_s} \Big\lvert \tilde{t}_i^{\beta_s}\Big\lvert (m_0^{\beta_s})^2}_{< T_i} \cdot {T_k} 
> -\frac{1}{2}\text{max}(x_{\{i,k\}})\cdot T_i \cdot T_k \nonumber\\
\geq & -\frac{1}{2}\text{max}(x_{\{i,k\}})\cdot \text{min}(T_l) \cdot \text{max}(T_l) \,,
\end{align}
where in the second inequality, we dropped the positive first term. In the third inequality, we plugged in the definition 
\eqref{SpecialContributiontoT0} of $S^{\beta_s}$ and picked the pair 
$\{i,k\}$ such that $k$ is an index so that 
\eqref{eq:IndexkForExcessSpecialBrane} is satisfied, and $i$ is the particular index such that $T_i=\text{min}(T_l)$ if this $i\neq k$. If 
$i=k$, pick any other index as $i$, and drop the $M$-matrix term. In the fourth inequality we used the second inequality in 
\eqref{eq:ExcessSpecial<Ti}. In the last inequality, we note that if we have used the first way of choosing the pair $\{i,k\}$, then 
$T_i=\text{min}(T_l)$ and $T_k\leq \text{max}(T_l)$; if we have used the second way of choosing the pair $\{i,k\}$, then 
$T_i \leq \text{max}(T_l)$ and $T_k = \text{min}(T_l)$. Combining this result with \eqref{eq:n0D5BoundPart1}, we get the desired bound 
\eqref{eq:n0D5bound} on $n_0^{D5}$.
\end{proof}

\subsubsection{Bounds on $\gamma$-branes}
\label{sec:BoundsGamma}

Finally, we derive a bound on the number of $\gamma$-brane configurations, i.e.~we bound the flux quanta $m^\gamma$.

The contribution of $\gamma$-branes to the $0^{\text{th}}$-tadpole is fixed by \eqref{T_0WithM} as
\beq \label{eq:gammaRel}
	-\sum_\gamma N^\gamma C(m^\gamma, m^\gamma)=T_0-n_0^{\rm D5}+\sum_{\beta}N^{\beta}\left[\frac{1}{2}x_{\{i,k\}}\tilde{t}_i^\beta \tilde{t}_k^{\beta}-M_{\{i,k\}}\left(b+\frac{m^{\beta}}{m^{\beta}_0}, b+\frac{m^{\beta}}{m^{\beta}_0}\right)\right](m_0^{\beta})^2\,.
\eeq
As by Proposition \ref{prop:gammaContribute-T0}, the LHS of this equation is positive, a solution to it only exists if the 
right hand side is also positive. Thus, this is the equation of an ellipsoid and the vector $m^\gamma$ of discrete flux quanta 
is given by the finite number of integral points on this ellipsoid. 
We denote the positive RHS of \eqref{eq:gammaRel} by $r^2$ with $r\in\mathbb{R}$. 

Consequently, the question of boundedness of 
$m^\gamma$ translates into showing boundedness of $r^2$. By \eqref{eq:gammaRel}
we have
\bea
	r^2\!\!&\!\!=\!\!&\!-\sum_\gamma N^\gamma C(m^\gamma, m^\gamma)\!=\!T_0\!-\!n_0^{\rm D5}+\sum_{\beta}N^{\beta}\Big[\tfrac{1}{2}x_{\{i,k\}}\tilde{t}_i^\beta \tilde{t}_k^{\beta}-M_{\{i,k\}}\Big(b+\tfrac{m^{\beta}}{m^{\beta}_0}, b+\tfrac{m^{\beta}}{m^{\beta}_0}\Big)\Big](m_0^{\beta})^2\,\nonumber\\
	&\leq &T_0+\sum_{\beta_s}S^{\beta_s}- \sum_{\beta_m,+}|R^{\beta_m}|\leq T_0+ \frac{1}{2}\text{max}\big(x_{\{i,k\}}\big) \cdot\text{min}(T_i)\cdot \text{max}(T_i) \,,
\eea
where we set $n_0^{\rm D5}=0$ and dropped some negative terms in the sum over $\beta$ to obtain the second line and used
\eqref{eq:SRbound} for the last inequality.

This argument and also Proposition \ref{prop:gammaContribute-T0} require  that the matrix $C$
is of negative signature $(0,n)$ when restricted to the subspace of vectors obeying \eqref{gammaSUSY}. As we have argued before, for the 
bases $B=\mathbb{F}_k$, $dP_n$, $n>1$ and the toric surfaces the matrix $C$ is of Minkowski signature and the vector $\frac{b}{2}+\frac{j}{j_0}$ is time-like. Thus, the above argument applies.

\section{Conclusions}
\label{sec:conclusion}

We have studied Type IIB compactifications on smooth Calabi-Yau
elliptic fibrations over almost Fano twofold bases $B$ with
magnetized D9-branes and D5-branes. We have proven that the tadpole
cancellation and SUSY conditions imply that there are only finitely
many such configurations.  We have derived explicit and calculable
bounds on all flux quanta (Table \ref{tab:m0bounds}, Theorem
\ref{thm:t_iBound}, Section \ref{sec:BoundsGamma}) as well as the number of D5-branes 
(Theorem \ref{th:boundniD5}, Theorem \ref{th:boundn0D5}), which are independent on the continuous moduli of the
compactification, in particular the K\"ahler moduli, as long as
the supergravity approximation is valid.

The presented proof applies for any geometry that meets the geometric conditions
listed at the beginning of Section \ref{sec:generic proof}. We have shown 
explicitly in Section \ref{sec:B2geometries} and Appendix \ref{app:MoriKaehlerFano} that these
geometric conditions are obeyed for the twofold bases $B$ given by the Hirzebruch surfaces 
$\mathbf{F}_k$, $k=0,1,2$, the generic del Pezzos $dP_n$, $n=0,\ldots, 8$ as 
well all toric varieties associated to the 16 reflexive two-dimensional 
polytopes. This in particular required showing the positive semi-definiteness of
the matrices $M_{\{i,k\}}$ defined in \eqref{M}. To this end we studied the 
K\"ahler cones of the generic $dP_n$ and explicitly constructed their 
K\"ahler cone generators, which are listed in Table \ref{tab:KCdPn} and reveal 
useful geometric properties of these K\"ahler cones.

Physically, we have proven that there  exists a finite
number of four-dimensional $\mathcal{N}=1$ supergravity theories realized by these
compactifications. Most notably, there arise only finitely many gauge
sectors in these theories with finitely many different chiral spectra. The details
of these gauge sectors are determined by  the bounded number of branes in a stack and the bounded magnetic flux quanta.
Concretely, this means that the ranks of the gauge groups are bounded,
that only certain matter representations with certain chiral indices
exist (which is always true in weakly coupled Type IIB) and that for fixed gauge group there exist only a finite set
of possible multiplicities for the matter fields. These finiteness
properties, and more broadly similar results elsewhere in the
landscape, are particularly interesting when contrasted to generic
quantum field theories.

While we have shown finiteness of these compactifications and provided
explicit bounds, we have not explicitly constructed all
of these  compactifications. It would
be interesting to systematically construct this finite set of
configurations  and extract generic features
of the four-dimensional effective
theories in this corner of the landscape. In addition, we have not
systematically explored  the bases $B$ for
which the proof applies, i.e. ~there may exist additional
algebraic surfaces satisfying  the geometric conditions of Section \ref{sec:generic proof}. 
Other points of interest would be to determine whether a simple
modification of our proof exists for blow-ups of singular elliptic
fibrations or elliptically fibered Calabi-Yau manifolds which do not
satisfy the supergravity approximation; in the latter case the
supersymmetry conditions receive corrections of various types. Of most interest would be to find a general proof for
a general Calabi-Yau threefold $X$. It seems plausible that there
are even more general proof techniques which utilize SUSY and tadpole
cancellation conditions to prove finiteness for a general $X$. For example, some of the arguments in the
proof presented here, e.g.~the ones used to eliminate the dependence of the SUSY conditions 
\eqref{eq:SUSYcondition} on the K\"ahler moduli, should still apply for general Calabi-Yau 
manifolds $X$.
In addition, string dualities of the
considered Type IIB configurations extend our finiteness proof to the
dual theories, for example to the heterotic string on certain elliptic
fibrations with specific vector bundles and to F-theory on certain
elliptic $K3$-fibered fourfolds. It is very important to work
out  the details of the duality maps and the analogs of
the bounds we found in the dual theories. 

The presented proof is based
on tadpole and supersymmetry conditions at weak coupling and large
volume of $X$. It is crucial for a better 
understanding of the string landscape to understand string
consistency conditions away from large volume and weak
coupling. This requires the understanding of perturbative and
non-perturbative corrections\footnote{See \cite{Grimm:2012rg,GarciaEtxebarria:2012zm,Grimm:2013gma,Grimm:2013bha} for recent 
computations of corrections to $\mathcal{N}=1$ couplings in M-/F-theory compactifications.} both in $\alpha'$ and in $g_S$;
for example, the supersymmetry conditions receive $\alpha'$-corrections from worldsheet instantons. 
Avenues towards a better understanding might be provided by applications of $\mathcal{N}=1$
mirror symmetry, i.e.~mirror symmetry, and $S$-duality.

 It is particularly interesting that the
  finiteness results we have proven and similar results elsewhere in
  the landscape do not have known analogs in generic quantum field
  theories. Such differences are one of the hallmarks of string
  compactifications, and it seems reasonable to expect that similar
  finiteness results can be proven for even the most general string
  compactifications, in particular those at small volume and strong
  coupling.  This would have profound implications for our picture of
  the landscape: while it is larger than originally thought, our
  results provide further evidence that it may, in fact, be finite.

 \acknowledgments We thank Mike Douglas, Antonella Grassi,
Albrecht Klemm, Dave Morrison, Hernan Piragua and Wati Taylor for 
useful conversations and correspondence. This research is supported in part 
supported by the DOE grant DE-SC0007901 (M.C. and D.K.), Dean’s Funds for 
Faculty Working Group (M.C. and D.K.), the Fay R. and Eugene L. Langberg Endowed 
Chair (M.C.), the Slovenian Research Agency (ARRS) (M.C.) and the NSF grant
PHY11-25915 (J.H.). J.H. thanks J.L. Halverson for her encouragement.

\appendix

\section{K\"ahler Cones of del Pezzo Surfaces \& their $M_{\{i,k\}}$-Matrices}
\label{app:MoriKaehlerFano}

In this Appendix we discuss in detail the structure of the K\"ahler cone of the del Pezzo surfaces $dP_n$ for $n\leq 8$. 
We are interested in the extremal rays, i.e. the generators, of these in general non-simplicial cones, and the existence of 
coverings of these cones by simplicial subcones so that conditions (1)-(3) listed at the beginning of Section \ref{sec:generic proof} are 
obeyed.

First, we expand the K\"ahler cone generators $D_i$ of $dP_n$ in the basis \eqref{eq:H2dPn} of $H^2(dP_n,\mathbb{Z})$
\beq \label{eq:v_i}
	D_i=(v_i)^1H+\sum_{j=1}^{n}(v_i)^j E_j\,,
\eeq
which maps every $D_i$ to a vector $v_i$ in $\mathbb{Z}^{n+1}$. With this definition, we obtain the matrices \eqref{M} in this basis as
\beq\label{eq:Malternative}
	M_{\{i,k\}}=\eta\cdot[x_{\{i,k\}}(v_i\cdot v_k^T+v_k\cdot v_i^T)-\eta]\cdot \eta\,,
\eeq
where $i\neq k$, $v^T$ denotes the transpose of a vector, '$\cdot$' denotes the matrix product and $\eta=\text{diag}(1,-1,\ldots,-1)$ is the 
standard Minkowski matrix in $n+1$ dimensions. We note that in order to check positive semi-definiteness of the matrices in 
\eqref{eq:Malternative}, it suffices to prove it for the matrices $\eta\cdot M_{\{i,k\}}\cdot\eta$, which is the matrix in the square brackets 
in \eqref{eq:Malternative}.

Next, we need the explicit form for the K\"ahler generators of $dP_n$. 
We present these by listing the corresponding vectors $v_i$ defined 
via \eqref{eq:v_i}. We explicitly solve \eqref{eq:KaehlerConedPn} over the integers to obtain the K\"ahler cone generators. For the simplicial cases $dP_0$, $dP_1$, $dP_2$ we obtain \eqref{eq:simplicialKCdPn} as discussed earlier. In the non-simplicial cases $dP_n$, $n>2$, we summarize the generators in Table \ref{tab:KCdPn}.

Here, the second column contains the schematic form of the vectors 
$v_i$, with each row containing all vectors of the same particular 
form. In each row, the explicit expressions for the $v_i$ are obtained 
by inserting the values listed in the third column for the place holder 
variables in the entries of $v_i$ in that row and by permuting the 
underlined entries of the vector $v_i$. The number of different vectors
in each row is given in the fourth column, where the two factors are
given by the number of elements in the list in the third column and the
number of permutations of the entries, respectively. The fifth column
contains a list of the Minkowski length of all vectors in a given row. 
We note that this column precisely contains the self-intersection of 
the curves associated to the $D_i$. All are either $0$ or $1$ and it 
can be checked that the intersections of the $v_i$ with 
$c_1(dP_n)=3H-\sum_i E_i\equiv 
(3,-1,\ldots,-1)$ are precisely $2$ or $3$, respectively, as 
required by \eqref{eq:KaehlerConedPn}.  

For example, in the second row of Table \ref{tab:KCdPn}, all vectors 
$v_i$ are of the form $v_i=(a,b,b,b)$ by the second column. By the 
third column, there are two different vectors of this type, namely 
$v_1=(2,-1,-1,-1)$ and $v_2=(1,0,0,0)$. Thus, there are precisely $2$ 
vectors as indicated in the fourth column and the Minkowski length of 
the two vectors is $1$, $1$, respectively, as in the last column of the second row.

We note that the K\"ahler cone generators and their grouping as in 
Table \ref{tab:KCdPn} can be understood by representation theory, 
recalling that the Weyl group naturally acts on $H_2(dP_n,\mathbb{Z})$.
\begin{table}[H]
\vspace{-0.2cm}
\hspace{-0.9cm}
\small
\begin{tabular}{|p{0.3cm}||>{\centering}p{3.3cm}|p{7.8cm}|>{\centering}p{0.98cm}|p{2cm}|}
\hline
 & \multicolumn{2}{c|}{K\"ahler cone generators $v_i\phantom{\sum^{A}_B}\hspace{-0.6cm}$}& $\#$ &$\eta(v_i,v_i)$ \\
\hline
\hline
\multirow{3}{*}{\!\!$dP_3$\!\!} & $(a,b,b,b)$ & $(a,b)\in \{(2,\text{-}1),(1,0)\}$& $2\cdot 1$ & $\{1,1\}$  \\
& $(c,\underline{d,e,e})\phantom{\int_{A_{\Sigma_k}}}\hspace{-0.77cm}$ & $(c,d,e)\in\{(1,\text{-}1,0)\}$& $3$ & $0$ \\
\cline{2-5}
& \multicolumn{2}{c|}{\phantom{mmmmmmmmmmm}\,\,\,Total number of K\"ahler generators \,\,\,\,\,\,\,\,\,\,\,$=\phantom{\int_{A_1}^O}\hspace{-0.55cm}$} &$5$  & \\
\hline \hline
\multirow{3}{*}{\!\!$dP_4$\!\!} & $(a,b,b,b,b)$ & $(a,b)\in \{(2,\text{-}1),(1,0)\}$& $2\cdot 1$ & $\{0,1\}$  \\
& $(c,\underline{d,e,e,e})\phantom{\int_{A_{\Sigma_k}}}\hspace{-0.77cm}$ & $(c,d,e)\in\{(2,0,\text{-}1),(1,\text{-}1,0)\}$& $2\cdot 4$ & $\{1,0\}$ \\
\cline{2-5}
& \multicolumn{2}{c|}{\phantom{mmmmmmmmmmm}\,\,\,Total number of K\"ahler generators \,\,\,\,\,\,\,\,\,\,\,$=\phantom{\int_{A_1}^O}\hspace{-0.55cm}$} &$10$  &\\
\hline 
\hline
\multirow{5}{*}{\!\!$dP_5$\!\!} & $(a,b,b,b,b,b)$ & $(a,b)\in\{(1,0)\}$& $1$ & $1$ \\
& $(c,\underline{d,e,e,e,e})$ & $(c,d,e)\in\{(3,\text{-}2,\text{-}1),(2,0,\text{-}1),(1,\text{-}1,0)\}$& $3\cdot 5$ & $\{0,0,1\}$ \\
& $(f,\underline{g,g,g,h,h})\phantom{\int_{A_{\Sigma_k}}}\hspace{-0.77cm}$ & $(f,g,h)\in\{(2,\text{-}1,0)\}$& $10$ & $0$ \\
\cline{2-5}
& \multicolumn{2}{c|}{\phantom{mmmmmmmmmmm}\,\,\,Total number of K\"ahler generators \,\,\,\,\,\,\,\,\,\,\,$=\phantom{\int_{A_1}^O}\hspace{-0.55cm}$} &$26$ &  \\
\hline 
\hline
\multirow{6}{*}{\!\!$dP_6$\!\!} & $(a,b,b,b,b,b,b)$ & $(a,b)\in\{(1,0),(5,-2)\}$& $2\cdot1$ & $\{1,1\}$  \\
& $(c,\underline{d,e,e,e,e,e})$ & $(c,d,e)\in\{(1,\text{-}1,0),(3,\text{-}2,\text{-}1)\}$& $2\cdot 6$ & $\{0,0\}$ \\
& $(f,\underline{g,g,g,g,h,h})$ & $(f,g,h)\in\{(2,\text{-}1,0)\}$& $15$ & $0$ \\
& $(i,\underline{j,j,j,k,k,k})$ & $(i,j,k)\in\{(2,\text{-}1,0),(4,\text{-}2,\text{-}1)\}$& $2\cdot 20$ & $\{1,1\}$ \\
& $(l,\underline{m,n,n,n,n,o})\phantom{\int_{A_{\Sigma_k}}}\hspace{-0.77cm}$ & $(l,m,n,o)\in\{(3,\text{-}2,\text{-}1,0)\}$& $30$ & $1$ \\
\cline{2-5}
& \multicolumn{2}{c|}{\phantom{mmmmmmmmmmm}\,\,\,Total number of K\"ahler generators \,\,\,\,\,\,\,\,\,\,\,$=\phantom{\int_{A_1}^O}\hspace{-0.55cm}$} &$99$ & \\
\hline 
\hline
\multirow{6}{*}{\!\!$dP_7$\!\!} & $(a,b,b,b,b,b,b,b)$ & $(a,b)\in\{(8,\text{-}3),(1,0)\}$& $2\cdot1$ & $\{1,1\}$ \\
& $(c,\underline{d,e,e,e,e,e,e})$ & $(c,d,e)\in\{(5,0,\text{-}2),(5,\text{-}1,\text{-}2),(4,\text{-}3,\text{-}1),(1,\text{-}1,0)\}$& $4\cdot 7$ &  $\{1,0,1,0\}$  \\
& $(f,\underline{g,g,g,h,h,h,h})$ & $(f,g,h)\in\{(7,\text{-}2,\text{-}3),(4,\text{-}2,\text{-}1),(2,0,\text{-}1),(2,\text{-}1,0)\}$& $4\cdot 35$ &  $\{1,0,0,1\}$  \\
& $(i,\underline{j,k,l,l,l,l,l})$ & $(i,j,k,l)\in\{(3,0,\text{-}2,\text{-}1)\}$& $42$ & $0$ \\
& $(m,\underline{n,o,o,p,p,p,p})$ & $(m,n,o,p)\in\{(6,\text{-}1,\text{-}3,\text{-}2),(3,\text{-}2,0,\text{-}1)\}$& $2\cdot 105$ & $\{1,1\}$ \\
& $(q,\underline{r,s,s,s,t,t,t})\phantom{\int_{A_{\Sigma_k}}}\hspace{-0.77cm}$ & $(q,r,s,t)\in\{(5,\text{-}3,\text{-}2,\text{-}1),(4,0,\text{-}2,\text{-}1)\}$& $2\cdot 140$ & $\{1,1\}$ \\
\cline{2-5}
& \multicolumn{2}{c|}{\phantom{mmmmmmmmmmm}\,\,\,Total number of K\"ahler generators \,\,\,\,\,\,\,\,\,\,\,$=\phantom{\int_{A_1}^O}\hspace{-0.55cm}$} &$702$ & \\
\hline 
\hline
\multirow{6}{*}{\!\!$dP_8$\!\!} & 
$(a,b,b,b,b,b,b,b,b)$ & $(a,b)\in\{(17,\text{-}6),(1,0)\}$& $2\cdot1$ & $\{1,1\}$  \\
& $(c,\underline{d,e,e,e,e,e,e,e})$ & $(c,d,e)\in\{(11,\text{-}3,\text{-}4),(10,\text{-}6,\text{-}3),(8,\text{-}1,\text{-}3),(8,0,\text{-}3),$ $\phantom{mmmmm}\,\,\,(4,\text{-}3,\text{-}1),(1,-\text{1},0)\}$& $6\cdot 8$ &  $\{0,1,0,1,0,0\}$  \\
& $(f,\underline{g,g,h,h,h,h,h,h})$ & $(f,g,h)\in\{(13,\text{-}6,\text{-}4),(5,0,\text{-}2)\}$& $2\cdot 28$ &  $\{1,1\}$  \\
& $(i,\underline{j,j,j,k,k,k,k,k})$ & $(i,j,k)\in\{(16,\text{-}5,\text{-}6),(2,\text{-}1,0)\}$& $2\cdot 56$ & $\{1,1\}$ \\
& $(l,\underline{m,n,o,o,o,o,o,o})$ & $(l,m,n,o)\!\in\!\{(14,\text{-}3,\text{-}6,\text{-}5),(7,\text{-}4,\text{-}3,\text{-}2),(5,\text{-}1,0,\text{-}2),$ $\phantom{mmmmmmm}(4,\text{-}3,0,\text{-}1)\}$& $4\cdot 56$ & $\{1,0,0,1\}$ \\
& $(p,\underline{q,q,q,q,r,r,r,r})$ & $(p,q,r)\!\in\!\{(10,\text{-}4,\text{-}3),(9,\text{-}4,\text{-}2),(2,\text{-}1,0)\}$& $3\cdot 70$ & $\{0,1,0\}$ \\
& $(s,\underline{t,u,u,v,v,v,v,v})$ & $(s,t,u,v)\in\!\{(10,\!\text{-}2,\!\text{-}5,\!\text{-}3),(10,\text{-}1,\text{-}3,\text{-}4),(9,\text{-}2,\text{-}4,\text{-}3),$ $\phantom{mmmmmm}\,\,\,\,(8,\text{-}5,\text{-}3,\text{-}2),(8,\text{-}4,\text{-}1,\text{-}3),(3,\text{-}2,0,\text{-}1)\}$& $6\cdot 168$ & $\{1,1,0,1,1,0\}$ \\
& $(w,\underline{x,y,y,y,z,z,z,z})$ & $(w,x,y,z)\in\{(15,\text{-}4,\text{-}6,\text{-}5),(12,\text{-}4,\text{-}3,\text{-}5),\phantom{mmmmm}$ $\phantom{mmmmmmm}\,(12,\text{-}2,\text{-}5,\text{-}4),(11,\text{-}6,\text{-}4,\text{-}3),$ $\phantom{mmmmmmm}\,(8,\text{-}4,\text{-}2,\text{-}3),(7,\text{-}1,\text{-}2,\text{-}3),(7,0,\text{-}2,\text{-}3),$ $\phantom{mmmmmmm}\,(6,\text{-}4,\text{-}1,\text{-}2),(6,\text{-}2,\text{-}3,\text{-}1),(5,\text{-}3,\text{-}2,\text{-}1),$ $\phantom{mmmmmmm}\,(4,0,\text{-}2,\text{-}1),(3,\text{-}2,0,\text{-}1)\}$& ${12\cdot 280}$ & $\{1,1,1,1,0,0,$ $\phantom{.\,}1,1,1,0,0,1\}$ \\
& $(\tilde{a},\underline{\tilde{b},\tilde{b},\tilde{c},\tilde{c},\tilde{d},\tilde{d},\tilde{d},\tilde{d}})$ & $(\tilde{a},\tilde{b},\tilde{c},\tilde{d})\in\{(6,\text{-}1,\text{-}3,\text{-}2)\}$& $420$ & $0$ \\
& $\phantom{\int_{A_1}^O}\hspace{-0.55cm}(\tilde{e},\underline{\tilde{f},\tilde{f},\tilde{g},\tilde{g},\tilde{g},\tilde{h},\tilde{h},\tilde{h}})$ & $(\tilde{e},\tilde{f},\tilde{g},\tilde{h})\in\{(14,\text{-}6,\text{-}5,\text{-}4),(4,0,\text{-}1,\text{-}2)\}$& $2\cdot 560$ & $\{1,1\}$ \\
& $\phantom{\int_{A_1}^O}\hspace{-0.55cm}(\tilde{i},\underline{\tilde{j},\tilde{k},\tilde{l},\tilde{l},\tilde{m},\tilde{m},\tilde{m},\tilde{m}})$ & $(\tilde{i},\tilde{j},\tilde{k},\tilde{l},\tilde{m})\in\{(12,\text{-}6,\text{-}5,\text{-}3,\text{-}4),(6,0,\text{-}1,\text{-}3,\text{-}2)\}$& $2\cdot 840$ & $\{1,1\}$ \\
& $(\tilde{n},\underline{\tilde{o},\tilde{p},\tilde{q},\tilde{q},\tilde{q},\tilde{r},\tilde{r},\tilde{r}})$ & $(\tilde{n},\tilde{o},\tilde{p},\tilde{q},\tilde{r})\in\{(13,\text{-}3,\text{-}6,\text{-}5,\text{-}4),(9,\text{-}5,\text{-}4,\text{-}3,\text{-}2),\phantom{mm}$ $\phantom{mmmmmmm}\,\,\,\,(9,\text{-}2,\text{-}1,\text{-}4,\text{-}3),(5,\text{-}3,0,\text{-}2,\text{-}1)\}$& $4\cdot 1120$ & $\{1,1,1,1\}$ \\
& $(\tilde{s},\underline{\tilde{t},\tilde{u},\tilde{u},\tilde{v},\tilde{v},\tilde{w},\tilde{w},\tilde{w}})$ & $(\tilde{s},\tilde{t},\tilde{u},\tilde{v},\tilde{w})\in\{(11,\text{-}2,\text{-}3,\text{-}5,\text{-}4),(10,\text{-}5,\text{-}2,\text{-}3,\text{-}4),\phantom{mm}$ $\phantom{mmmmmmmm}(8,\text{-}1,\text{-}4,\text{-}3,\text{-}2),(7,\text{-}4,\text{-}3,\text{-}1,\text{-}2)\}$& $4\cdot 1680$ & $\{1,1,1,1\}$\normalsize \\
\cline{2-5}
& \multicolumn{2}{c|}{\phantom{mmmmmmmmmmm}\,\,\,Total number of K\"ahler generators \,\,\,\,\,\,\,\,\,\,\,$=\phantom{\int_{A_1}^O}\hspace{-0.55cm}$} &$19440$ &\\
\hline 
\end{tabular}
\caption{K\"ahler cone generators for $dP_n$. The underlined entries of the $v_i$ are permuted.}
\label{tab:KCdPn}
\end{table}
For instance the K\"ahler cone generators of $dP_n$, $n=2,\ldots, 6$ 
form the representations $\mathbf{3}$, 
$(\bar{\mathbf{3}}\otimes \mathbf{1})\oplus (\mathbf{1}\otimes 
\mathbf{2})$, 
$\mathbf{5}\oplus\bar{\mathbf{5}}$, $\mathbf{16}\oplus\mathbf{10}$ and 
$\mathbf{78}\oplus\mathbf{27}$ under the corresponding groups $A_1$, 
$A_2\times A_1$, $A_4$, $D_5$ and $E_6$, respectively. Here the first 
representation in all direct sums is formed by all generators with 
Minkowski length $1$ and the second one is formed by generators 
with Minkowski length $0$. These results can be 
worked out explicitly by computing the Dynkin labels of the generators 
in Table \ref{tab:KCdPn} for the canonical roots $\alpha_i$, which 
are the $-2$-curves in $H_2(dP_n,\mathbb{Z})$ orthogonal to $c_1(dP_n)$. 
Thus, the zero weight vector is identified with $c_1(dP_n)$. For $dP_7$
only the union of the generators of the K\"ahler and Mori cone have
a representation theoretical decomposition as $\mathbf{912}\oplus 
\mathbf{133}$ (some of the weights of the $\mathbf{912}$ 
have higher multiplicities yielding only $576$ different weights), where the first representation contains the length $1$ 
and the second one the length $0$ K\"ahler cone generators.

Next, we make one important observation. As one can check explicitly from 
Table \ref{tab:KCdPn} and \eqref{eq:simplicialKCdPn}, for every del 
Pezzo $dP_n$ with $n>1$, the first Chern class
$c_1(dP_n)\equiv (3,-1,\ldots,-1)$ is proportional to the sum of all
K\"ahler cone generators $v_i$
\beq \label{eq:c1inCenter}
	c_1(dP_n)\equiv (3,-1,\ldots,-1)=\frac{1}{A_n\cdot N}\sum_{i=1}^N v_i
\eeq
where $N$ denotes the total number of K\"ahler cone generators of $dP_n$, cf.~Table \ref{tab:KCdPn}.
The positive proportionality factor $A_n$ depends on $n$ and reads
\beq \label{eq:Ac1}
	A_3=\frac{2}{5}\,,\quad A_4=\frac{1}{2}\,,\quad A_5=\frac{17}{26}\,,\quad A_6=\frac{10}{11}\,,\quad A_7=\frac{55}{39}\,,\quad A_8=\frac{26}{9}
\eeq
for $dP_3$, $dP_4$, $dP_5$, $dP_6$, $dP_7$ and $dP_8$, respectively.
This means that $c_1(dP_n)$ 
is in the center of the K\"ahler cone of all del Pezzo surfaces with 
$n>1$. 

This implies that we can  find a cover of the K\"ahler cone by
simplicial subcones so that properties (1)-(3)  at the beginning of section \ref{sec:generic proof} are satisfied. We present two such 
covers:
\\

\noindent{\textbf{Cover 1:}}
Intersect the K\"ahler cone 
with a hyperplane that is normal to $c_1(dP_n)$ and passes through $c_1(dP_n)$.
This yields an $n$-dimensional polytope with vertices corresponding to the generators of the K\"ahler cone. 
Triangulate this polytope with star being 
$c_1(dP_n)$. This triangulation induces a decomposition of the K\"ahler cone into simplicial
subcones. As the generators of one simplicial subcone, take $c_1(dP_n)$ and 
those generators $v_i$ of the K\"ahler cone that go through the vertices of an $n$-dimensional cone of the triangulated 
polytope.

In this covering of the K\"ahler cone, properties (2) and (3) are satisfied: we obviously have $b_i$ all positive, because $c_1(dP_n)$ is one of the 
generators in every simplicial subcone. From \eqref{eq:basisexp} we get $b_i=0$ for all $D_i\neq c_1(dP_n)$
and $b_K=1$, where $K$ denotes the index such that $D_K=c_1(dP_n)$. In addition, we have $C_{KK}=\mathcal{K}_{000}=9-n$ and 
$C_{iK}=\mathcal{K}_{00i}=2,\,3$
for $i\neq K$ by \eqref{eq:C_IJKrels} and \eqref{eq:C00iC000dPn} and $T_K=12\int_B c_1^2=12\mathcal{K}_{000}=12(9-n)$ by 
\eqref{eq:T_iEvaluated} and \eqref{eq:C00iC000dPn}. We discuss why property (1) is satisfied later.\\

\noindent{\textbf{Cover 2:}}
Although the above cover 1 obeys all the required properties listed at the beginning of Section \ref{sec:generic proof}, it slightly
increases the bounds because it increases $\text{max}(T_i)$ for $n\leq 6$ in which case $\text{max}(T_i)=T_K=12(9-n)$ is larger 
than the $T_i$ found in \eqref{eq:T_iEvaluateddPn}.

Thus, we provide the following alternative cover which exists if the K\"ahler cone is sufficiently symmetric, in addition to $c_1(dP_n)$ 
being its center. Take a vertex of the polytope 
constructed in cover 1. Construct the line through that vertex and the star, i.e.~$c_1(dP_n)$. This line has to intersect the boundary of
the polytope at another point. This point lies on a certain facet of this polytope. Take the vertices of this facet together with the original 
vertex we have started with to define a simplicial subcone of the K\"ahler cone.  Notice that this subcone contains $c_1(dP_n)$ and the 
cone formed by the vertices of this facet and $c_1(dP_n)$, i.e.~a subcone in cover 1. Repeat this procedure for all vertices of the polytope.
If the K\"ahler cone is sufficiently symmetric, each facet will be hit exactly once. Thus, each subcone in cover 1 is contained in a 
corresponding subcone defined in this way.  Consequently, since cover 1 covers the K\"ahler cone completely, so does cover 2.

This cover also satisfies conditions (1)-(3) at the beginning of Section \ref{sec:generic proof}. We again leave the discussion of condition (1) for later. 
Conditions (2) and (3) are satisfied since $c_1(dP_n)$ is contained in each subcone, which implies 
$b_i\geq 0$ for all $i$, and by \eqref{eq:C00iC000dPn} all $\mathcal{K}_{00i}$ are positive integers.
In addition, the advantage of this cover is that all generators of all simplicial subcones are generators of the K\"ahler cone. Thus in all bounds derived in this work, we have that $\text{max}(T_i)$ is given precisely by \eqref{eq:T_iEvaluateddPn}. Given the fact that the generators
of the K\"ahler cone sit in representations of  Lie algebras, which implies that the K\"ahler cone is  symmetric, and that 
$c_1(dP_n)$ lies in its center, we expected this cover 2 to exist.
\\

Finally, we discuss why condition (1), i.e. the positive semi-definiteness of the matrices $M_{\{i,k\}}$ in \eqref{M}, is satisfied in both Cover 1 and Cover 2. We notice the following fact: for both covers, in order to show that the matrices \eqref{M} are positive semi-definite,
we only have to prove that these matrices written in the form \eqref{eq:Malternative} are positive semi-definite for all possible choices of two vectors $v_i$, $v_j$ of Table \ref{tab:KCdPn}. This is clear for Cover 2, because the generators of all simplicial subcones are generators of the K\"ahler cone. For Cover 1, in every simplicial subcone, all matrices $M_{\{i,j\}}$ with $i,j\neq K$ involve only the generators $v_i$, $v_j$. 
Thus, we only have to consider the matrices $M_{\{i,K\}}$ with $i\neq K$ (recall that we only have to show positive 
semi-definiteness of the matrices $M_{\{i,j\}}$ for $i\neq j$). For these we use
\begin{lemma}
In Cover 1, let $K$ be the index corresponding to $c_1(dP_n)$, then the matrices $M_{\{i,K\}}$ for all $i\neq K$ are positive semi-definite, if all matrices
$M_{\{i,j\}}$ for all pairs of generators $v_i$, $v_j$ of the K\"ahler cone are positive semi-definite.
\end{lemma}
\begin{proof}
Using  the first Chern class $c_1(dP_n)\equiv(3,-1,\ldots,-1)$ and $\lambda_j=\frac{1}{A_n\cdot N}$, we 
obtain
\begin{equation} \label{eq:Mikc1}
M_{\{i,K\}}=x_{\{i,K\}}(v_i\cdot c_1(dP_n)^t+c_1(dP_n)\cdot v_i^t)-\eta=\sum_{j=1}^N \lambda_{j} x_{\{i,K\}}(v_i\cdot v_j^t+v_j\cdot v_i^t)-\eta\,,
\end{equation} 
where we used \eqref{eq:c1inCenter}.
Choose $x_{\{i,K\}}$ for every $i$ so that the following equality is satisfied
\begin{equation} \label{eq:xiK}
\sum_{j=1}^N \lambda_j \frac{x_{\{i,K\}}}{x_{\{i,j\}}}=x_{\{i,K\}}\sum_{j=1}^N  \frac{\lambda_j}{x_{\{i,j\}}}=x_{\{i,K\}}\frac{1}{A_n}\Big\langle\frac{1}{x_{\{i,j\}}}\Big\rangle_j \stackrel{!}{=}1\,,
\end{equation}
where $\big\langle\frac{1}{x_{\{i,j\}}}\big\rangle_j$ denotes the average of $\frac{1}{x_{\{i,j\}}}$ with $i$ kept fixed and $j$ varied over all 
K\"ahler cone generators. Then, \eqref{eq:Mikc1} can be written as
\begin{eqnarray} \label{eq:MiKFinal}
M_{\{i,K\}}&=&\sum_{j=1}^{N}\lambda_j \frac{x_{\{i,K\}}}{x_{\{i,j\}}}x_{\{i,j\}}(v_i\cdot v_j^t+v_j\cdot v_i^t)-\eta=\sum_{j=1}^N \lambda_j \frac{x_{\{i,K\}}}{x_{\{i,j\}}}(x_{\{i,j\}}(v_i\cdot v_j^t+v_j\cdot v_i^t)-\eta)\nonumber\\
&=&\sum_{j=1}^N \lambda'_j(x_{\{i,j\}}(v_i\cdot v_j^t+v_j\cdot v_i^t)-\eta)=\sum_{j=1}^N \lambda'_jM_{\{i,j\}}\,,
\end{eqnarray} 
where we set $\lambda_j'=\lambda_j \frac{x_{\{i,K\}}}{x_{\{i,j\}}}$. We note that $M_{\{i,K\}}$ is defined in terms of generators of the 
K\"ahler 
cone and $\lambda_j'\geq 0$ for all $j=1,\ldots, N$. Thus, if all the $M_{\{i,j\}}$ are positive semi-definite, then
$M_{\{i,K\}}$ will be automatically positive semi-definite  because it is just a positive linear combination of the $M_{\{i,j\}}$ by 
\eqref{eq:MiKFinal}. A positive linear combination of positive semi-definite matrices is again positive semi-definite.
\end{proof}

Thus, it only remains to show positive semi-definiteness of the matrices $M_{\{i,k\}}$ defined in \eqref{eq:Malternative} for any choice of two K\"ahler 
cone generators of $dP_n$ from Table \ref{tab:KCdPn}. We note that the 
K\"ahler cone generators of $dP_n$ are obtained by
permutations of the vectors in Table \ref{tab:KCdPn}. Most of these 
permutations simply interchange the rows and columns of the matrices  
\eqref{eq:Malternative}, which does not affect their eigenvalues. Thus,
we only have to consider matrices  \eqref{eq:Malternative} that do not differ only by such a permutation. 
We provide an efficient algorithm making use of this permutation
symmetry to generate all  matrices 
\eqref{eq:Malternative} with different sets of eigenvalues. Recall that to check 
positive-semi-definiteness for any $M_{\{i,k\}}$, it suffices to check positive-semi-definiteness for 
$\tilde{M}_{\{i,k\}}$, defined as
\begin{equation}\label{eq:tildeM}
\tilde{M}_{\{i,k\}}=x_{\{i,k\}}(v_i\cdot v_k^T+v_k\cdot v_i^T)-\eta \,.
\end{equation}

For each $\tilde{M}_{\{i,k\}}$ define $(v_i, v_k)$ as the pair of K\"ahler cone generators in its definition \eqref{eq:tildeM}. By definition of $M_{\{i,k\}}$, we have $i\neq k$ in $(v_i, v_k)$. For each $dP_n$, we define an equivalence relation on the set of all pairs $(v_i, v_k)$ and show if $(v_i, v_k)\sim (v'_i, v'_k)$ and $x_{\{i,k\}}=x'_{\{i,k\}}$, the corresponding matrices $\tilde{M}_{\{i,k\}}$ and $\tilde{M'}_{\{i,k\}}$ have the same sets of eigenvalues. 

\begin{definition}\label{Definition:EquivalenceRelation}
For each $dP_n$, let $\{(v_i, v_k)\}$, $i\neq k$, be the set of all pairs of its K\"ahler cone generators. The symmetric group $S_n$ of degree $n$ acts
on the K\"ahler cone generator $v_i \in \mathbb{Z}^{1+n}$  by permuting its last $n$ components, cf. the second column of Table \ref{tab:KCdPn}. Define an equivalence relation $\sim$ on  $\{(v_i, v_k)\}$ by $(v_i, v_k)\sim (v'_i, v'_k)$ if $(v'_i, v'_k)=(\sigma(v_i), \sigma(v_k))$, for some $\sigma \in S_n$.
\end{definition}

\begin{lemma}\label{lemma:EquivClassSameEigenvalue}
Suppose $(v_i, v_k)\sim (v'_i, v'_k)$. Let $\tilde{M}_{\{i,k\}}$ and $\tilde{M'}_{\{i,k\}}$ be the matrix defined by $(v_i, v_k)$ and $(v'_i, v'_k)$, 
respectively, with $x_{\{i,k\}}=x'_{\{i,k\}}$, in $\eqref{eq:tildeM}$. Then $\tilde{M}_{\{i,k\}}$ and $\tilde{M'}_{\{i,k\}}$ have the same set of  eigenvalues. 
\end{lemma}

\begin{proof}
Let $\sigma \in S_n$ so that $(v'_i, v'_k)=(\sigma(v_i), \sigma(v_k))$. Denote the permutation matrix that permutes the $j^{\text{th}}$ and $l^{\text{th}}$ rows/columns by $P_{jl}$. Since any $\sigma\in S_n$ can be written as a product of such permutation matrices, we can WLOG assume $\sigma=P_{jl}$. Then we have
\begin{align}
\tilde{M'}_{\{i,k\}} & = x_{\{i,k\}}(P_{jl} v_i v_k^T P_{jl}^T + P_{jl} v_k v_i^T P_{jl}^T)-\eta 
=  x_{\{i,k\}}(P_{jl} v_i v_k^T P_{jl}^T + P_{jl} v_k v_i^T P_{jl}^T)-P_{jl} \eta P_{jl}^T \nonumber\\
& = P_{jl} [x_{\{i,k\}}(v_i v_k^T  + v_k v_i^T )-\eta ]P_{jl}^T 
= P_{jl} \tilde{M}_{\{i,k\}} P_{jl}^T \,.
\end{align}
This implies that the characteristic polynomials of $\tilde{M}_{\{i,k\}}$ and $\tilde{M'}_{\{i,k\}}$ are the same,
\begin{align}
\text{det}\Big(\tilde{M'}_{\{i,k\}}-\lambda I\Big) & = \text{det}\Big(P_{jl} \tilde{M}_{\{i,k\}} P_{jl}^T -\lambda P_{jl} I P_{jl}^T\Big) 
= \text{det}\Big(P_{jl} \Big(\tilde{M}_{\{i,k\}} -\lambda I \Big)P_{jl}^T \Big) \nonumber\\
& = \text{det}(P_{jl}) \text{det}\Big(\tilde{M}_{\{i,k\}} -\lambda I \Big)\text{det}(P_{jl}^T) = \text{det}\Big(\tilde{M}_{\{i,k\}} -\lambda I \Big)\,.
\end{align}
\end{proof}
\vspace{-0.2cm}
Lemma \ref{lemma:EquivClassSameEigenvalue} shows that for each equivalence class $[(v_i,v_k)]$, we just need to pick any representative $(v_i,v_k)$ 
and check if there exists an $x_{\{i,k\}}\in \mathbb{Q^+}$ such that $(v_i,v_k)$ and $x_{\{i,k\}}$ defines a positive semi-definite matrix 
$\tilde{M}_{\{i,k\}}$ according to \eqref{eq:tildeM}. If such an $x_{\{i,k\}}$ exists, any $\tilde{M'}_{\{i,k\}}$ with $(v'_i, v'_k) \sim (v_i, v_k)$ will be 
automatically positive semi-definite for $x'_{\{i,k\}}=x_{\{i,k\}}$. For each $dP_n$, in order to find all different equivalence classes, we start by picking 
an arbitrary pair $(v_i,v_k)$ from Table \ref{tab:KCdPn} and carry out 
the following algorithm: 
\begin{itemize}
\item[\textbf{(1)}] Fix $v_i$ and only permute the entries of $v_k$. Indeed, if $v'_i=\sigma(v_i), v'_k=\tau(v_k)$, then $(v'_i,v'_k)\sim (v_i, \sigma^{-1}\tau(v_k))$. Let $\tau'=\sigma^{-1}\tau$, then we have $[(v'_i,v'_k)]=[(v_i,\tau'(v_k))]$.

\item[\textbf{(2)}] Only permute those entries in $v_k$ for which the corresponding entries in $v_i$ are 
different from each other. Permuting two entries in $v_k$ when the 
corresponding two entries in the fixed vector $v_i$ are the same is equivalent 
to the action of permuting these two entries for both vectors. Thus, the resulting pair of vectors $(v_i,v'_k)\sim (v_i,v_k)$.
\end{itemize}
Pick a different pair $(w_i,w_k)$ of K\"ahler cone generators from Table \ref{tab:KCdPn} and repeat (1), (2).

For example, consider $dP_8$. Suppose we begin by picking  
$v_i=(a,b,b,b,b,b,b,b,b)$ and $v_k=(s,t,u,u,v,v,v,v,v)$ from the second column 
of Table \ref{tab:KCdPn}. By (1) 
above, we can fix $v_i$ and only consider permutations in the last eight entries 
of $v_k$. By (2), however, we do not need to consider any permutation in the 
last eight entries in $v_k$, because the last eight entries in the fixed vector 
$v_i$ are the same; they are all equal to $b$. Thus, there is only one equivalence class $[(v_i,v_k)]$. From the third column of 
Table \ref{tab:KCdPn}, there are two sets of different values for 
$v_i=(a,b,b,b,b,b,b,b,b)$, and six sets of different values for 
$v_k=(s,t,u,u,v,v,v,v,v)$. Thus there will be $2\cdot 6=12$ different  
$\tilde{M}'_{\{i,k\}}$ matrices to check for positive semi-definiteness. Next pick a different pair of $(w_i,w_k)$ and repeat this process.

We obtain that the matrices  \eqref{eq:Malternative} are positive
semi-definite for any choice of two K\"ahler cone generators in
Table \ref{tab:KCdPn} and $x_{\{i,k\}}$ of the form
\beq \label{eq:xik}
	x_{\{i,k\}}=\frac{1}{a}\,\qquad\text{for} \qquad a\in\{1,2,\ldots,19\}\,.
\eeq
More precisely, for $dP_2$ and $dP_3$ all $x_{\{i,k\}}=1$, for $dP_4$ and $dP_5$ we have  $x_{\{i,k\}}=1,\,\frac{1}{2}$, for 
$dP_6$ we have $x_{\{i,k\}}=\frac{1}{a}$ with $a\in\{1,2,\ldots,4\}$, for $dP_7$ we find $x_{\{i,k\}}=\frac{1}{a}$ with $a\in\{1,2,\ldots,7\}$ and for $dP_8$ all values in \eqref{eq:xik} are assumed.

\section{Geometric Data of almost Fano Twofolds for computing Explicit Bounds}
\label{app:explicitdata}

In this appendix, we summarize the geometric data of Hirzebruch
surfaces $\mathbb{F}_k$, $k=0,1,2$, the del Pezzo surfaces 
$dP_n$, $n=2,\ldots, 8$, and the toric varieties associated to the 16 reflexive 
polytopes that is necessary to explicitly compute the various
bounds derived in this paper.

We begin with the bases $\mathbb{F}_k$ and $dP_n$.
The following results in Table \ref{tab:cover12} are 
derived employing \eqref{eq:T_iEvaluatedFk}, \eqref{eq:T_iEvaluateddPn}, the two 
covers of the K\"ahler cones of $dP_n$ constructed in Appendix 
\ref{app:MoriKaehlerFano}, \eqref{eq:Ac1} and the values of $x_{\{i,k\}}$ listed
below \eqref{eq:xik}. 

First, we list the maximal and minimal values of 
$x_{\{i,k\}}$ and $T_i$ for the bases $\mathbb{F}_k$ and $dP_2$ that have
a simplicial K\"ahler cone. For the non-simplicial K\"ahler cones, we obtain
different results for the two different covers of their K\"ahler cones. We note
that for both cover 1 and 2 the values below \eqref{eq:xik} apply. Indeed, this 
is precisely what we get in the second and third column under cover 2. However, 
for cover 1, these numbers have to be multiplied by appropriate $A_n$ in 
\eqref{eq:Ac1}. Indeed, by \eqref{eq:xiK} we have $x_{\{i,K\}}=A_n(\langle x_{\{i,k\}}\rangle_j)^{-1}$. By 
\eqref{eq:Ac1}, we have $A_n\leq 1$ for $n\leq 6$, i.e.~the minimum value of 
$x_{\{i,K\}}$ is bounded by $A_n\cdot\text{min}(x_{\{i,K\}})$, but the maximum 
is unaffected, as indicated in the first four rows of the second and third 
column in Table \ref{tab:cover12} under cover 1. For $dP_7$ and $dP_8$, 
we have  $A_n>1$, thus $x_{\{i,K\}}\leq A_n\text{max}(x_{\{i,k\}})=A_n$ and 
the minimum is unaffected, as displayed in the last two rows of the second and 
third column in Table \ref{tab:cover12} for cover 1. 

\begin{table}[H]
\begin{center}
\begin{tabular}{|c|c|c|c|c|}
\hline 
& $\text{max}(x_{\{i,k\}})$ & $\text{min}(x_{\{i,k\}})$ & $\text{max}(T_i)$\rule{0pt}{12pt} &
 $\text{min}(T_{i})$\\
 \hline
 $\mathbb{F}_k$ &  $1$&  $1$ & $24+12k$ & $24$\rule{0pt}{12pt} \\
 $dP_2$ &  $1$&  $1$ & $36$ & $24$ \\
 \hline
 \hline
 &\multicolumn{4}{c|}{Cover 1 of K\"ahler cone of $dP_n$\rule{0pt}{12pt}} \\
 \hline
$dP_3$ &  $1$& $A_3\leq$ & $72$ & $24$, $36$\rule{0pt}{12pt} \\
$dP_4$ & $1$& $\frac{1}{2} A_4\leq$ & $60$ & $24$, $36$\rule{0pt}{12pt} \\
$dP_5$ & $1$& $\frac{1}{2} A_5\leq$ & $48$ & $24$, $36$\rule{0pt}{12pt} \\
$dP_6$ & $1$& $\frac{1}{4}A_6\leq$ & $36$ & $24$, $36$\rule{0pt}{12pt} \\
$dP_7$ & $A_7\geq$& $\frac{1}{7}$ & $36$ & $24$\rule{0pt}{12pt}\\
$dP_8$ & $A_8\geq$& $\frac{1}{19}$ & $36$ & $12$\rule{0pt}{12pt}\\ 
\hline
\hline
&\multicolumn{4}{c|}{Cover 2 of K\"ahler cone of $dP_n$\rule{0pt}{12pt}} \\
 \hline
$dP_3$ &  $1$ & $1$ & $36$ & $24$, $36$\rule{0pt}{12pt}\\
$dP_4$ &  $1$ & $\frac{1}{2}$ & $36$ & $24$, $36$\rule{0pt}{12pt}\\
$dP_5$ &  $1$ & $\frac{1}{2}$ & $36$ & $24$, $36$\rule{0pt}{12pt}\\
$dP_6$ &  $1$ & $\frac{1}{4}$ & $36$ & $24$, $36$\rule{0pt}{12pt}\\
$dP_7$ &  $1$ & $\frac{1}{7}$ & $36$ & $24$, $36$\rule{0pt}{12pt}\\
$dP_8$ &  $1$ & $\frac{1}{19}$ & $36$ & $24$, $36$\rule{0pt}{12pt}\\ \hline
\end{tabular}
\caption{Key geometrical data for the computation of the explicit bounds
derived in the proof.}
\label{tab:cover12}
\end{center}
\end{table}

In addition, without 
knowing every simplicial subcone in the two covers explicitly, we can not
determine the explicit value $\text{min}(T_i)$ for both covers. Therefore,
depending on the chosen subcone, employing \eqref{eq:T_iEvaluateddPn}, we either
obtain $24$ or $36$ as indicated in the last column of Table \ref{tab:cover12}.
However, in the case of cover 1 we have $T_K=24$, $12$ for $dP_7$ and $dP_8$, 
respectively. Since by construction, the first Chern class $c_1(dP_n)$ is
in every subcone, we know that $\text{min}(T_i)=T_K=24$, $12$ for $dP_7$ and 
$dP_8$, respectively.

Finally, in Table \ref{table:toric surface} we display
the relevant topological data of the toric varieties constructed from the 16 
reflexive two-dimensional polytopes which is relevant to
our finiteness proof in section \ref{sec:generic proof}. We confirmed that the first Chern class 
$c_1(B)$ is inside the K\"ahler cone in all these cases, i.e.~Cover 1 constructed in Appendix \ref{app:MoriKaehlerFano} 
exists for these  non-simplicial K\"ahler cones. As explained there, in this cover the conditions (2) and (3) listed at the beginning of Section \ref{sec:generic proof}
are obeyed. We also checked that the matrices \eqref{eq:Malternative} are all positive semi-definite for $x_{\{i,k\}}$ of the form $x_{\{i,k\}}=\frac{1}{a}$ with $a\in\{1,\ldots,6\}$, i.e.~condition (1) listed in Section \ref{sec:generic proof} is also satisfied. 

\begin{table}
\centering
\scalebox{.8}{\begin{tabular}{c|c|c|c|c}
Polytope& $\int c_2$ & $\int c_1^2$ & $|\text{K.C. Gens}|$ & List of $T_{i}=12\,\int_{D_i} c_1$ \\ \hline 
$ 2 $ & $ 4 $ & $ 8 $ & $ 2 $ & $ (24, 24) $ \\
$ 3 $ & $ 4 $ & $ 8 $ & $ 2 $ & $ (24, 36) $ \\
$ 4 $ & $ 4 $ & $ 8 $ & $ 2 $ & $ (24, 48) $ \\
$ 5 $ & $ 5 $ & $ 7 $ & $ 3 $ & $ (24, 24, 36) $ \\
$ 6 $ & $ 5 $ & $ 7 $ & $ 3 $ & $ (24, 36, 48) $ \\
$ 7 $ & $ 6 $ & $ 6 $ & $ 5 $ & $ (24, 24, 24, 36, 36) $ \\
$ 8 $ & $ 6 $ & $ 6 $ & $ 4 $ & $ (24, 36, 24, 48) $ \\
$ 9 $ & $ 6 $ & $ 6 $ & $ 5 $ & $ (24, 36, 24, 48, 36) $ \\
$ 10 $ & $ 6 $ & $ 6 $ & $ 4 $ & $ (24, 48, 72, 36) $ \\
$ 11 $ & $ 7 $ & $ 5 $ & $ 7 $ & $ (24, 36, 48, 24, 36, 72, 48) $ \\
$ 12 $ & $ 7 $ & $ 5 $ & $ 8 $ & $ (24, 24, 36, 36, 48, 48, 24, 36) $ \\
$ 13 $ & $ 8 $ & $ 4 $ & $ 10 $ & $ (24, 48, 36, 72, 48, 36, 24, 72, 48, 48) $ \\
$ 14 $ & $ 8 $ & $ 4 $ & $ 13 $ & $ (24, 24, 36, 48, 36, 48, 24, 72, 36, 48, 72, 48, 36) $ \\
$ 15 $ & $ 8 $ & $ 4 $ & $ 12 $ & $ (24, 36, 24, 48, 36, 48, 36, 48, 48, 24, 24, 36) $ \\
$ 16 $ & $ 9 $ & $ 3 $ & $ 21 $ & $ (24, 24, 36, 72, 48, 36, 48, 36, 48, 36, 48, 72, 72, 72, 36, 48, 24, 36, 72, 48, 72) $ \\
\end{tabular}}
\caption{Displayed are some of the relevant data for the smooth almost Fano 
toric surfaces obtained from fine star triangulations of the two-dimensional 
reflexive polytopes in Figure \ref{fig:2dpoly}. }
\label{table:toric surface}
\end{table}

\section{An analytic proof of positive semi-definiteness of the $M_{\{i,k\}}$-Matrices}
\label{app:CveticTheorem}

In this section we provide an alternative general proof of positive 
semi-definiteness of the $M_{\{i,k\}}$-matrices, in comparison to 
the numerical proof given in Appendix \ref{app:MoriKaehlerFano} for 
specific $B$. 

We recall that to check 
positive-semi-definiteness for any $M_{\{i,k\}}$ defined in 
\eqref{eq:Malternative} it suffices to check positive 
semi-definiteness for the matrix $\tilde{M}_{\{i,k\}}$ defined in 
\eqref{eq:tildeM}. The advantage of the following general proof is 
that it 
predicts a precise value of $x_{\{i,k\}}$ for which each 
$M_{\{i,k\}}$ is positive semi-definite. Thus, we do not have to 
search for the existence of such an $x_{\{i,k\}}$ numerically. To be 
precise, we will show that we can always choose
 \begin{equation} \label{eq:xikStressEnergy}
x_{\{i,k\}}=\frac{1}{C_{ik}}
\end{equation}
to make each $M_{\{i,k\}}$ positive semi-definite. We note, however, 
that  such a choice may not produce the best bounds (since the 
various bounds derived depend on $x_{\{i,k\}}$). Hence, in order to
minimize the various bounds we may still want to numerically  
find alternative values for $x_{\{i,k\}}$, for which the matrices 
\eqref{eq:Malternative}, \eqref{eq:tildeM} are  also
positive semi-definite.

The correctness of the value \eqref{eq:xikStressEnergy} can
be motivated physically as follows. Consider a system of
two particles with masses $m=1$ with the Lorentz-invariant 
Lagrangian
\beq\label{lagr}
L^{i,k}=p^i\cdot p^{k}\,,\qquad i\neq k\,,
\eeq
where $p^i$ for every $i,k=1,\ldots, N$ are the particle momenta.
Due to space-time invariance, the  respective Noether currents  are 
stress-energy tensors,
\beq
T^{i,k}_{\mu,\nu}=p_\mu^i\, p_\nu^{k}+p_\mu^{k}\, p_\nu^i- L^{i,k}\eta_{\mu\nu}    \,,\qquad \ i\ne k\, \,.
\label{stress}
\eeq
With the identification $\frac{1}{x^{i,k}}\equiv L^{i,k}$, these  stress-energy tensors are precisely  the matrices
\eqref{eq:tildeM} multiplied by $\frac{1}{x^{i,k}}$. By the positive 
energy theorem in
general relativity the $T^{i,k}_{\mu,\nu}$ are 
positive semi-definite for every chosen pair of
time- or light-like $(n+1)$-vectors $p^i$, $p^k$.

In the following, we prove explicitly that the the matrices in
\eqref{eq:tildeM}, i.e.~the stress energy tensors 
\eqref{stress}, are indeed positive semi-definite for 
time- or light-like $(n+1)$-vectors $p^i$, $p^k$.
To this end, we will need the following general fact:
\begin{lemma}\label{lem:M+A^TMA+}
For any $n\times n$ matrix $M$ and any invertible $n\times n$ matrix $A$, $M$ is positive semi-definite if and only if $A^TMA$ is positive semi-definite.
\end{lemma}

Using of Lemma \ref{lem:M+A^TMA+}, we can prove positive semi-definiteness of $\tilde{M}_{\{i,k\}}$ by instead proving positive semi-definiteness of $A^T\tilde{M}_{\{i,k\}}A$, where $A$ is a suitably chosen invertible matrix so that $A^T\tilde{M}_{\{i,k\}}A$ takes a simpler form than $\tilde{M}_{\{i,k\}}$. We will discuss how to choose $A$ shortly. First, recall from Table \ref{tab:KCdPn} that each K\"ahler cone generator $v_i$ is either time-like or light-like with Minkowski inner product $\eta(v_i,v_i)$ either 1 or 0, and all the K\"ahler cone generators belong to the same light cone (the future-directed light cone). We choose $A$ as follows:\\\

\noindent \textbf{Case 1.} Suppose $\tilde{M}_{\{i,k\}}$, defined in 
\eqref{eq:tildeM}, has at least one of its $v_i, v_k$ 
with Minkowski inner product 1. WLOG say $\eta(v_i,v_i)=1$. Then 
there is a matrix $A\in O(1,n)$ such that 
\begin{equation}\label{eq:ADefinition}
A^Tv_i=(1,0,...,0)^T\,.
\end{equation}
We note that this is just a Lorentz transformation to the rest 
frame.
Pick this matrix as the invertible matrix $A$ in Lemma \ref{lem:M+A^TMA+}.\\\

\noindent \textbf{Case 2.} Suppose $\tilde{M}_{\{i,k\}}$, defined in 
\eqref{eq:tildeM}, has both of its $v_i, v_k$ with 
Minkowski inner product 0. Then there exists a Lorentz 
transformation $A'\in O(1,n)$ such that 

\begin{equation}\label{eq:A'Definition}
A^{'T}v_i=(a_0, a_0, 0,...,0)^T, \qquad v_k^TA'=(b_0, b_1, b_2,0,...,0)\,,
\end{equation}
where $a_0, b_0, b_1, b_2 \in \mathbb{Q}$ and $b_0^2-b_1^2-b_2^2=0.$ Pick $A'$ as the invertible matrix in Lemma \ref{lem:M+A^TMA+}.\\\

The above mentioned matrices in $O(1,n)$ exist because of the following general lemma:

\begin{lemma}\label{lem:ExistenceOfAinLorentzGroup}
For any vector $v \in \mathbb{R}^{1,n}$ which Minkowski inner product $\eta(v,v)=1$, there exists a matrix $A\in O(1,n)$ such that $A^Tv=(1,0,...,0)^T$. For any pair of vector $v_i, v_k \in \mathbb{R}^{1,n}$ both with Minkowski inner product $\eta(v_i,v_i)=\eta(v_k,v_k)=0$, there exists a matrix $A'\in O(1,n)$ such that $A^{'T}v_i=(a_0, a_0, 0,...,0)^T$, $v_k^TA'=(b_0, b_1, b_2,0,...,0)$ where $a_0, b_0, b_1, b_2 \in \mathbb{R}$ and $b_0^2-b_1^2-b_2^2=0.$
\end{lemma}

\begin{proof}
First consider any $v \in \mathbb{R}^{1,n}$ with Minkowski inner product $\eta(v,v)=1$. Since $\eta(v,v)=1\neq 0$, we can carry out the Gram-Schmidt process starting with $v$ as the first vector to generate an orthonormal basis $\{e_1=v, e_2,...,e_{n+1}\}$ for $\mathbb{R}^{1,n}$. Define the $(1+n)\times (1+n)$ matrix $B$ whose $i$-th column is $e_i$, and define $A=\eta B$. Then $A^Tv=(1,0,...,0)^T$ by orthonormality. Both $B$ and $\eta$ are in $O(1,n)$ because each has its columns orthonormal to one another under the $(1,n)$ Minkowski metric. Thus $A=\eta B\in O(1,n)$.

Next consider any pair of vector $v_i, v_k \in \mathbb{R}^{1,n}$, both with Minkowski inner product $\eta(v_i,v_i)=\eta(v_k,v_k)=0$. If both are equal to the trivial vector $(0,...,0)^T$, let $A'$ be any matrix in $O(1,n)$ and we are done with $a_0=b_0=b_1=b_2=0$. Thus assume at least one of them, WLOG say $v_i$, is not the trivial vector. Let $v_i=(a_0,\mathbf{a})^T$ where $\mathbf{a}=(a_1,...,a_n)^T\in \mathbb{R}^{n}$. Since $\eta(v_i,v_i)=0 $ and $v_i$ is not the trivial vector, the Euclidean norm of $\mathbf{a}$, $|\mathbf{a}|=a_0\neq 0$ ($a_0$ is positive because $v_i$ is in the positive light cone). We can thus use $\mathbf{a}/|\mathbf{a}|$ as the first vector in the Gram-Schmidt process on $\mathbb{R}^{n}$ to generate an orthonormal basis $\{e_1=\mathbf{a}/|\mathbf{a}|, e_2,...,e_n\}$ for $\mathbb{R}^{n}$. Define the $n\times n$ matrix $B'$ whose $i$-th column is $e_i$. Then define the $(1+n)\times (1+n)$ block diagonal matrix $B^{''}$ by 

\begin{equation}
B^{''}=\begin{pmatrix}
	1 & 0 \\
	0 & B'
\end{pmatrix}\,.
\end{equation}
$B^{''} \in O(1,n)$ because its columns are orthonormal. Also $B^{''T}v_i=(a_0, a_0, 0,...,0)^T$. Let $v_k^TB^{''}=(b_0,b_1,\mathbf{b'})$ where $\mathbf{b'}=(b'_2,...,b'_n) \in \mathbb{R}^{n-1}$. If $\mathbf{b'}$ is the trivial vector in $\mathbb{R}^{n-1}$, we are done by setting $A'=B^{''}$ and $b_2=0$. If $\mathbf{b'}$ is not the trivial vector, we can again use $\mathbf{b'}/|\mathbf{b'}|$ as the first vector in the Gram-Schmidt process on $\mathbb{R}^{n-1}$ to generate an orthonormal basis $\{e_1=\mathbf{b'}/|\mathbf{b'}|, e_2,...,e_{n-1}\}$ for $\mathbb{R}^{n-1}$. Define the $(n-1)\times (n-1)$ matrix $C'$ whose $i$-th column is $e_i$. Then define the $(1+n)\times (1+n)$ block diagonal matrix  $C^{''}$ by 

\begin{equation}
C^{''}=\begin{pmatrix}
	1 & 0 & 0 \\
	0 & 1 & 0 \\
	0 & 0 & C'
\end{pmatrix}\,.
\end{equation}
$C^{''} \in O(1,n)$ because its columns are orthonormal. Let $A'=B^{''}C^{''}$. $A' \in O(1,n)$ because $B^{''},C^{''}$ are. We also have $A^{'T}v_i=(a_0, a_0, 0,...,0)^T$, $v_k^TA'=(b_0, b_1, b_2,0,...,0)$ where $b_2=|\mathbf{b'}|$. Notice that $b_0^2-b_1^2-b_2^2=\eta(A^{'T}v_k,A^{'T}v_k)=\eta(v_k,v_k)=0$, where in the second equality we used the facts that $O(1,n)$ is closed under transposition, so $A^{'T}\in O(1,n)$, and that the Lorentz group $O(1,n)$ preserves $\eta(\cdot , \cdot)$.
\end{proof}

Before justifying the choice $x_{\{i,k\}}=1/C_{ik}$, we need to show $C_{ik}\neq 0$ for $i\neq k$ (by definition we always have $i\neq k$ in $x_{\{i,k\}}$ and $M_{\{i,k\}}$). Also recall that in \eqref{M}, we require $x_{\{i,k\}}\in \mathbb{Q}^+$. Thus a prerequisite for the choice $x_{\{i,k\}}=1/C_{ik}$ to make sense is that $C_{ik}>0$ for $i\neq k$ ($C_{ik}$ is already an integer since it is an intersection number). We have the following lemma:

\begin{lemma}
$C_{ik}\geq 0$. Furthermore, $C_{ik}> 0$ if $i\neq k$; $C_{ii}= 0$ if and only if $v_i$ is lightlike; i.e. $\eta(v_i,v_i)=0$.
\end{lemma}

\begin{proof}
Recall we have 
\begin{equation}\label{eq:CikDefinition}
C_{ik}=\eta(v_i,v_k)\,.
\end{equation}
Also, by Table \ref{tab:KCdPn}, all the K\"ahler cone generators 
$v_i,v_k$ are either time-like or light-like vectors belonging to 
the same light cone. In addition, of course neither of them is the 
trivial vector $\mathbf{0}$, because they generate the K\"ahler 
cone. This means all their inner products are non-negative, i.e. 
$C_{ik}=\eta(v_i,v_k) \geq 0$, where equality $C_{ik}=\eta(v_i,v_k) = 0$ holds only when $v_i$ and $v_k$ are parallel light-like vectors. This implies that $v_i$ and $v_k$ are not independent, so they must be the same K\"ahler cone generator $v_i=v_k$.
\end{proof}

With this, we can prove the following theorem:

\begin{theorem}
Let $x_{\{i,k\}}=1/C_{ik}$. Then $M_{\{i,k\}}$ is positive semi-definite.
\end{theorem}

\begin{proof}
It is equivalent to prove that with $x_{\{i,k\}}=1/C_{ik}$, 
$A^T\tilde{M}_{\{i,k\}}A$ or $A^{'T}\tilde{M}_{\{i,k\}}A'$, 
depending on which case above we are refering to is positive 
semi-definite, where $A$ (or $A'$) is the appropriately chosen 
matrix in $O(1,n)$ discussed above.\\\

\noindent \textbf{Case 1.} Suppose $\tilde{M}_{\{i,k\}}$, defined in \eqref{eq:tildeM}, has at least one of its $v_i, v_k$ with Minkowski inner product 1. WLOG say $\eta(v_i,v_i)=1$. Then
\bea
A^T\tilde{M}_{\{i,k\}}A &= & x_{\{i,k\}}A^T(v_i\cdot v_k^T+v_k\cdot v_i^T)A-\eta \nonumber\\
&= & \frac{1}{c_0} \left((1,0,...,0)^T(c_0,c_1,...,c_n) + (c_0,c_1,...,c_n)^T(1,0,...,0) \right)-\eta\nonumber\\
&= & 
\begin{pmatrix}
	1 & \frac{c_1}{c_0} & \frac{c_2}{c_0} & \frac{c_3}{c_0} & ... & \frac{c_n}{c_0} \\
	\frac{c_1}{c_0} & 1 & 0 & 0 & ... & 0\\
	\frac{c_2}{c_0} & 0 & 1 & 0 & ... & 0\\
	\ . & . & . & . & . & .\\
	\frac{c_n}{c_0} & 0 & 0 & . & . & 1\\
\end{pmatrix}\,,
\eea
where in the first equality, we used the fact that $A\in O(1,n)$ if and only if $A^T\eta A =\eta$. In the second equality, we used \eqref{eq:ADefinition} and let $A^Tv_k=(c_0,c_1,...,c_n)^T$, so $C_{ik}=\eta(v_i, v_k)= \eta(A^Tv_i, A^Tv_k)=\eta((1,0,...,0)^T, (c_0,c_1,...,c_n)^T)=c_0$ (notice that $O(1,n)$ is closed under transposition, so $A^T\in O(1,n)$ and thus $A^T$ preserves the inner product $\eta(\cdot , \cdot)$). It is not hard to see that the characteristic equation of $A^T\tilde{M}_{\{i,k\}}A$ is 
\begin{equation}
\text{det}\left(A^T\tilde{M}_{\{i,k\}}A - \lambda I \right)=(1-\lambda)^{n-1}\left(\lambda^2-2\lambda+1-\frac{1}{c_0^2}\sum_{j=1}^n c_j^2\right)=0\,,
\end{equation}
so the eigenvalues are 
\begin{equation}\label{eq:Case1Eigenvalues}
\lambda=\left\{\underbrace{1,...1}_{n-1},\left(1\pm \sqrt{\sum_{j=1}^n \frac{c_j^2}{c_0^2}}\right)\right\}\,.
\end{equation}
Since 
\begin{equation}
0\leq 0\text{ or } 1= \eta(v_k,v_k)=\eta(A^Tv_k,A^Tv_k)=\eta((c_0,c_1,...,c_n)^T,(c_0,c_1,...,c_n)^T)=c_0^2-\sum_{j=1}^nc_j^2\,,
\end{equation}
we must have 
\begin{equation}
1\geq \sqrt{\sum_{j=1}^n \frac{c_j^2}{c_0^2}}\,,
\end{equation}
so all the eigenvalues in \eqref{eq:Case1Eigenvalues} are 
non-negative. In particular, if $\eta(v_k,v_k)=c_0^2-\sum_{j=1}^nc_j^2=1$, the eigenvalues will be
\begin{equation}\label{eq:Case1EigenvaluesVkTimelike}
\lambda=\left\{\underbrace{1,...1}_{n-1},\left(1\pm \sqrt{1-\frac{1}{c_0^2}}\right)\right\}\,.
\end{equation}
If $\eta(v_k,v_k)=c_0^2-\sum_{j=1}^nc_j^2=0$, the eigenvalues will be
\begin{equation}\label{eq:Case1EigenvaluesVkLightlike}
\lambda=\left\{\underbrace{1,...1}_{n-1},0,2\right\}\,.
\end{equation}

\noindent \textbf{Case 2.} Suppose $\tilde{M}_{\{i,k\}}$, defined in \eqref{eq:tildeM}, has both of its $v_i, v_k$ with Minkowski inner product 0. Then 
\bea
 &&A^{'T}\tilde{M}_{\{i,k\}}A' = x_{\{i,k\}}A^{'T}(v_i\cdot v_k^T+v_k\cdot v_i^T)A'-\eta \nonumber\\
 &&= \frac{1}{b_0-b_1} \left((1,1,0,...,0)^T(b_0,b_1,b_2,0...,0) + (b_0,b_1,b_2,0,...,0)^T(1,1,0,...,0) \right)-\eta\nonumber\\
&&=  
\begin{pmatrix}
\frac{b_0+b_1}{b_0-b_1} & \frac{b_0+b_1}{b_0-b_1} & \frac{b_2}{b_0-b_1} &0 & ... & 0 \\
\frac{b_0+b_1}{b_0-b_1} & \frac{b_0+b_1}{b_0-b_1} & \frac{b_2}{b_0-b_1} &0 & ... & 0 \\
\frac{b_2}{b_0-b_1} & \frac{b_2}{b_0-b_1} & 1 & 0 & ... & 0\\
\ 0 & . & 0 & 1 & . & .\\
\ . & . & . & . & . & .\\
\ 0 & . & 0 & . & . & 1\\
\end{pmatrix}\,,
\eea
where in the second equality we used \eqref{eq:A'Definition} and $C_{ik}=\eta(v_i, v_k)= \eta(A^{'T}v_i, A^{'T}v_k)=\eta((a_0,a_0,0,...,0)^T, (b_0,b_1,b_2,0,...,0)^T)=a_0(b_0-b_1)$. Letting 
\begin{equation}
s\equiv \frac{b_0+b_1}{b_0-b_1} \qquad \sqrt{s}=\frac{b_2}{b_0-b_1}\,,
\end{equation}
where in the second equation we used the relationship $b_0^2-b_1^2-b_2^2=0$, it is not hard to see that the characteristic equation of $A^{'T}\tilde{M}_{\{i,k\}}A'$ is 
\begin{equation}
\text{det}\left(A^{'T}\tilde{M}_{\{i,k\}}A' - \lambda I \right)=(1-\lambda)^{n-2}\lambda^2(2s+1-\lambda)=0\,,
\end{equation}
so the eigenvalues are 
\begin{equation}\label{eq:Case2Eigenvalues}
\lambda=\left\{\underbrace{1,...1}_{n-2},0,0, \frac{3b_0+b_1}{b_0-b_1}\right\}\,.
\end{equation}
The last eigenvalue $\frac{3b_0+b_1}{b_0-b_1}$ is positive because $b_0^2-b_1^2-b_2^2=0$, so $|b_0|>|b_1|$.
\end{proof}

Notice that the only required condition for this general proof is 
that all the K\"ahler cone generators $v_i,v_k$ are either time-like 
or light-like, and belong to the same light cone. This light cone 
does not need to be the positive one. Indeed, it is not hard to see 
that if all the K\"ahler cone generators were to belong to the 
negative light cone, the proof still holds with slight modifications 
at the relevant parts. Also, the time-like K\"ahler cone generators 
can always be rescaled to have Minkowski inner product 
$\eta(v_i,v_i)=1$. In summary, we have the following corollary:

\begin{corollary}
If all the K\"ahler cone generators $v_i,v_k$ are either time-like or light-like, and belong to the same light cone, then each matrix $M_{\{i,k\}}$ will be positive semi-definite by setting $x_{\{i,k\}}=1/C_{ik}$.
\end{corollary}

\bibliographystyle{JHEP}
\bibliography{refs}

\end{document}